\documentclass[12pt,reqno]{article}
\usepackage[letterpaper,margin=1.6cm]{geometry}
\usepackage{latexsym, amsfonts, amsthm,amssymb,amscd,amsmath,makeidx}
\usepackage{enumerate}
\usepackage[normalem]{ulem}
\usepackage{exscale}
\usepackage{overpic} 
\usepackage{color} 
\usepackage{graphicx}
\usepackage{overpic}  
\usepackage{bm}
\usepackage{comment}
\usepackage{caption}  
\usepackage{float}
\usepackage{array} 
\usepackage{booktabs} 
\usepackage{tabularx}

\usepackage{subcaption}
\usepackage{mathrsfs} 
\usepackage{amssymb}
\usepackage{mathtools}

\DeclarePairedDelimiter\floor{\lfloor}{\rfloor}

\usepackage{tikz}
\usepackage{blindtext}
\usepackage{bbm}
\usepackage{fancyhdr} 
\usepackage{mathrsfs}
\usepackage{wrapfig}
\usepackage{hyperref}
\usepackage{bookmark}
\usepackage{cleveref}
\usepackage{cases}
\usepackage{multirow}
\usepackage{pdflscape}
\usepackage[title]{appendix}

\hypersetup{
	colorlinks,
	citecolor=red,
	filecolor=black,
	linkcolor=blue,
	urlcolor=black
}

\definecolor{DarkGreen}{rgb}{0.2,0.6,0.2}

\def\eps{\varepsilon}
\def\Om{\Omega}\def\om{\omega}
\def\trace#1{\text{trace}\, #1}

\newcommand{\norm}[1]{\left\lVert {#1}\right\rVert}
\newcolumntype{Y}{>{\centering\arraybackslash}X}

\def\lra{\longrightarrow}

\def\ua{\uparrow}
\def\da{\downarrow}

\def\wh{\widehat}
\def\wt{\widetilde}

\newcommand{\ra}[1]{\renewcommand{\arraystretch}{#1}}
\newcolumntype{C}{>{\centering\arraybackslash}X}

\def\ignore#1{}

\newcommand\scalemath[2]{\scalebox{#1}{\mbox{\ensuremath{\displaystyle #2}}}}

\numberwithin{equation}{section}

\usetikzlibrary{positioning,calc,cd}

\allowdisplaybreaks[4]

\def\cov{\mbox{cov}}
\def\var{\mbox{var}}

\def\cF{{\mathscr F}}
\def\cG{{\mathscr G}}

\def\cR{{\mathscr R}}

\newtheorem{theorem}{Theorem}[section]
\newtheorem{proposition}[theorem]{Proposition}
\newtheorem{lemma}[theorem]{Lemma}
\newtheorem{corollary}[theorem]{Corollary}

\theoremstyle{definition}
\newtheorem{definition}[theorem]{Definition}
\newtheorem{example}[theorem]{Example}

\newtheorem{remark}[theorem]{Remark}

\def\lra{\longrightarrow}
\def\eps{\varepsilon}
\def\<{\langle}
\def\>{\rangle}
\def\wt#1{\widetilde{#1}}

\def\lra{\longrightarrow}

\def\ua{\uparrow}
\def\da{\downarrow}

\def\wh{\widehat}
\def\wt{\widetilde}

\def\cov{\hbox{\rm cov}}\def\var{\text{\rm Var}}

\def\argmin{\mathop{\hbox{\rm arg\,min}}}

\newcommand{\indep}{\perp \!\!\! \perp}

\newcommand{\calN}{\mathcal{N}}

\newcommand{\calS}{\mathcal{S}}

\newcommand{\calU}{\mathcal{U}}

\newcommand{\bE}{\mathbb{E}}
\newcommand{\bR}{\mathbb{R}}

\newcommand{\bN}{\mathbb{N}}
\newcommand{\bP}{\mathbb{P}}

\newcommand{\bT}{\mathbb{T}}
\newcommand{\bZ}{\mathbb{Z}}
\newcommand{\bQ}{\mathbb{Q}}

\begin{document}
	
	\title{Estimating the roughness exponent of stochastic volatility from discrete observations of the integrated variance}
	\author{Xiyue Han$^*$ and Alexander Schied\thanks{
			Department of Statistics and Actuarial Science, University of Waterloo, 200 University Ave W, Waterloo, Ontario, N2L 3G1, Canada. E-Mails: {\tt xiyue.han@outlook.com, aschied@uwaterloo.ca}.\hfill\break
			}}
	\date{\normalsize First version: June 25, 2023\\
	\normalsize This version: April 15, 2026}
	\maketitle

\begin{abstract} We consider the problem of estimating the roughness of the volatility process in a stochastic volatility model that arises as a nonlinear function of fractional Brownian motion with drift. To this end, we introduce a new estimator that measures the so-called roughness exponent of a continuous trajectory, based on discrete observations of its antiderivative. The estimator has a very simple form and can be computed with great efficiency on large data sets. It is not derived from distributional assumptions but from strictly pathwise considerations.  We provide conditions on the underlying trajectory under which our estimator converges in a strictly pathwise sense. Then we verify that these conditions are satisfied by almost every sample path of fractional Brownian motion (with drift). As a consequence, we obtain strong consistency theorems in the context of a large class of rough volatility models, such as the rough fractional volatility model and the rough Bergomi model. We also demonstrate that our estimator is robust with respect to proxy errors between the integrated and realized variance, and that it can be applied to estimate the roughness exponent directly from the price trajectory. Numerical simulations show that our estimation procedure performs well after passing to a scale-invariant modification of our estimator. 
\end{abstract}

\smallskip
\noindent
\textbf{MSC2020 subject classifications:} 91G70, 62P05, 60F15, 60G22

\smallskip
\noindent
\textbf{Keywords:} Rough volatility, estimation of the roughness exponent, rough fractional stochastic volatility model, rough Bergomi model, strong consistency

\section{Introduction}

Consider a stochastic volatility model whose price process satisfies
\begin{equation}\label{price process eq}
dS_t=\sigma_tS_t\,dB_t,\qquad S_0=s_0>0,
\end{equation}
where $B$ is a standard Brownian motion and $\sigma_t$ is a continuous and adapted stochastic process. Since the publication of the seminal paper \cite{GatheralRosenbaum} by Gatheral, Jaisson, and Rosenbaum, it has been widely accepted that the sample paths of $\sigma_t$ often do not exhibit diffusive behavior but instead are much rougher. A specific example suggested in  \cite{GatheralRosenbaum} is to model the log volatility by a fractional Ornstein--Uhlenbeck process. That is, 
\begin{equation}\label{log volatility eq}
\sigma_t=\exp(X^H_t),
\end{equation}
where $X^H$ solves the following integral equation
\begin{equation}\label{OU process eq}
X^H_t=x_0+\rho\int_0^t(\mu-X^H_s)\,ds+W^H_t,\qquad t\ge0,
\end{equation}
for a fractional Brownian motion $W^H$ with Hurst parameter $H\in(0,1)$. In this model, the \lq roughness' of the trajectories of $X^H$ is governed by the Hurst parameter $H$, and it was pointed 
 out 
 in \cite{GatheralRosenbaum} that rather small values of $H$ appear to be most adequate for capturing the stylized facts of empirical volatility time series. Such a model is often referred to as the rough fractional stochastic volatility model.\footnote{In this paper, we include the case $H \ge 1/2$, in which the volatility process $\sigma$ would not be considered \lq rough'.} Since the publication of \cite{GatheralRosenbaum}, many alternative rough volatility models have been proposed, e.g., the rough Heston model \cite{ EuchRosenbaumGatheral2019roughening, EuchRosenbaum2019Characteristic, EuchRosenbaumRoughHeston18} and the rough Bergomi model \cite{BayerGatheralFriz16, Forde2022, Jacquier18RoughBergomi}.

 The present paper contributes to the literature on rough volatility by considering the statistical estimation of the degree of roughness of the volatility process $\sigma_t$. 
 There are several difficulties that arise in this context.

 The first difficulty consists in the fact that in reality the volatility process $\sigma_t$ cannot be observed directly; only the asset prices $S$ are known. Thus, one typically computes the quadratic variation of the log stock prices, 
 \begin{equation}\label{realized variance eq}
 \<\log S\>_t = \int_{0}^{t}\sigma^2_s \,ds,
 \end{equation}
which is also called the integrated variance, and then performs numerical differentiation to estimate proxies $\wh\sigma_t$ for the actual values of $\sigma_t$. The roughness estimation is then based on those proxy values  $\wh\sigma_t$. For instance, this two-step procedure underlies  the statistical analysis for empirical volatilities in \cite{GatheralRosenbaum}, where roughness estimates were based on empirical values  $\wh\sigma_t$  taken from the Oxford-Man Institute of Quantitative Finance Realized Library. A problem with that approach is that estimation errors in the empirical values $\wh\sigma_t$ might substantially distort the outcomes in the final roughness estimation;  see Fukasawa et al.~\cite{Fukasawa2022Estimation} and Cont and Das \cite{ContDasArtefact}.

As a matter of fact, the quadratic variation \eqref{realized variance eq} itself is usually approximated by a finite sum of the form $\sum_i(\log S_{t_i}-\log S_{t_{i-1}})^2$ based on discrete observations $ S_{t_i}$ of the price process. This finite sum can be regarded as the (cumulative) realized variance of the price trajectory $S$. The bias caused by the proxy error between the quadratic variation $ \<\log S\>_t$ and the finite sum $\sum_i(\log S_{t_i}-\log S_{t_{i-1}})^2$ is emphasized in \cite{Fukasawa2022Estimation}, where it is assumed that the proxy errors are log-normally distributed and independent of the Brownian motion $B$ in \eqref{price process eq}, and a Whittle-type estimator for the Hurst parameter is developed based on quasi-likelihood. Another attempt to tackle this measurement error is made by Bolko et al.~\cite{BolkoPakkanen2022GMM}, where in a similar framework, the proposed estimator is based on the generalized method of moments approach. Chong et al.~\cite{Chong2022CLT, Chong2022Minimax} substantially extend the previous results by alleviating the assumption on proxy errors and basing the volatility model on a semi-parametric setup, in which,   with the exception of the Hurst parameter of the underlying fractional Brownian motion, all components are fully non-parametric. One of the conclusions from \cite{BolkoPakkanen2022GMM, Fukasawa2022Estimation, Chong2022CLT} is that the error arising from approximating the quadratic variation \eqref{realized variance eq} with finite sum $\sum_i(\log S_{t_i}-\log S_{t_{i-1}})^2$ could lead to a biased estimation result unless properly controlled.

Here, we analyze new estimators for the roughness of the volatility process $\sigma$ that are based directly on discrete observations of the quadratic variation \eqref{realized variance eq} or on $\sum_i(\log S_{t_i}-\log S_{t_{i-1}})^2$. Our estimators take very simple forms and can be computed with great efficiency on large data sets. They are not derived from distributional assumptions, as most other estimators in the literature, but from strictly pathwise considerations whose development was initiated in \cite{HanSchiedMatrix, HanSchiedHurst}. Furthermore, our estimators do not actually measure the traditional Hurst parameter, which may be unrelated to the actual roughness of trajectories \cite{Gneiting2004Hurst}. Instead, our estimators measure the so-called roughness exponent, which was introduced in \cite{HanSchiedHurst} as the reciprocal of the critical exponent for the power variations of trajectories.
For fractional Brownian motion, this roughness exponent coincides with the Hurst parameter, but it can also be computed for many other trajectories, including certain fractal functions.

In \cite{HanSchiedHurst}, we establish conditions under which a given trajectory $x\in C[0,1]$ admits a roughness exponent $R$, and we provide several estimators that approximate $R$, based on the Faber--Schauder expansion of $x$. In \cite{HanSchiedMatrix}, we derive a robust method for estimating the Faber--Schauder coefficients of $x$ for the situation in which only the antiderivative $y(t)=\int_0^tx(s)\,ds$, and not $x$ itself, is observed on a discrete time grid. As explained in greater detail in \Cref{rationale section}, that method, when combined with the estimators from  \cite{HanSchiedHurst}, gives rise to the specific form of the estimators which we propose here.  

In \Cref{pathwise section}, we formulate conditions on the trajectory $x$ under which our basic estimator $\wh \cR_n(v)$ converges to the roughness exponent of $x$, resting on discrete observations of the function $v(t)=\int_0^tg(x(s))\,ds$, where $g$ is a suitable $C^2$-function. In \Cref{Section Proof fbm}, we then verify that the aforementioned conditions on the trajectory $x$ are satisfied by almost every sample path of fractional Brownian motion (with drift). This verification immediately yields the strong consistency of our estimator for the case in which the stochastic volatility is a nonlinear function of a fractional Brownian motion with drift as in \Cref{thm main rv}. In particular, this result applies to the rough fractional stochastic volatility model \eqref{log volatility eq} and \eqref{OU process eq}.  In \Cref{Section Proof Bergomi}, we decompose the Riemann--Liouville process into a fractional Brownian motion and an absolutely continuous Gaussian process. By means of a pathwise analysis we show  in \Cref{Thm Rough Bergomi} that $\wh \cR_n$ is a consistent estimator of the roughness exponent of the rough Bergomi model.

To incorporate the measurement errors between the integrated and realized variance into our estimation framework, we construct a new estimator $\wt \cR_n$ based on $\wh \cR_n$ in \Cref{RV section}. We then demonstrate the strong consistency of this estimator $\wt \cR_n$ in \Cref{thm main realized} under condition \eqref{RV condition eq}, which only specifies the required number of data points to compute realized variance so as to ensure the strong consistency of $\wt \cR_n$. Condition \eqref{RV condition eq} was previously proposed in \cite{Fukasawa2022Estimation} in a very similar form to ensure a Whittle estimator is consistent when directly applying to the realized variance. However, this consistency is also dependent on a distributional assumption on the proxy errors between integrated and realized variance. In contrast, the consistency of $\wt \cR_n$ does not hinge on any specific distributional assumptions and can be regarded as model-free. Moreover, in \Cref{thm main realized necessary}, we show that condition \eqref{RV condition eq} is not only sufficient but also necessary for the consistency of $\wt \cR_n$. Furthermore, it is worthwhile to point out that $\wh \cR_n$ and $\wt \cR_n$ admit very simple forms, and condition \eqref{RV condition eq} is easy to verify. These findings demonstrate that $\wh \cR_n$ is robust with respect to approximation errors. Moreover, a key insight from \Cref{RV section} is that estimating the roughness exponent of the volatility process from the price trajectory can be split into two distinct steps: constructing a consistent estimator based on integrated variance and approximating the integrated variance by realized variance.

We believe that the fact that our estimator is built on a strictly pathwise approach makes it very versatile and applicable also in situations in which trajectories are not based on fractional Brownian motion. As a matter of fact, our Examples \ref{Takagi example} and \ref{Takagi example 2} illustrate that our estimation procedure can work very well for certain deterministic fractal functions, and \Cref{Random Takagi example} shows that our estimator $\wh \cR_n$ can estimate the roughness exponent for non-Gaussian random processes. Furthermore, the simulation study presented in \Cref{simulation section} demonstrates that our estimation framework applies to processes in the Cauchy class \cite{Gneiting2004Hurst}. This class of processes permits a complete bifurcation between the degree of roughness and the long-range dependence, distinguishing it fundamentally from fractional Brownian motion.

One disadvantage of our original estimator $\wh\cR_n$ is that it is not scale invariant. Using an idea from  \cite{HanSchiedHurst}, we thus propose a scale-invariant modification $\cR^s_n$ of $\wh\cR_n$ in \Cref{scale-invariant section}. In \Cref{Theorem scale rate}, we obtain the rate of convergence of $\cR^s_n$ for fractional Brownian motion with drift.
The subsequent \Cref{simulation section} contains a simulation study illustrating the performance of our estimators. This study illustrates that passing to the scale-invariant estimator can greatly improve the estimation accuracy in practice.

\section{Main results}
\label{Section Main}

Consider a stochastic volatility model whose price process satisfies
\begin{equation}\label{price process eq 2}
dS_t=\sigma_tS_t\,dB_t,\qquad S_0=s_0>0,
\end{equation}
where $B$ is a standard Brownian motion and $\sigma_t$ is a continuous and adapted stochastic process.\footnote{One can easily add a drift to \eqref{price process eq 2}
 by performing an equivalent change of measure as in Section 2 of \cite{BayerFrizGatheral}. Since the results in our paper are almost-sure results, they will be unaffected by any equivalent change of measure.}
As explained in the introduction, our first goal in this paper is to estimate the roughness of the trajectories  $t\mapsto \sigma_t$ directly from discrete, equidistant  observations of the integrated variance,
 \begin{equation}\label{realized variance eq 2}
 \<\log S\>_t = \int_{0}^{t}\sigma^2_s \,ds,
 \end{equation}
 without having first to compute numerical approximations for  $\sigma_t$, e.g., via numerical differentiation of $t\mapsto\<\log S\>_t$. This is important, because in reality the volatility $\sigma_t$ is not directly observable and numerical errors in the computation of the numerical approximations might distort the roughness estimate (see, e.g., \cite{ContDasArtefact}).  In \Cref{RV section}, we will discuss the even more realistic situation in which also  the asset prices themselves can only be observed at discrete time points, thus leading to the need for an additional approximation of the integral in \eqref{realized variance eq 2}. 
  
 While our main results are concerned with rough stochastic volatility models based on fractional Brownian motion, a significant portion of our approach actually works completely trajectorial-wise, in a model-free setting; see \Cref{Section Pathwise}. So let $x:[0,1]\to\bR$  be any continuous function.
 For  $p \ge1$, the $p^\text{th}$ variation of the function $x$ along the $n^{\text{th}}$ dyadic partition  is defined as
 \begin{equation}\label{eq p variation}
\<x\>^{(p)}_n:= \sum_{k = 0}^{2^n-1}|x((k+1)2^{-n}) - x(k2^{-n})|^p.
\end{equation}
If there exists ${R} \in [0,1]$ such that 
\begin{equation*}
\lim_{n \ua \infty} \<x\>^{(p)}_n = \begin{cases}
0 &\quad \text{for} \quad p > 1/{R},\\
\infty &\quad \text{for} \quad p < 1/{R},
\end{cases}
\end{equation*}
we follow \cite{HanSchiedHurst} in referring to ${R}$ as the \textit{roughness exponent} of $x$. Intuitively, the smaller $R$ the rougher the trajectory $x$ and vice versa. Moreover, if $x$ is a typical sample path of  fractional Brownian motion, the roughness exponent ${R}$ is equal to the  Hurst parameter~\cite[Theorem 5.1]{HanSchiedHurst}. An analysis of  general properties of the roughness exponent can be found in  \cite{HanSchiedHurst}. There, we also provide an estimation procedure for ${R}$ from discrete observations of the trajectory $x$. 
However, the problem of estimating ${R}$ for a trajectory of stochastic volatility is more complex, because volatility cannot be measured directly. If we assume that asset prices can be observed continuously, then we can also observe their integrated variance \eqref{realized variance eq 2}. In our current pathwise setting, this corresponds to making  discrete observations of 
\begin{equation}\label{y antiderivative eq}
y(t)=\int_0^t g(x(s))\,ds,\qquad 0\le t\le 1,
\end{equation}
where $g:\bR\to\bR$ is sufficiently regular. For instance, in the rough fractional stochastic volatility model \eqref{price process eq}  and  \eqref{log volatility eq}, where log-volatility is given by a fractional Ornstein--Uhlenbeck process  \eqref{OU process eq}, we will take $x$ as a trajectory of the fractional Ornstein--Uhlenbeck process and $g(t)=(e^{t})^2=e^{2t}$. 

Let us now introduce our estimator. Suppose that for some given  $n\in\bN$  we have the  discrete observations $\{y(k2^{-n-2}):k=0,\dots, 2^{n+2}\}$ of the function  $y$ in \eqref{y antiderivative eq}. Based on these data points, we introduce the coefficients
\begin{equation}\label{eq vartheta}
\vartheta_{n,k}:= 2^{3n/2+3}\left(y\Big(\frac{4k}{2^{n+2}}\Big)-2y\Big(\frac{4k+1}{2^{n+2}}\Big)+2y\Big(\frac{4k+3}{2^{n+2}}\Big)-y\Big(\frac{4k+4}{2^{n+2}}\Big)\right),
\end{equation}
for $0 \le k \le 2^{n}-1$.
Our estimator for the roughness exponent of the trajectory $g\circ x$ is now given by
\begin{equation}\label{wh Rn eq}
\wh\cR_n(y): = 1-\frac{1}{n}\log_2 \sqrt{\sum_{k = 0}^{2^{n}-1}\vartheta_{n,k}^2}.
\end{equation}
This estimator  was first announced in \cite[Remark 2.2]{HanSchiedMatrix}, with a forward reference to this present paper for a mathematical analysis of its properties. The focus of \cite{HanSchiedMatrix} was on the theoretical properties of the approximated Faber–Schauder coefficients $\vartheta_{n,k}$, which provide the foundation for establishing the consistency of $\wh\cR_n$. In \Cref{rationale section}, we further clarify the intuition behind $\wh\cR_n$ and its relationship with the findings of \cite{HanSchiedMatrix, HanSchiedHurst}. Our estimator $\wh\cR_n$ admits a simple form, which is easy to calculate, and it will also lead to a new estimator $\wt\cR_n$ that can estimate the roughness exponent from  the realized variance $\sum_i(\log S_{t_i}-\log S_{t_{i-1}})^2$; see \Cref{RV section}.

\subsection{Strong consistency theorems}\label{consistency section}

We can now state our main results,  which show the strong consistency of $\wh\cR_n$ when applied to 
a variety of rough stochastic volatility models. In the sequel, $W^H=(W^H_t)_{0\le t\le1}$  denotes a fractional Brownian motion with Hurst parameter $H$, defined on a given probability space $(\Om,\cF,\bP)$.

\begin{theorem}\label{thm main rv} 
Suppose that    $g\in C^2(\bR)$ is strictly monotone and that   $X$ is given by 
\begin{equation}\label{XH eq}
X_t:=x_0+W^H_t+\int_0^t\xi_s\,ds,\qquad0\le t\le 1,
\end{equation}
 where $\xi$ is progressively measurable with respect to the natural filtration of $W^H$ and satisfies the following additional assumption. 
 \begin{enumerate}
 	\item \label{thm main rv part a}For $\bP$-a.e.~$\omega \in \Omega$, we have
 	\begin{equation}\label{eq new cond}
 		\lim_{n \ua \infty} 2^{n(2H-1)}\int_{0}^{1}\sum_{k = 0}^{2^n-1}\left[\int_{\frac{s+k}{2^n}}^{\frac{s+k+1}{2^n}}\xi_u\,du\right]^2 \,ds = 0,
 	\end{equation}
	where we put $\xi_t=0$ for $t>1$.
 	\item \label{thm main rv part b} For $\bP$-a.e.~$\omega \in \Omega$, $\xi_t(\omega) \in L^p([0,1])$ for some $p > \frac{5}{5 - 2H}$.
 \end{enumerate}
Then $\bP$-a.e.~sample path of the stochastic process
$$ g(X_t)=g\bigg(x_0+\int_0^t\xi_s\,ds+W^H_t\bigg),\qquad 0\le t\le1,
$$
admits the roughness exponent $H$. Moreover, for
$Y_t = \int_{0}^{t}g(X_s) \,ds$, we have $\lim_{n}\wh\cR_n(Y) =H$ $\bP$-a.s.
\end{theorem}

The following proposition provides a set of alternative conditions that imply the conditions \ref{thm main rv part a} and \ref{thm main rv part b} in \Cref{thm main rv}.

\begin{proposition}\label{cor main rv} Conditions \ref{thm main rv part a} and \ref{thm main rv part b} in \Cref{thm main rv} hold if $\xi$ satisfies the following two assumptions. 
	\begin{enumerate}[{\rm (a')}]
		\item \label{cor main rv part a}If $H \le 1/2$, we assume  that the function $t\mapsto\xi_t$ is $\bP$-a.s.~bounded in the sense that  there exists a finite random variable $C$ such that $|\xi_t(\om)|\le C(\om)$ for a.e.~$t$ and $\bP$-a.e.~$\om$. 
		\item \label{cor main rv part b}If $H>1/2$, we assume that for $\bP$-a.e.~$\om$ the function $t\mapsto \xi_t(\om)$ is H\"older continuous with some exponent $\alpha(\om)>2H-1$.
	\end{enumerate}
\end{proposition}

\begin{proof} Under conditions (a') and (b'),  there exists a finite random variable $C$ such that $|\xi_t(\omega)| \le C(\omega)$ for a.e.~$t$ and a.s.~$\omega \in \Omega$. In particular, condition \ref{thm main rv part b} of \Cref{thm main rv}  is satisfied. Moreover, 
	\begin{equation*}
		2^{n(2H-1)}\int_{0}^{1}\sum_{k = 0}^{2^n-1}\left[\int_{\frac{s+k}{2^n}}^{\frac{s+k+1}{2^n}}\xi_u\,du\right]^2 \,ds \le 2^{n(2H-3)}\int_{0}^{1}\sum_{k = 0}^{2^n-1} C^2(\omega) \,ds = 2^{n(2H-2)} C^2(\omega) \rightarrow 0
	\end{equation*}
	as $n \ua \infty$, which yields the condition \ref{thm main rv part a} of \Cref{thm main rv}. 
\end{proof}

\Cref{thm main rv} applies in particular to the case in which $X$ is a  fractional Ornstein--Uhlenbeck process.
	This process
	was first introduced in \cite{CheriditoFOU} as the pathwise solution of the integral equation
	\begin{equation}\label{OU process eq 2}
		X^H_t=x_0+\rho\int_0^t(\mu-X^H_s)\,ds+W^H_t,\qquad t\in[0,1],
	\end{equation}
	where $x_0,\rho,\mu\in\mathbb{R}$ are given parameters, and it was suggested by  Gatheral et al.~\cite{GatheralRosenbaum} as a suitable model for log volatility, i.e., $\sigma_t=e^{X^H_t}$. To see that  \Cref{thm main rv} applies in this case, we let  $\xi_t:= \rho(\mu-X_t^H)$. Since $\bP$-almost all sample paths of $W^H$ are H\"older continuous  for any exponent $\alpha<H$, the same is true for $X^H$ and in turn for $\xi$. Moreover, 
  $H>2H-1$ and so conditions (a') and (b') from \Cref{cor main rv}  are satisfied.

	\begin{remark}\label{thm main rv variants remark} The validity of \Cref{thm main rv} under the stronger conditions (a') and (b') of \Cref{cor main rv} can alternatively be obtained via an application of \cite[Theorem 1.4]{HanSchiedGirsanov}, which provides a criterion for the absolute continuity of the law of the process $(x_0+\int_0^t\xi_s\,ds+W^H_t)_{0\le t\le 1}$ with respect to the law of $x_0+W^H$. When taking this route, the natural filtration of $W^H$ can be replaced by a larger filtration satisfying \cite[Definition 1.3]{HanSchiedGirsanov}. By using \cite[Theorem 1.6]{HanSchiedGirsanov} instead of \cite[Theorem 1.4]{HanSchiedGirsanov}, one sees that the process $X$ in \eqref{XH eq} can also be  an adapted solution to the fractional integral equation  
	\begin{equation}\label{eq integral equ}
		X_t=x_0+\int_0^tb(X_s)\,ds+W^H_t,\qquad 0\le t\le1,
	\end{equation}
	where $b:\bR\to\bR$ is locally bounded and, for $H>1/2$, locally H\"older continuous with some exponent $\alpha>2-1/H$. Nevertheless, although conditions (a') and (b') are slightly stronger, they allow us to derive the rate of convergence of the estimator $\wh \cR_n$ and related scale-invariant estimators; see \Cref{scale-invariant section}.	\end{remark}

Many rough stochastic volatility models in the literature are not based on fractional Brownian motion but on the closely related Riemann--Liouville process. This is in particular the case for the  rough Bergomi model \cite{Forde2022}. Here, we work with the following rescaled version of the Riemann--Liouville process that is given by 
\begin{equation*}
	Z^H_t = \int_0^t (t-s)^{H-\frac{1}{2}}\,dW_s, \qquad t \ge 0. 
\end{equation*}
for a parameter $H\in(0,1)$ and a standard Brownian motion $W$ on $(\Omega,\cF,\bP)$. Moreover, the following shifted Riemann--Liouville process is sometimes considered in the literature. For a parameter  $\nu \ge 0$, it is by 
$$Z^{H,\nu}_t := Z^H_{t + \nu} - Z^H_\nu,\qquad t \ge 0.$$

\begin{theorem}\label{Thm Rough Bergomi}
	Suppose that $g \in C^2(\bR)$ is strictly monotone and that $X$ is given by 
	\begin{equation}\label{eq RL with drift}
		X_t := x_0 + Z^{H,\nu} _t+  \int_{0}^{t}\xi_s \, ds \qquad 0 \le t \le 1,
	\end{equation}
	where $\nu\ge0$ and $\xi_t \in C[0,1]$ is progressively measurable with respect to the natural filtration of $Z^{H,\nu}$ and satisfies conditions \ref{thm main rv part a} and \ref{thm main rv part b} in \Cref{thm main rv}.
Then $\bP$-a.e.~sample path of the stochastic process
	\begin{equation*}
		g(X_t) = g\left(x_0 + Z^{H,\nu} _t + \int_{0}^{t}\xi_s \,ds\right), \qquad 0 \le t \le 1,
	\end{equation*}
	admits the roughness exponent $H$. Moreover, for $Y_t = \int_0^t g(X_s)\,ds$, we have $\lim_n \wh \cR_n(Y) = H$ $\bP$-a.s.
\end{theorem}
The above theorem in particular applies to the following class of stochastic volatility models whose volatility processes $\sigma$ admit the following dynamics,
\begin{equation}\label{eq VHu}
	\sigma_t = \exp(V^{H,\nu}_t) \quad \text{and} \quad V^{H,\nu}_t := \gamma Z^{H,\nu}_t - \frac{\gamma^2}{4H}t^{2H}, \qquad  t \in [0,1].
\end{equation}
Taking $\nu = 0$ recovers the rough Bergomi model in the form discussed in \cite[Equation (28)]{Forde2022}. On the other hand, as $\nu \ua \infty$, the process $Z^{H,\nu}$ converges to fractional Brownian motion  in distribution. Thus, varying $\nu$ generates a continuum of stochastic volatility models that interpolate between the rough Bergomi model and a model based on fractional Brownian motion. As a direct corollary of \Cref{Thm Rough Bergomi}, stochastic volatility models of the form \eqref{eq VHu} admit the roughness exponent $H$, and our estimator $\wh \cR_n$ is a consistent estimator of this exponent.
\begin{corollary}\label{Cor Rough Bergomi}
	Let $\sigma$ be as in \eqref{eq VHu} for given $\nu \ge 0$. Then the volatility process $\sigma_t$ admits the roughness exponent $H$ with probability one, and we have 
	\begin{equation}
		\lim_{n \ua \infty} \wh \cR_n \left(\int_{0}^{\cdot}\sigma_s^2 \,ds\right) = H \qquad \bP\text{-}a.s.
	\end{equation}
\end{corollary}

\subsection{A scale-invariant estimator}\label{scale-invariant section}

By definition, the roughness exponent is scale-invariant, but our estimator $\wh \cR_n$ is not. To wit,  for every trajectory $y\in C[0,1]$ we have
\begin{equation}\label{eq scale invariant}
	\wh\cR_n(\lambda y) - \wh\cR_n(y) = -\frac{\log_2 |\lambda|}{n} \quad \text{for} \quad \lambda \neq 0.
\end{equation}
Consequently, a scaling factor $\lambda$ may either remove or introduce a bias into an estimate, and it can notably slow down or speed up the convergence of $\wh\cR_n(y)$. This will be illustrated by the simulation studies provided in \Cref{simulation section}. Since the scaling behavior \eqref{eq scale invariant} is identical to that of the estimator $\wh R^\ast_n$ in \cite{HanSchiedHurst}, we can construct scale-invariant modifications of $\wh\cR_n$ in a manner completely analogous to the definitions in \cite[Section 8]{HanSchiedHurst}. 
Here, we carry this out for the analogue of sequential scaling proposed in \cite[Definition 8.1]{HanSchiedHurst}. The underlying idea is fairly simple: We choose $m<n$ and then search for that scaling factor $\lambda$ that minimizes  the weighted mean-squared differences $\wh\cR_k(\lambda y) - \wh\cR_{k-1}({\lambda y})$ for $k = m+1,\dots,n$. The intuition is that such an optimal scaling factor $\lambda$ enforces the convergence of the estimates $\wh\cR_k(\lambda y)$.

 \begin{definition}\label{sls and tls def}Fix $m\in\bN$ and $\alpha_0,\dots,\alpha_{m}\ge0$ with $\alpha_0>0$.
For $n>m$, the \textit{sequential scaling factor} $\lambda_{n}^{s}$ and the \textit{sequential scale estimate} $\cR^s_{n}(y)$ are defined as follows,
	\begin{equation}\label{eq_def_seqloc}
		\begin{split}
			\lambda^{s}_{n}&:= \argmin_{\lambda > 0} \sum_{k = n-m}^{n}\alpha_{n-k}\Big(\wh\cR_k(\lambda y) - \wh\cR_{k-1}(\lambda y)\Big)^2 \quad \text{and} \quad
			\cR^s_{n}(y):= \wh\cR_{n}(\lambda^{s}_{n} y).
		\end{split}
	\end{equation}
	The corresponding mapping $\cR^s_{n}:C[0,1]\rightarrow \bR$ will be called the \textit{sequential scale estimator}.
	 \end{definition}

Just as Proposition 8.3 in \cite{HanSchiedHurst}, one can prove the following result.

\begin{proposition}\label{scale-inv prop} Consider the context of Definition~\ref{sls and tls def}  with fixed $m\in\bN$ and $\alpha_0,\dots,\alpha_{m}\ge0$ such that $\alpha_0>0$.
\begin{enumerate}
\item The  optimization problem  \eqref{eq_def_seqloc} admits a unique solution for every function $y\in C[0,1]$. In particular, all objects in Definition~\ref{sls and tls def} are well defined.
\item The sequential scale estimator $\cR^s_{n}$ can be represented as follows as a  linear combination of $\wh\cR_{n-m-1},\dots, \wh\cR_n$,
\begin{equation*}
	\cR_n^s = \beta_{n,n}\wh\cR_n + \beta_{n,n-1}\wh\cR_{n-1} + \cdots + \beta_{n,n-m-1}\wh\cR_{n-m-1},
\end{equation*}
where
\begin{equation}\label{eq beta nk}
	\scalemath{0.9}{\beta_{n,k}=\begin{cases}\displaystyle1+\frac{\alpha_0}{c^{\textrm{\rm s}}_{n}n^2(n-1)}&\text{if $k=n$,}\\
		\displaystyle\frac1{c^{\textrm{\rm s}}_{n}nk}\Big(\frac{\alpha_{n-k}}{k-1}-\frac{\alpha_{n-k-1}}{k+1}\Big)&\text{if $n-m\le k\le n-1$,}\\
		\displaystyle
		\frac{-\alpha_m}{c^{\textrm{\rm s}}_{n}n(n-m)(n-m-1)}&\text{if $ k= n-m-1$,}
	\end{cases}\quad\text{for}\quad c^{s}_{n}:= \sum_{k = n-m}^{n}\frac{\alpha_{n-k}}{k^{2}(k-1)^{2}}.}
\end{equation}

\item The sequential scale estimator is scale-invariant. That is,  for $n>m$, $y \in C[0,1]$, and $\lambda \neq 0$, we have $\cR^s_{n}(\lambda y) = \cR^s_{n}(y)$.
	\item If $y\in C[0,1]$ and $R\in[0,1]$ are such that there exists $\lambda\neq 0$ for which $|\wh\cR_n(\lambda y)-R|=\mathcal{O}(a_n)$ as $n\ua\infty$ for some sequence $(a_n)$ with $a_n=o(1/n)$, then $			|\cR^s_{n}(y)-R| =\mathcal{O}(na_n)$. \label{scale-inv item d}
\end{enumerate}
\end{proposition}

In particular, assertion \ref{scale-inv item d} of \Cref{scale-inv prop} establishes the consistency criteria and obtains the rate of convergence of the sequential scale estimator $\cR^s_n$. By applying the above result, the following theorem verifies the strong consistency and provides the rate of convergence of $ \cR^s_n$ for the fractional Brownian motion with drift. The proof of \Cref{Theorem scale rate} is provided in \Cref{Section Proof Rate}.

\begin{theorem}\label{Theorem scale rate}
Let $m \in \bN$, $\alpha_0 > 0$ and $\alpha_1, \cdots, \alpha_m \ge 0$ be fixed. Let $X$ be as in \Cref{cor main rv} and 
\begin{equation*}
	Y_t = \int_0^t X_s \,ds, \qquad t \in [0,1].
\end{equation*}
 Then the following almost sure rate of convergence holds for the sequential scale estimator $\cR^s_n$,
		\begin{equation}\label{eq rate of convergence Y}
			|\cR^s_n(Y) - H| = \mathcal{O}\big(2^{-n/2}\sqrt{\log n}\big).
	\end{equation}
\end{theorem}

\begin{remark}\label{Remark scale rate}
We now investigate the impact of the integration operation on the estimation of the roughness exponent. Let $R^s_n$ be the sequential scale estimator based on $\wh R_n$ as in \cite[Definition 8.1]{HanSchiedHurst}, suppose that, in contrast to the setting of \Cref{Theorem scale rate}, the process $X$ as in \eqref{XH eq} can be directly observed. It is shown in \cite[Corollary 8.4]{HanSchiedHurst} that when applying $R^s_n$ to estimate the Hurst parameter for $X$, the following almost sure rate of convergence holds 
	\begin{equation}\label{eq rate of convergence X}
		|R^s_n(X) - H| = \mathcal{O}\big(2^{-n/2}\sqrt{\log n}\big).
	\end{equation}
  Comparing the upper bounds in \eqref{eq rate of convergence X} and \eqref{eq rate of convergence Y} shows that both $R^s_n$ and $\cR^s_n$ attain the same asymptotic order of convergence in the worst case. Hence, while the actual convergence rates may differ, the integration step does not worsen the asymptotic rate of decay guaranteed by these upper estimates.
   
    It is also worthwhile to point out that \Cref{Theorem scale rate} has recently been extended to  the case in which $Y$ is the antiderivative of $g \circ X$ for a nonlinear function $g \in C^1[0,1]$; see \cite[Theorem 1.2]{HanSchiedSIFIN}.  However, \Cref{Theorem scale rate} is not obsolete, because it is the main ingredient in the proof of \cite[Theorem 1.2]{HanSchiedSIFIN}.
\end{remark}

\subsection{Estimation of the roughness exponent from realized variance}\label{RV section}

As previously mentioned, the integrated variance \eqref{realized variance eq} is usually approximated by a finite sum of the form $\sum_i(\log S_{t_i}-\log S_{t_{i-1}})^2$ based on discrete observations $ S_{t_i}$ of the price process. It has been shown in the literature \cite{BolkoPakkanen2022GMM, Fukasawa2022Estimation} that the error arising from the approximation of the integrated with the realized variance introduces an \lq extra' noise to the true \lq signal' $\<\log S\>$, often resulting in underestimating the roughness exponent of the hidden volatility process $\sigma$. 

The error between the integrated and realized variance has been previously studied in the framework of roughness estimation. Those studies were pioneered by Fukasawa et al.~\cite{Fukasawa2022Estimation}, where a Whittle-type estimator for the Hurst parameter was developed under the assumption that approximation errors are log-normally distributed and independent of the Brownian motion $B$ in \eqref{price process eq}. Under a similar setting, Bolko et al.~\cite{BolkoPakkanen2022GMM} proposed an estimator based on the generalized method of moments approach. Finally, Chong et al.~\cite{Chong2022CLT, Chong2022Minimax} alleviated the assumption on proxy errors in a semi-parametric setup, in which,   with the exception of the Hurst parameter of the underlying fractional Brownian motion, all components are fully nonparametric. 

In this section, we will show that our estimator is effective in accommodating the proxy error between the realized and integrated variance. To this end, we continue to consider a stochastic volatility model whose price process satisfies
\begin{equation}\label{price process eq new}
dS_t=\sigma_tS_t\,dB_t,\qquad S_0=s_0>0,
\end{equation}
where $B$ is a standard Brownian motion and $\sigma_t$ is a continuous and adapted stochastic process.  We denote the integrated variance by
\begin{equation}\label{IV process eq}
	Y_t = \<\log S\>_t = \int_0^t \sigma^2_s\,ds, \qquad t \in [0,1].
\end{equation}
Furthermore, we let $(m_n)_{n \in \bN_0}$ be a fixed increasing sequence, where $m_n \ge 1$ can be regarded as the number of observed data points used to approximate the integrated variance 
$Y_{(k+1)2^{-n}}-Y_{k2^{-n}}$
over each interval $[k2^{-n},(k+1)2^{-n}]$. We denote
by
\begin{equation}\label{eq process RV}
	\wh Y^{(n)}_t := \sum_{i = 1}^{\floor{2^{n}m_nt}}\left(\log S_{\frac{i}{2^n m_n}} - \log S_{\frac{i-1}{2^n m_n}} \right)^2
\end{equation}
the realized variance calculated from the price process \eqref{price process eq new} with a mesh size $(2^n m_n)^{-1}$. Let us define the proxy coefficients $\wt \vartheta_{n,k}$ computed from $\wh Y^{(n)}_t$,
\begin{equation}\label{eq wt vartheta}
	\wt \vartheta_{n,k} := 2^{3n/2+3}\left(\wh Y^{(n+2)}_{\frac{4k}{2^{n+2}}} - 2\wh Y^{(n+2)}_{\frac{4k+1}{2^{n+2}}} + 2\wh Y^{(n+2)}_{\frac{4k+3}{2^{n+2}}} - \wh Y^{(n+2)}_{\frac{4k+4}{2^{n+2}}}\right).
\end{equation}
By replacing $\vartheta_{n,k}$ with $\wt \vartheta_{n,k}$, we can construct an estimator $\wt \cR_n$ that directly estimates the roughness exponent from the realized variance as follows 
\begin{equation*}
	\wt \cR_n(S) := 1 - \frac{1}{n}\log_2 \sqrt{\sum_{k = 0}^{2^n-1}\wt \vartheta^2_{n,k}},
\end{equation*}
where the argument  $S$ of $\wt \cR_n$ emphasizes that $\wt \cR_n$ estimates directly from the price process $S$. The following theorem shows that if the sequence $(m_n)$ is properly chosen, the estimators $\wh \cR_n$ and $\wt \cR_n$ converge to the same value.

\begin{theorem}\label{thm main realized}
	Let $Y_t$ and $\wh Y^{(n)}_t$ be as in \eqref{IV process eq} and \eqref{eq process RV}. Suppose that $\bE\left[\sup_{t \in [0,1]}\sigma^{16}_s\,ds\right] < \infty$ and that there exists $R \in (0,1)$ such that $\wh \cR_n(Y) \rightarrow R$ with probability one. If 
	\begin{equation}\label{RV condition eq}
		\lim_{n \ua \infty}\frac{1}{n}\log_2 m_n > 2R,
	\end{equation}
	then $\wt \cR_n(S) \rightarrow R$ as $n \ua \infty$ with probability one.
\end{theorem}

Let us comment on the above theorem. First of all, \Cref{thm main realized} provides a simple condition guaranteeing that the estimates $\wt \cR_n(S)$ and $\wh \cR_n(Y)$ converge to the same limit as $n$ approaches infinity. Particularly, \Cref{thm main realized} specifies the number of data points required to compute the realized variance for each interval $[k2^{-n},(k+1)2^{-n}]$. For instance, if the volatility process $\sigma_t$ admits a relatively small roughness exponent, the realized variance only needs to be computed on a relatively sparse mesh so that the estimator $\wt \cR_n$ is a consistent estimator. This is because a rough \lq signal' $\sigma_t$ is very persistent against  the proxy error between the integrated and realized variance. In contrast, to achieve a consistent estimation of the roughness exponent for a relatively smooth volatility process $\sigma_t$, the proxy error between the integrated and realized variance has to be very small so that the actual roughness of the \lq signal' $\sigma_t$ is not overshadowed by the \lq noise' $|Y_t - \wh Y^{(n)}_t|$. 

It is also worth noting that the estimators $\wh \cR_n$ and $\wt \cR_n$ take remarkably simple forms. Moreover, condition \eqref{RV condition eq} is straightforward to verify. Furthermore, one conceptual contribution of \Cref{thm main realized} is to show that estimating the \lq roughness' of the volatility process from price trajectory can be decomposed into two distinct tasks, even in a strictly pathwise manner. First, it involves constructing a consistent estimator from discrete observations of integrated variance. Second, one needs to analyze the perturbation caused by observation noise to determine the number of observed prices included in each realized variance calculation.  

As noted in the introduction, Bolko et al.~\cite{BolkoPakkanen2022GMM} and Fukasawa et al.~\cite{Fukasawa2022Estimation} were among the first to address the proxy error between integrated and realized volatility, $\wh Y^{(n)}_t - Y_t$, and to construct consistent estimators of the true Hurst parameter based on price observations. In particular, the condition \eqref{RV condition eq} was first discovered in \cite[Theorem 2.8]{Fukasawa2022Estimation} in a very similar form. However, the methods in \cite{Fukasawa2022Estimation} depend on specific distributional assumptions about the proxy error. For instance, in \cite[Theorem 2.1]{Fukasawa2022Estimation}, the authors assume that the triangular arrays
	\begin{equation*}
	\left(\sqrt{\frac{m_n}{2}}\log \frac{\wh Y^{(n)}_{2^{-n}i} - \wh Y^{(n)}_{2^{-n}(i-1)}  }{Y_{2^{-n}i} - Y_{2^{-n}(i-1)}}\right)_{1 \le i \le 2^n} 
	\end{equation*}
	converge to an (infinite-dimensional) Gaussian vector of standard normal random variables. By contrast, we only impose the moment condition $\bE\left[\sup_{t \in [0,1]}\sigma^{16}_s\,ds\right] < \infty$. By Fernique's theorem, this moment condition is satisfied as soon as the volatility is of the form $\sigma_t=g(X_t)$ for some Gaussian process $X$ and a function $g$ with $|g(x)|\le C e^{|x|^{\beta}}$ for some $C>0$ and $\beta<2$.

	Finally, as stated in \Cref{thm main realized}, the condition \eqref{RV condition eq} is a sufficient condition for the consistency of the estimator $\wt \cR_n$. In the following theorem, we are going to show that this condition is also a necessary condition, at least in the special case of the volatility model \eqref{price process eq new} in which the Brownian motion  $B$ driving the asset price process $S$ is independent of the volatility process $\sigma$.  In this simple case, the proxy error $|Y_t - \wh Y^{(n)}_t|$ is completely traceable even for a finite generation $n$, and this fact allows us to show that the condition \eqref{RV condition eq} is also a necessary condition for the consistency of the estimator $\wt \cR_n$. 
	
	\begin{theorem}\label{thm main realized necessary}
		Suppose that $t \mapsto \sigma_t$ is a progressively measurable process that is $\bP$-a.s.~bounded in the sense that there exists a finite random variable $C$ such that $|\sigma_t(\omega)| \le C(\omega)$ for a.e.~$t \in [0,1]$ and $\bP$-a.e.~$\omega \in \Omega$. 
		Let furthermore $S$ be given by $dS_t=\sigma_tS_t\,dB_s$, where $B$ is a standard Brownian motion independent of $\sigma$ and 
		let $Y_t$ and $\wh Y^{(n)}_t$  as in \eqref{IV process eq} and \eqref{eq process RV}, respectively. Suppose finally that $\lim_n n^{-1}\log_2 m_n = 2\alpha$ and $\lim_n\wh \cR_n(Y) = R$ $\bP$-a.s. Then we have $\bP$-a.s.,
		\begin{equation*}
			\lim_{n \ua \infty} \wt \cR_n(S) = \begin{cases}
				\alpha, &\text{if $\alpha < R,$}\\
				R, &\text{if $\alpha > R.$}
			\end{cases}
		\end{equation*} 
	\end{theorem}
	
To the best of our knowledge, the above theorem provides the first result showing that condition~\eqref{RV condition eq} is not only sufficient but also necessary. In particular, it shows that insufficient sampling in the computation of realized variance leads to an underestimation of the roughness exponent. As a direct corollary, \Cref{thm main realized necessary} also confirms that the asymptotic sampling rate required in condition~\eqref{RV condition eq} cannot be further relaxed in general. Hence, it establishes the optimality of asymptotic sampling rate in condition~\eqref{RV condition eq}. 

In practice, condition \eqref{RV condition eq} has the disappointing effect that $\mathcal{O}(2^{2n})$ data points need to be added between each original observation point, if we want to replace the integrated variance $Y$ with the realized variance $\wh Y^{(n+2)}$ in the estimator $\wt \cR_{n}$ and are agnostic about the true value of $R$. We emphasize that this effect is not a shortcoming of the structure of our estimator  $\wh\cR_n$  but solely result of the poor performance of the realized variance  $\wh Y^{(n+2)}$  as approximation of the integrated variance $Y$. This poor performance has been previously noted in the literature, and it has led to development of other more efficient estimators of the integrated variance; see, e.g., \cite{AitSahaliaJacod} and the references therein. In future research, we aim to analyze the performance of our estimator if in \eqref{eq wt vartheta}, $\wh Y^{(n)}$ is replaced by such a more efficient estimate for the integrated variance $Y$.

\subsection{Simulation study}\label{simulation section}

In this section, we illustrate the application of  \Cref{thm main rv} and \Cref{scale-inv prop}  by means of numerical simulations. We will see that the estimation performance can be significantly boosted by replacing $\wh\cR_n$ with the sequential scale estimator $\cR^s_n$.

We start by illustrating \Cref{thm main rv} for the simple choice $g(x)=x$. Recall from \eqref{eq vartheta} and \eqref{wh Rn eq} that for given $n\in\bN$, the computation of $\wh\cR_n(y)$ requires observations of the trajectory $y$ at all values of the time grid $\bT_{n+2}:=\{k2^{-n-2}:k=0,1,\dots, 2^{n+2}\}$. When using for $y$ the antiderivative of a sample path of fractional Brownian motion $W^H$, we generate the values of $W^H$ on the finer grid $\bT_N$ with $N=n+6$ using the R-package {\tt Yumia} \cite{Yumia}. Then we put
\begin{equation}\label{discrete fBM Y eq}
Y_{k2^{-n-2}}:=2^{-N}\sum_{j = 1}^{2^{N-n-2}k}W^H_{j2^{-N}},\quad k=0,1,\dots, 2^{n+2},\end{equation}
which is an approximation of $\int_0^tW^H_s\,ds$ by Riemann sums. Our corresponding simulation results are displayed in \Cref{fig:Box Ori}.

\begin{figure}[H]
	\centering
	\includegraphics[width=8cm]{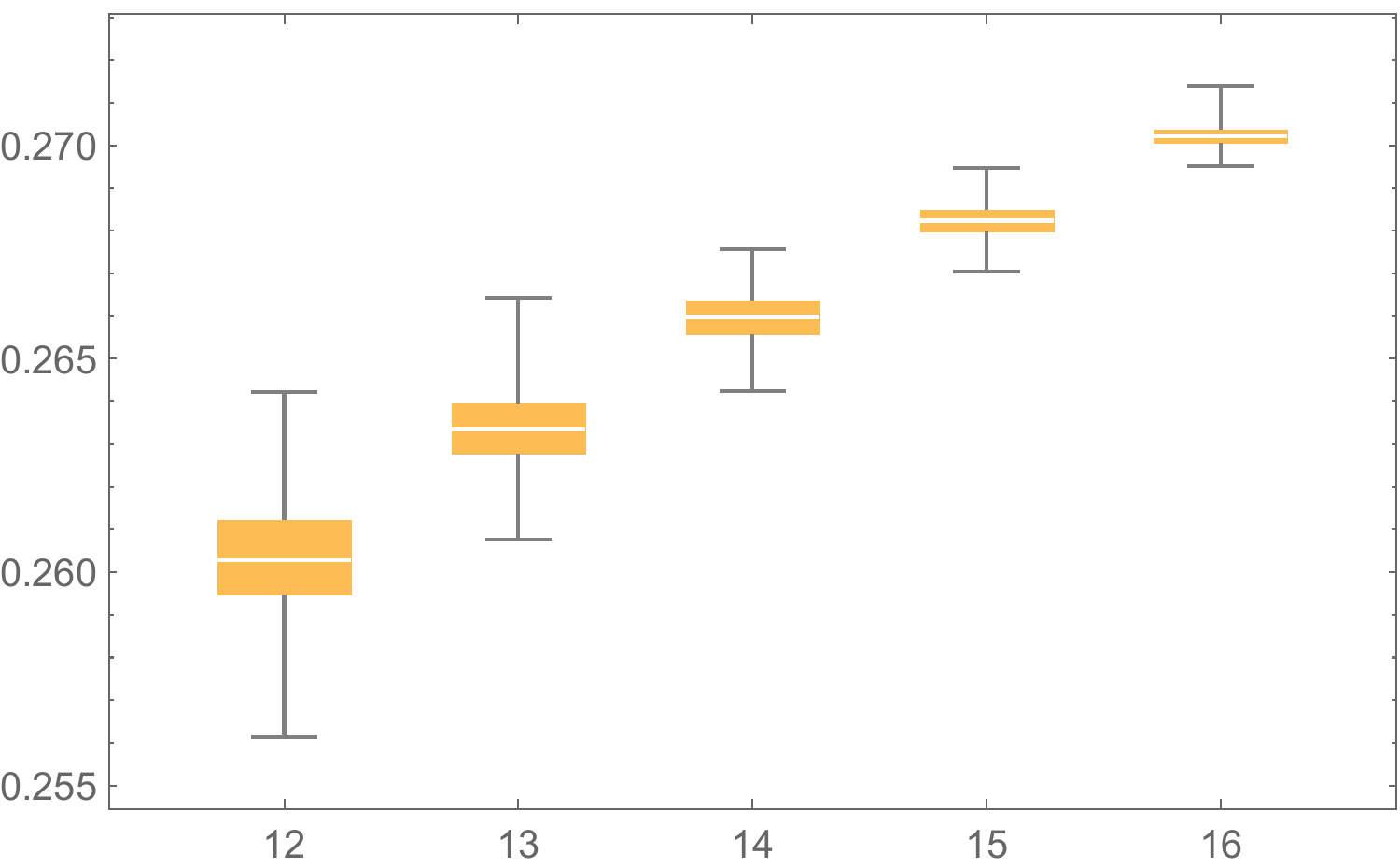}
	\quad		\includegraphics[width=8cm]{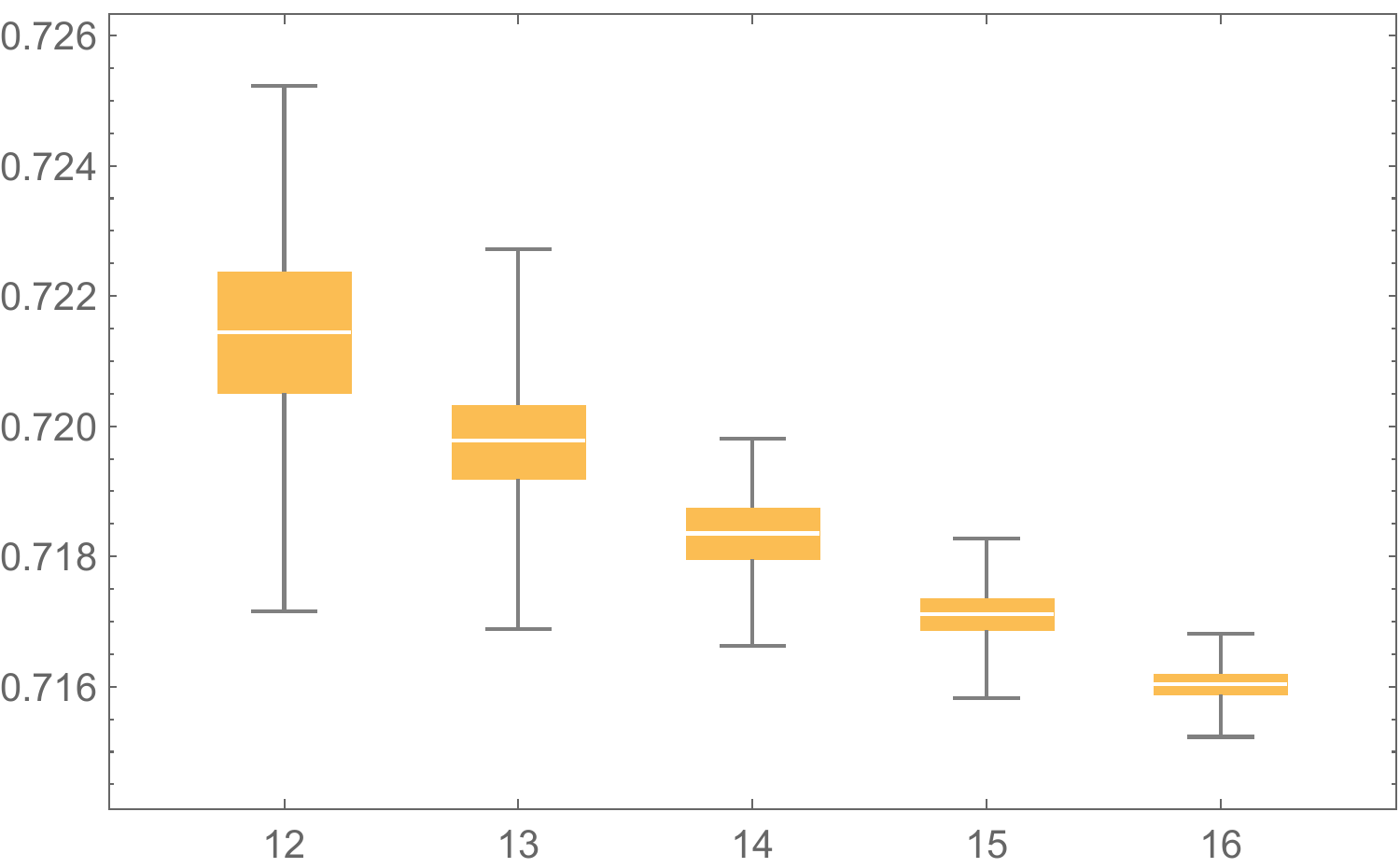}
	\caption{Box plots of the estimates $\wh\cR_n(Y)$ for $n=12,\dots, 16$, based on  1,000 sample paths of fractional Brownian motion with $H = 0.3$ (left), $H = 0.7$ (right), and $Y$ as in \eqref{discrete fBM Y eq}.  }
	\label{fig:Box Ori}
\end{figure}

As one can see from \Cref{fig:Box Ori}, the estimator $\wh\cR_n$ performs relatively well but also exhibits a certain bias. 
This bias can be completely removed by passing to the scale-invariant estimator $\cR^s_n$; see \Cref{fig:Box Ori si}.

\begin{figure}[H]
	\centering
	\includegraphics[width=8cm]{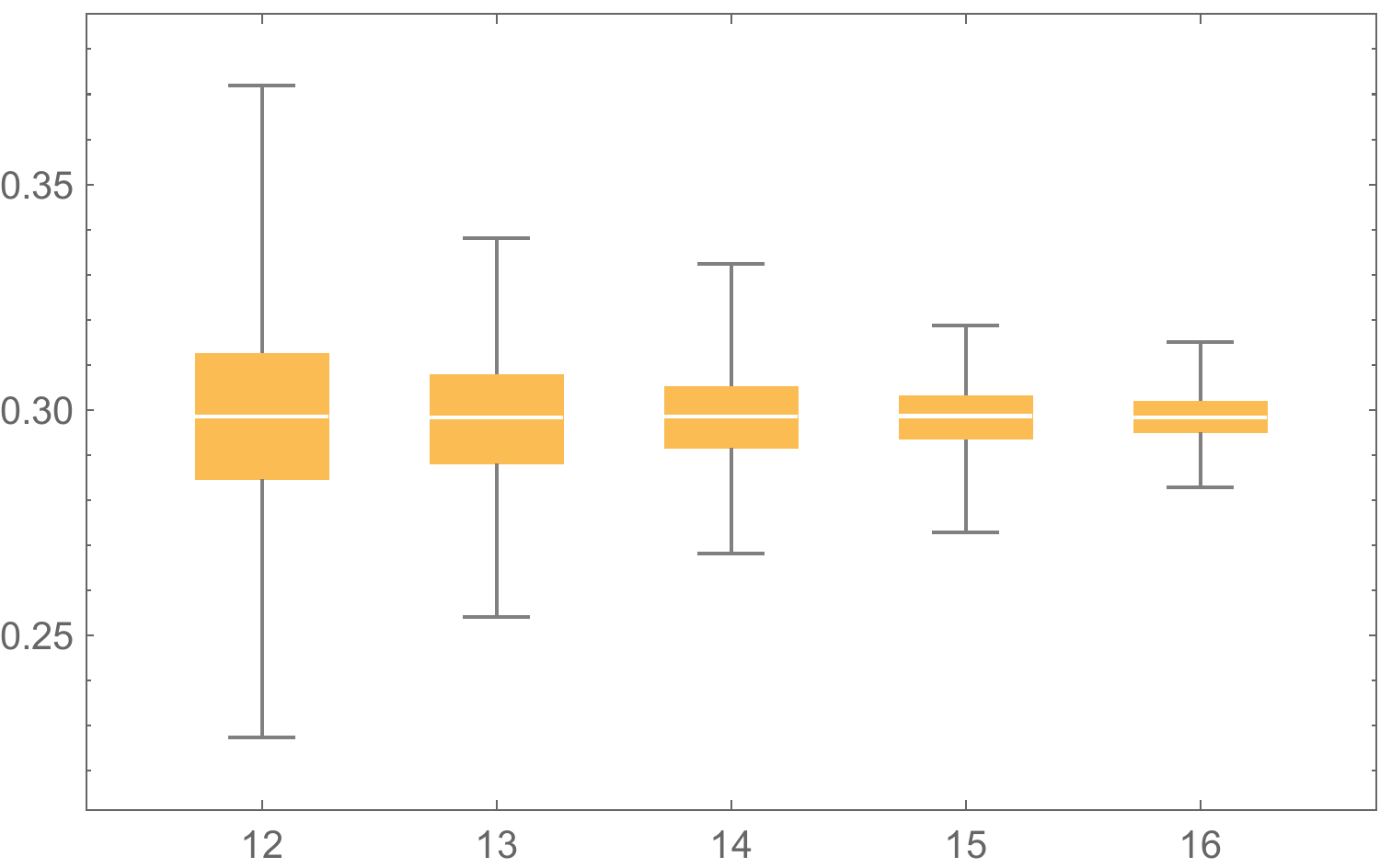}
	\quad		\includegraphics[width=8cm]{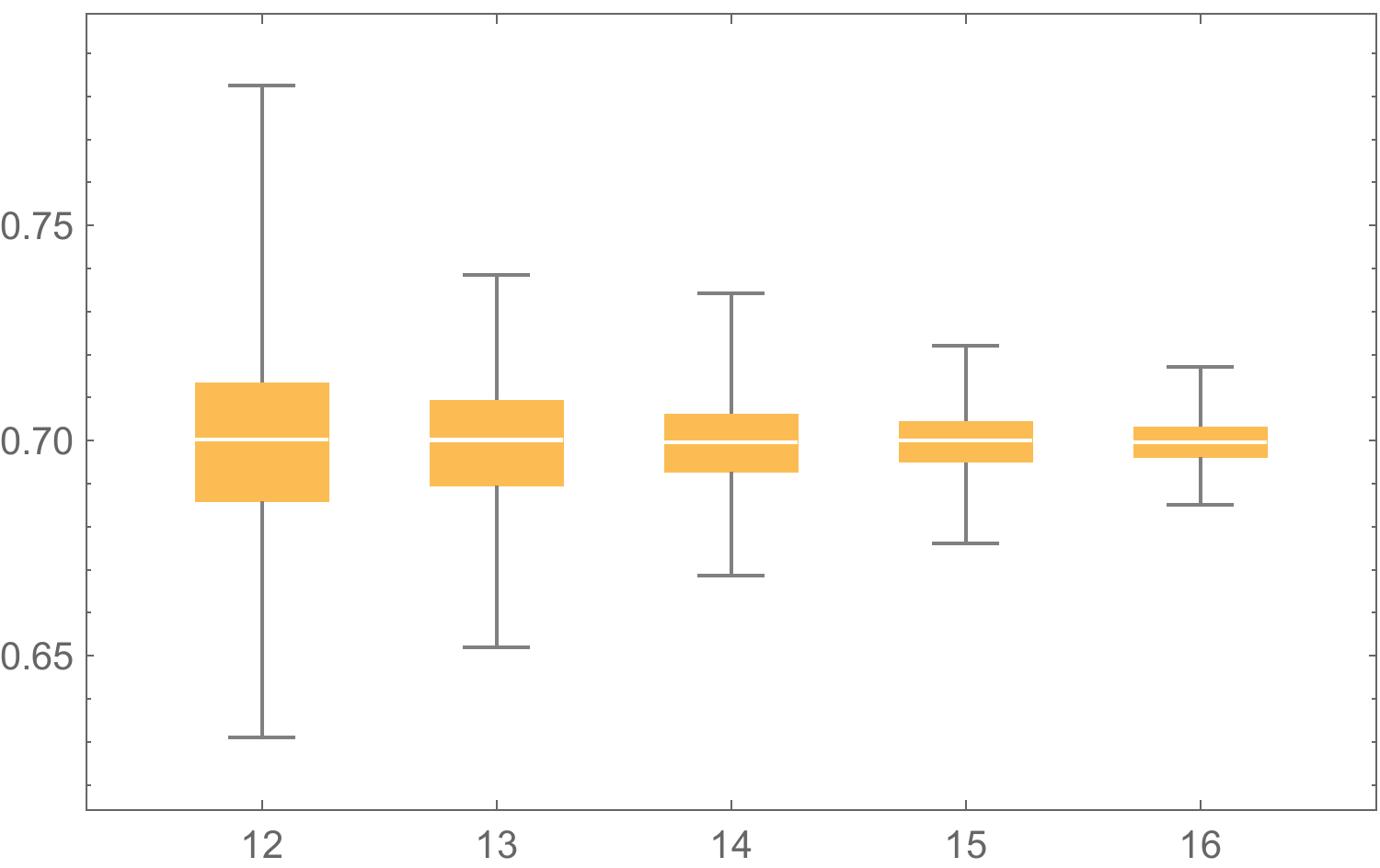}
	\caption{Box plots of the sequential scale estimates $ \cR^s_n(Y)$ for $n=12,\dots, 16$, based on  1,000 sample paths of fractional Brownian motion with $H = 0.3$ (left), $H = 0.7$ (right), and $Y$ as in \eqref{discrete fBM Y eq}. The other parameters are chosen to be $m = 3$ and $\alpha_k = 1$ for $k = 0,1,2,3$.} 
	\label{fig:Box Ori si}
\end{figure}

Now we apply our estimator $\wh\cR_n$ to a model in which log volatility, $\log\sigma_t$, is given by a fractional Ornstein--Uhlenbeck process $X^H$ of the form 
$$X^H_t=x_0+\rho\int_0^t(\mu-X^H_s)\,ds+W^H_t,\qquad t\in[0,1],
$$
and we make discrete observations of the process
$$\int_0^t \sigma^2_s\,ds=\int_0^t e^{2X^H_s}\,ds,\qquad 0\le t\le1.
$$
To this end, we take again $N=n+6$ and simulate the values $X^H_{k2^{-N}}$ ($k=0,\dots, 2^N$) by the R-package {\tt Yumia} \cite{Yumia}. Then we put 
\begin{equation}\label{discrete OU Y eq}
Y^{\sigma}_{k2^{-n-2}}:=2^{-N}\sum_{j = 1}^{2^{N-n-2}k}\exp\Big(2X^H_{j2^{-N}}\Big),\quad k=0,1,\dots, 2^{n+2},\end{equation}
which is an approximation of $\int_0^te^{2X^H_s}\,ds$ by Riemann sums. As one can see from \Cref{fig:Box Ori FOU}, the original estimator $\wh\cR_n$ performs rather poorly in this case, while the sequential scale estimator $\cR^s_n$ performs almost as well as for the simple case  $Y_t=\int_0^tW^H_s\,ds$.  This is due to the fact that the function $g(t)=e^{2t}$ used in \eqref{discrete OU Y eq} distorts substantially the scale of the underlying process, but this distortion can be remedied by using the sequential scale estimator.

\begin{figure}[H]
	\centering
	\includegraphics[width=8cm]{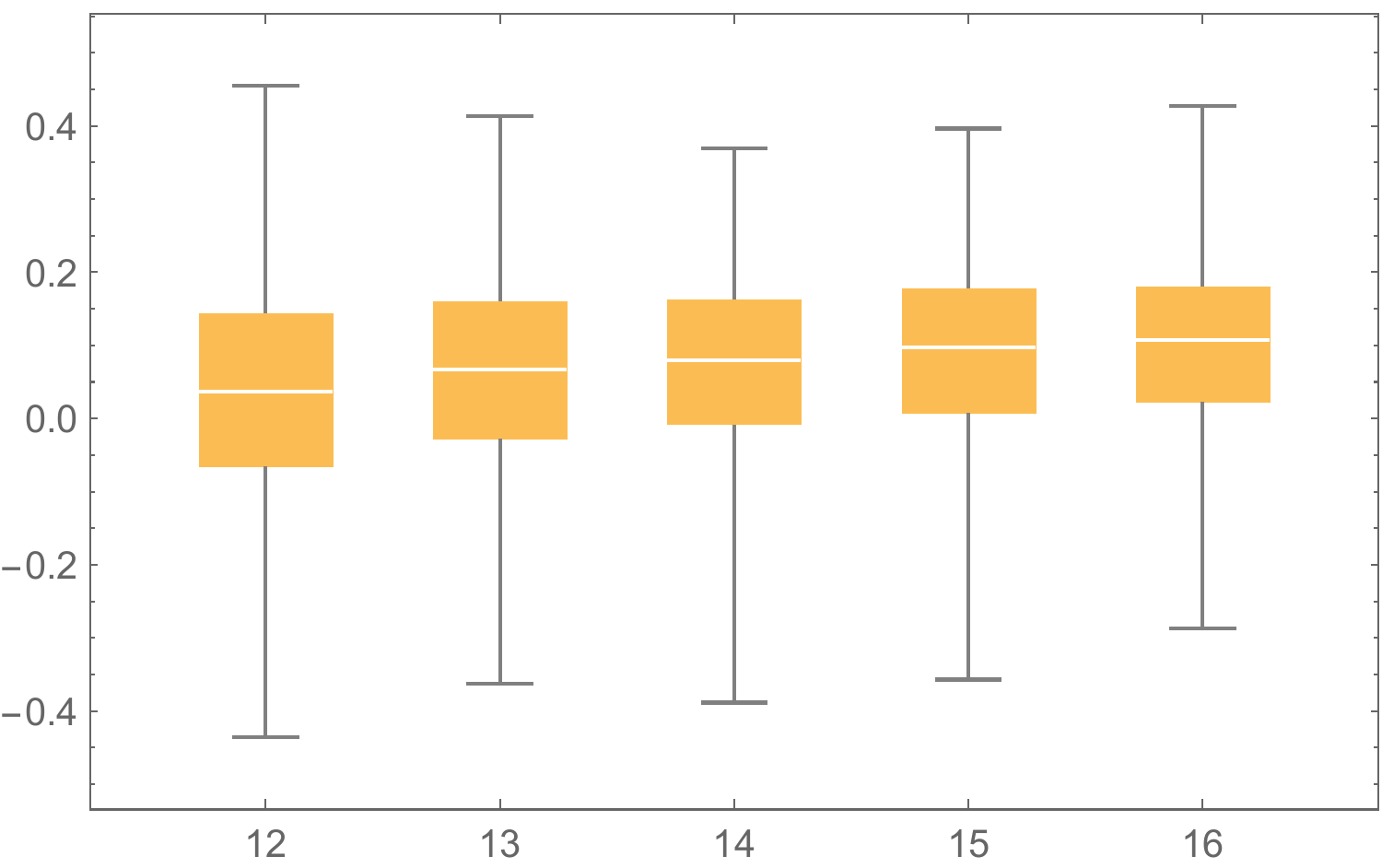}
	\quad		\includegraphics[width=8cm]{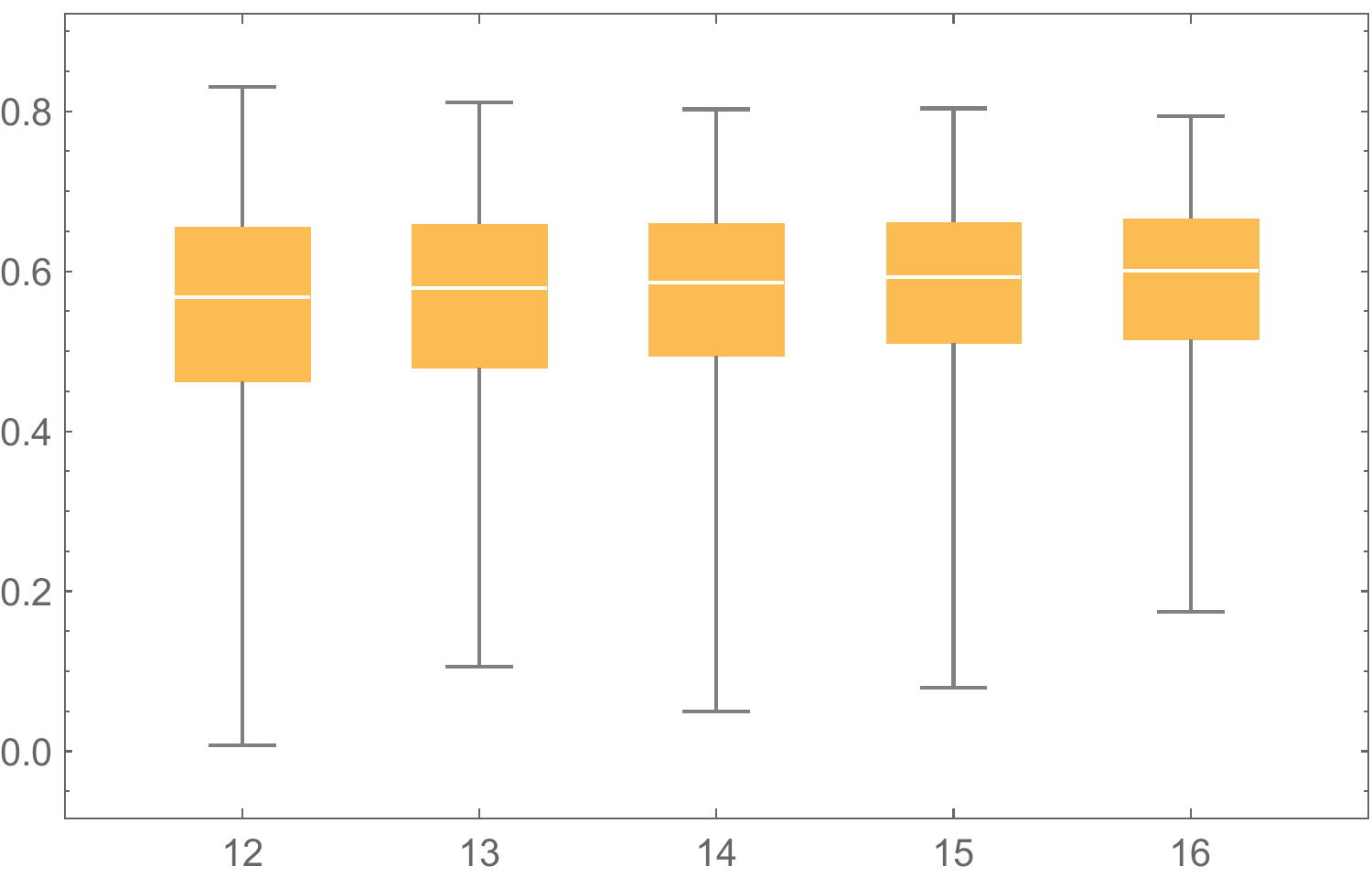}\hfill
	
	\includegraphics[width=8cm]{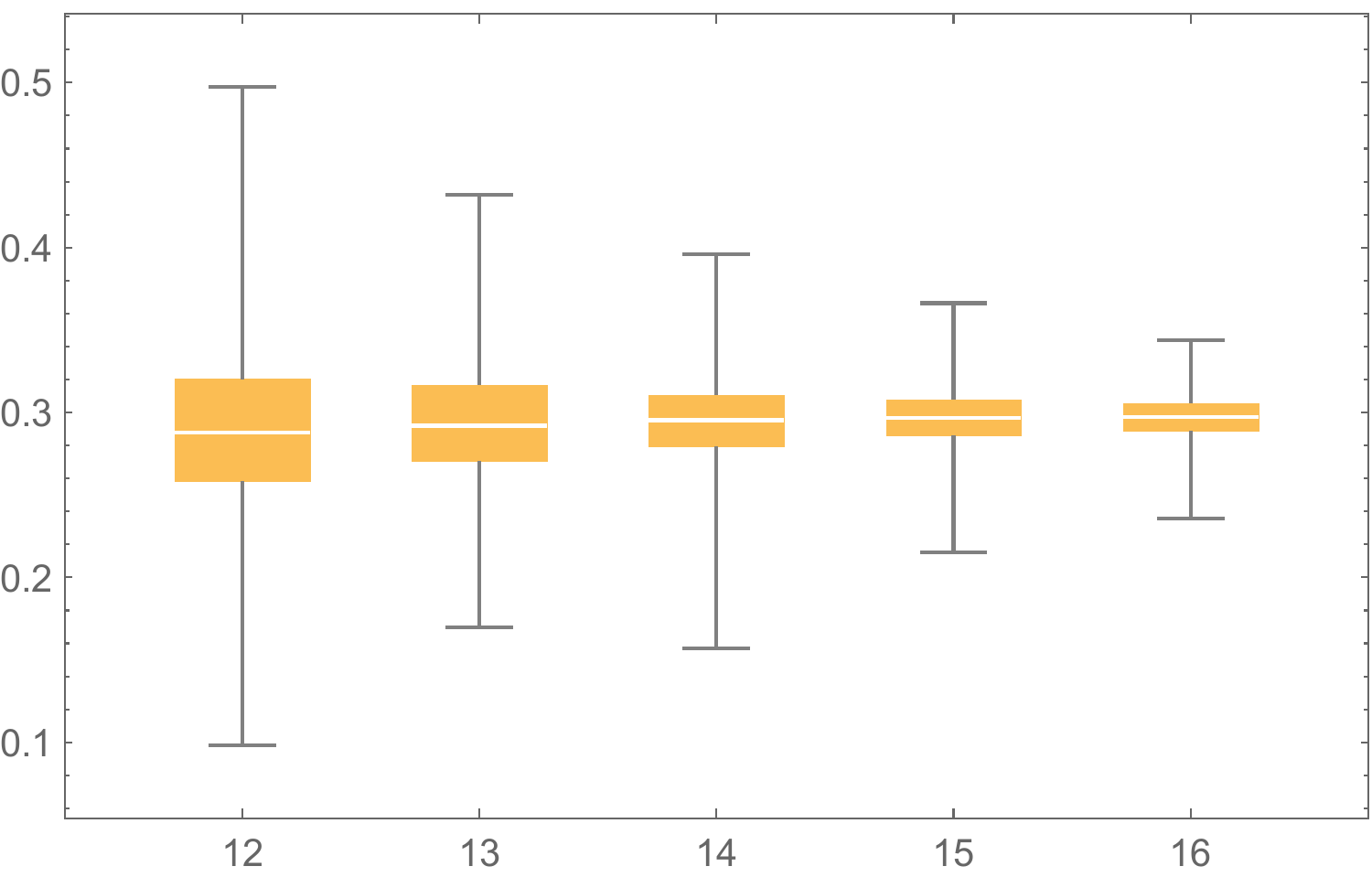}
	\quad		\includegraphics[width=8cm]{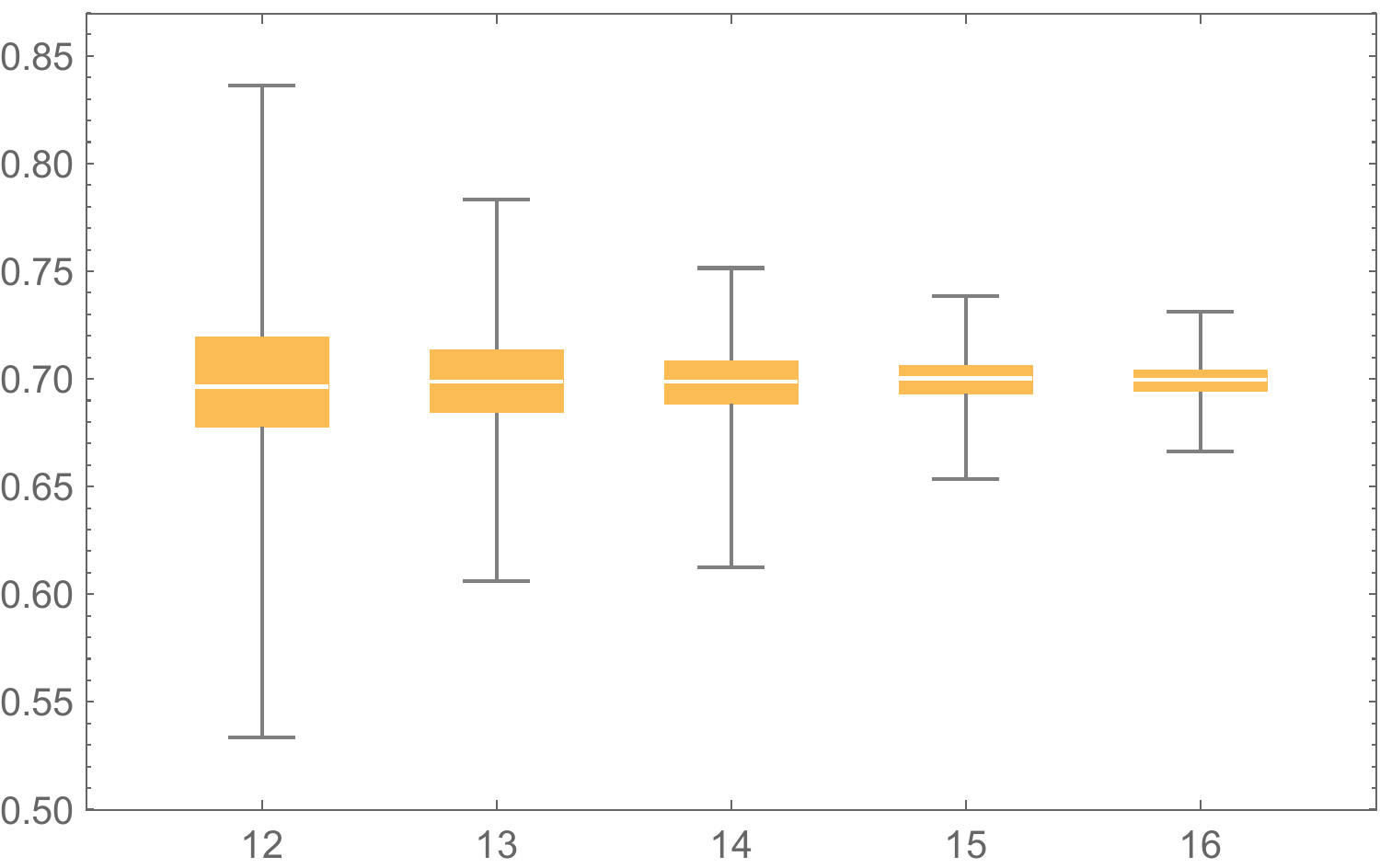}
	\caption{Box plots of the original estimates $\wh\cR_n(Y^{\sigma})$ (top) and the sequential scale estimates $\cR^s_n(Y^{\sigma})$ (bottom) for $n=12,\dots, 16$ based on  1,000 simulations   of the antiderivative of the exponential Ornstein--Uhlenbeck process \eqref{discrete OU Y eq}
   with $H = 0.3$ (left) and $H = 0.7$ (right). The other parameters are chosen as $x_0 = 0$, $\rho = 0.2$,  $\mu = 2$, $m = 3$  and $\alpha_k = 1$ for $k = 0,1,\dots, 3$. }
	\label{fig:Box Ori FOU}
\end{figure}

These results suggest that the sequential scale estimator provides an accurate estimate of the actual roughness exponent. However, the above simulation studies are based on fractional Brownian motion with drift. Next, we will further illustrate the robustness of our estimator $\cR^s_n$ by investigating its performance on processes that are not based on fractional Brownian motion. To this end, consider the Cauchy class introduced by Gneiting and Schlather \cite{Gneiting2004Hurst}, which consists of centered stationary Gaussian process $Z^{\alpha,\beta}$ for $\alpha,\beta \in (0,2)$ with continuous sample paths and covariance function 
\begin{equation*}
	\bE\left[Z^{\alpha,\beta}_sZ^{\alpha,\beta}_t\right] = \big(1 + |t-s|^\alpha\big)^{-\frac{\beta}{\alpha}}, \qquad s,t \ge 0.
\end{equation*}
It was shown by Barndorff-Nielsen et al.~\cite{Barndorff2009BSS} that $Z^{\alpha,\beta}$ admits the roughness exponent $R = \alpha/2$ with probability one. In the following, we take $N = n+5$ and simulate the sample paths of $Z^{\alpha,\beta}$ on the finer grid $\bT_n$, using the R-package  {\tt GeoModels} \cite{GM2024}. Then we approximate the antiderivative $\int_0^t Z^{\alpha,\beta}_s\,ds$ by 
\begin{equation*}
	Y^{z}_{k2^{-n-2}}:=2^{-N}\sum_{j = 1}^{2^{N-n-2}k}Z^{\alpha,\beta}_{j2^{-N}},\quad k=0,1,\dots, 2^{n+2},
\end{equation*}
and apply the estimator $\cR^s_n$ to each simulated sample path. The summary statistics are provided in the following table. 

	\begin{table}[H]
	\ra{1.1}
	\centering
	\begin{tabularx}{0.8\textwidth}{CCCCC}\toprule
		& \multicolumn{4}{c}{Summary statistics of the estimator $\cR^s_n$} \\
		\cmidrule{2-5}
		& \multicolumn{2}{c}{$\alpha = 0.4$} & \multicolumn{2}{c}{$\alpha = 0.8$}\\
		\cmidrule(lr){2-3}
		\cmidrule(lr){4-5} 
		\centering
		& Mean & Standard Deviation & Mean & Standard Deviation \\ \midrule
		$\beta = 0.4$ & 0.15914 & 0.053949 & 0.38983 & 0.053733 \\
		$\beta = 0.8$ & 0.14824 & 0.054001 & 0.38931 & 0.053732 \\
		$\beta = 1.2$ & 0.13735 & 0.054058 & 0.38879 & 0.053730 \\
		$\beta = 1.6$ & 0.12648 & 0.054121 & 0.38825 & 0.053729 \\
		\cmidrule{2-5}
			& \multicolumn{2}{c}{$\alpha = 1.2$} & \multicolumn{2}{c}{$\alpha = 1.6$}\\
		\cmidrule(lr){2-3}
		\cmidrule(lr){4-5} 
		\centering
		& Mean & Standard Deviation & Mean & Standard Deviation \\ \midrule
		$\beta = 0.4$ & 0.59436 & 0.053476 & 0.79680 & 0.053109 \\
		$\beta = 0.8$ & 0.59444 & 0.053475 & 0.79686 & 0.053109 \\
		$\beta = 1.2$ & 0.59451 & 0.053475 & 0.79693 & 0.053109 \\
		$\beta = 1.6$ & 0.59459 & 0.053474 & 0.79699 & 0.053109 \\
		\bottomrule
	\end{tabularx}
	\caption{Summary statistics of the sequential scale estimates $\cR^s_n(Y^z)$ with $n = 14$. The other parameters are chosen as $m = 2$ and $\alpha_k = 1$ for $k = 0,1,2$.}\label{Table Cauchy}
	\end{table}

As we can see from the above table, the sequential scale estimator $\cR^s_n$ provides an almost unbiased estimator for the roughness exponent of $Z^{\alpha,\beta}$ if $\alpha = 0.8, 1.2,1.6$. However, for the case $\alpha = 0.4$, we can see that the estimates $\cR^s_n(Y^z)$ deviate from its actual roughness exponent. This bias might result from the fact that the approximation error $|Y^z - \int_0^{\cdot} Z^{\alpha,\beta}_s\,ds|$ is relatively significant since $Z^{\alpha,\beta}$ generates rough sample paths. Nevertheless, the above numerical studies suggest that our estimator can be applied to a wide range of trajectories and is not limited to fractional Brownian processes with drift.

\section{Pathwise estimation}\label{Section Pathwise}

In this section, we formulate conditions on a single trajectory $x\in C[0,1]$ and its antiderivative $y(t)=\int_0^tx(s)\,ds$ under which the estimates $\wh\cR_n(y)$ converge to the roughness exponent of $x$. In \Cref{Section Proof}, we will then verify that these conditions are satisfied for the typical sample paths of fractional Brownian motion. The results in the present section are, hence, of independent interest in situations in which it is not clear whether a given trajectory $x$ arises from fractional Brownian motion (with drift).
We start by summarizing some key results and concepts from  \cite{HanSchiedMatrix, HanSchiedHurst} and also outline our rationale behind the specific form of the estimator $\wh\cR_n$.

\subsection{The rationale behind the estimator $\wh\cR_n$}\label{rationale section}

Recall that the Faber--Schauder functions are defined as 
	\begin{equation*}
		e_{-1,0}(t):= t, \quad e_{0,0}(t):= (\min\{t,1-t\})^{+}, \quad e_{m,k}(t):= 2^{-m/2}e_{0,0}(2^mt-k)
	\end{equation*}
	for $t \in \bR$, $m \in \bN$ and $k \in \bZ$. It is well known that the restrictions of the Faber--Schauder functions to $[0,1]$ form a Schauder basis for $C[0,1]$. More precisely, our function $x \in C[0,1]$ can be uniquely represented as the uniform limit $x=\lim_nx_n$, where 
	\begin{equation}\label{eq_Faber_expansion}
		x_n= x(0)+\left(x(1)-x(0)\right)e_{-1,0} + \sum_{m = 0}^{n-1}\sum_{k = 0}^{2^m-1}\theta_{m,k}e_{m,k},
	\end{equation}
	and the Faber--Schauder coefficients $\theta_{m,k}$ are given by 
	\begin{equation}\label{Faber--Schauder eq}
		\theta_{m,k} = 2^{m/2}\left(2x\Big(\frac{2k+1}{2^{m+1}}\Big)-x\Big(\frac{k}{2^m}\Big)-x\Big(\frac{k+1}{2^m}\Big)\right).
	\end{equation}
	As a matter of fact, it is easy to see that the function $x_n$ is simply the linear interpolation of $x$ based on the supporting grid $\bT_n=\{k2^{-n}:k=0,\dots, 2^n\}$.

 In \cite{HanSchiedHurst}, we derived simple conditions under which the trajectory $x$ admits a roughness exponent ${R}\in[0,1]$ and also suggested a way in which ${R}$ can be estimated from discrete observations of $x$. Specifically, it follows from Theorem 2.4 and Proposition 4.1 in \cite{HanSchiedHurst} that if the Faber--Schauder coefficients satisfy the so-called reverse Jensen condition (see Definition 2.3 in \cite{HanSchiedHurst}) and the sequence
\begin{equation}\label{cH eq}
\wh R^*_n(x):=1-\frac1{n}\log_2\sqrt{\sum_{k=0}^{2^n-1}\theta_{n,k}^2}
\end{equation}
converges to a finite limit $R$, then $x$ admits the roughness exponent $R$.

Note that it is assumed in \cite{HanSchiedHurst} that the trajectory $x$ can be observed directly. This, however, is not the case in the context of our present paper, where 
 $x$ is the (squared) volatility in a stochastic volatility model. So let us suppose now
 that we can only observe  the values of the antiderivative $y(t)=\int_0^tx(s)\,ds$ takes on the supporting grid $\bT_{n+2}$. If we can interpolate the data points $\{y(t):t\in\bT_{n+2}\}$ by means of a piecewise quadratic function $y_{n+2}\in C^1(\bR)$, then its derivative $y'_{n+2}$ will be a continuous and piecewise linear function with supporting grid $\bT_{n+1}$ and hence  representable in the form 
\begin{equation}\label{estimated FS coefficients eq}
y'_{n+2}=\hat x_0+\wh\theta_{-1,0}e_{-1,0} + \sum_{m = 0}^{n+1}\sum_{k = 0}^{2^m-1}\wh\theta_{m,k}e_{m,k}
\end{equation}
for some initial value $\hat x_0$ and certain coefficients $\wh\theta_{m,k}$. Such a piecewise quadratic $C^1$-interpolation $y_{n+2}$ exists in the form of the standard quadratic spline interpolation. Unfortunately, though, it is well known that  quadratic spline interpolation suffers  some  serious drawbacks: 
\begin{itemize}
\item the initial value $\hat x_0$ is not uniquely determined by the given data $\{y(t):t\in\bT_{n+2}\}$;
\item the values $y_{n+2}(t)$ depend in a highly sensitive manner on the choice of $\hat x_0$;
\item the values $y_{n+2}(s)$ depend in a nonlocal way on the given data  $\{y(t):t\in\bT_{n+2}\}$, i.e., altering one data point $y(t)$ may affect the value   $y_{n+2}(s)$ also if $s$ is located far away from $t$. 
\end{itemize}
In \cite{HanSchiedMatrix}, we investigate the analytical properties of the estimated 
Faber--Schauder coefficients $\wh\theta_{m,k}$ defined in \eqref{estimated FS coefficients eq}. It turns out that,  when looking at quadratic spline interpolation through the lens of these coefficients, a miracle occurs.
To see what happens, let us recall from \cite[Theorem 2.1]{HanSchiedMatrix} the formula for the Faber--Schauder coefficients  of $y'_{n+2}$ for the generations $m=0,\dots, n$ and for generation $n+1$,
\begin{align}\scalemath{0.85}{
\wh\theta_{m,k}}&= \scalemath{0.85}{2^{n+m/2+3}\sum_{j = 1}^{2^{n+1-m}}(-1)^j\left(y\Big(\frac{k}{2^m}+\frac{j}{2^{n+2}}\Big)-y\Big(\frac{k}{2^m}+\frac{j-1}{2^{n+2}}\Big)+y\Big(\frac{k+1}{2^m}-\frac{j-1}{2^{n+2}}\Big)-y\Big(\frac{k+1}{2^m}-\frac{j}{2^{n+2}}\Big)\right)},\label{eq wh vartheta}\\
\scalemath{0.85}{ \wh\theta_{n+1,k} }&=\scalemath{0.85}{-2^{(n+1)/2+2}\hat x_0 -2^{3(n+1)/2+4}\sum_{j = 1}^{2k}(-1)^{j}\left(y\Big(\frac{j}{2^{n+2}}\Big)-y\Big(\frac{j-1}{2^{n+2}}\Big)\right)}\label{eq vartheta n+1}
 \\&\quad\scalemath{0.85}{+3 \cdot 2^{3(n+1)/2+2} \left(y\Big(\frac{2k+1}{2^{n+2}}\Big)-y\Big(\frac{2k}{2^{n+2}}\Big)\right)-2^{3(n+1)/2+2} \left(y\Big(\frac{2k+2}{2^{n+2}}\Big)-y\Big(\frac{2k+1}{2^{n+2}}\Big)\right).}\nonumber
\end{align}
As one can see immediately from those formulas,  the coefficients in generations $m=0,\dots n$ are independent of $\wh x_0$, whereas the coefficients in generation $n+1$ contain the additive term $-2^{(n+1)/2+2}\hat x_0$, which translates any error made in estimating $\hat x_0$ into a $2^{(n+1)/2+2}$-fold error for each final-generation coefficient. 
Moreover, for  $m=0,\dots n$, each $\wh\theta_{m,k}$ depends only on those data points $y(t)$ for which~$t$ belongs to the closure of the support of the corresponding wavelet function $e_{m,k}$. Thus, the entire nonlocality of the function $y_{n+2}$ arises from the coefficients in generation $n+1$, while the coefficients of all lower generations depend locally on the given data. We refer to \cite[Figure 2]{HanSchiedMatrix} for an illustration.

The main results in \cite{HanSchiedMatrix} concern error bounds for the estimated Faber--Schauder coefficients $\wh\theta_{m,k}$. Specifically, we found that the $\ell_2$-norm of the combined errors in generations $m=0,\dots n$ is typically benign, whereas the error in the 
  final generation $m=n+1$ can be larger than a factor of size $\mathcal{O}(2^n)$ times the error of all previous generations combined. While the exact error bounds from \cite{HanSchiedMatrix} will not be needed in our present paper, the proof of \Cref{gamma vs sn lemma} will rely on an algebraic representation of the error terms obtained in \cite[Lemma 3.2]{HanSchiedMatrix} and recalled in \Cref{eq vector w} below.

The above-mentioned facts make it clear that the coefficients in generations $m=0,\dots n$ provide robust estimates for the corresponding true coefficients, while the estimates $\wh\theta_{n+1,k}$ are highly non-robust and should be discarded. It is now obvious that in estimating the roughness exponent of $x$ from the data $\{y(t):t\in\bT_{n+2}\}$, we should replace the true coefficients $\theta_{n,k}$ in our formula \eqref{cH eq} for $\wh R^*_n(x)$ with their estimates $\wh\theta_{n,k}$. It remains to note that $\wh\theta_{n,k}$ is in fact equal to $\vartheta_{n,k}$ defined in \eqref{eq vartheta}, so that we finally arrive at the rationale behind our estimator $\wh\cR_n$. 

The purpose of the following example is twofold. First, it serves as an illustration for the fact that our method is not limited to stochastic processes,  but that it works also in nonparametric situations of extreme 
model uncertainty, where only the trajectory is given and no information whatsoever on its distribution is available.  
Second, it illustrates in a very concise and transparent manner why the final generation of estimated Faber--Schauder functions must be excluded from the estimation process. 

\begin{example}\label{Takagi example} For ${R}\in(0,1]$, let  $x^{R}\in C[0,1]$ be the function with Faber--Schauder coefficients $\theta_{n,k} = 2^{n(1/2-{R})}$. These functions belong to the well-studied class of fractal Takagi--Landsberg functions. It was shown in \cite[Theorem 2.1]{MishuraSchied2}  that $x^{R}$ has the roughness exponent ${R}$. Moreover, for $y^{R}(t) = \int_{0}^{t}x^{R}(s)\,ds$, it was shown in \cite[Example 2.3]{HanSchiedMatrix}  that the robust approximation \eqref{eq vartheta} based on discrete observations of $y^R$ recovers exactly the Faber--Schauder coefficients of $x^{R}$. That is, for $n \in \bN$ and  $0 \le k \le 2^n-1$, we have 
	$
	\vartheta_{n,k} = \theta_{n,k} = 2^{n(\frac{1}{2}-{R})}
	$.
	It follows that
		\begin{equation*}
	\wh\cR_n(y^R)=  1 - \frac{1}{n}\log_2\sqrt{\sum_{k = 0}^{2^{n}-1}2^{(1-2{R})n}} = 1 - \frac{1}{n}\log_22^{(1-{R})n} = {R}.
	\end{equation*}
	Hence, the estimator $\wh\cR_n$ is not only consistent but also exact in the sense that it gives the correct value ${R}$ for every finite $n$.

Now we replace $\vartheta_{n,k}=\wh\theta_{n,k}$ with the final-generation estimates $\wh\theta_{n+1,k}$ as defined in \eqref{eq vartheta n+1}. Note that this requires the choice of an initial value $\hat x_0$. The corresponding estimator is given by
\begin{equation*}
	\wh M_n(y^{R}): = 1 - \frac{1}{n+1}\log_2 \sqrt{\sum_{k = 0}^{2^{n+1}-1} \wh\theta_{n+1,k}^2}.
	\end{equation*}
We get from  \cite[Example 3.2]{HanSchiedMatrix} that for ${R}<1/2$,
	\begin{equation*}
	\wh\theta_{n+1,k} = -2^{(n+1)/2+2}\hat x_0 + \sum_{m = n+1}^{\infty} 2^{m(\frac{1}{2}-{R})} = -2^{(n+1)/2+2}\hat x_0+ \frac{2^{(n+1)(\frac{1}{2}-{R})}}{1 - 2^{\frac{1}{2}-{R}}}.
	\end{equation*}
Hence,
	\begin{equation*}
	\sqrt{\sum_{k = 0}^{2^{n+1}-1} \wh\theta_{n+1,k}^2 }= 2^{n+1}\left|-4\hat x_0+ \frac{2^{-(n+1){R}}}{1 - 2^{\frac{1}{2}-{R}}}\right| .	\end{equation*}
	It follows that
		$$\lim_{n \ua \infty} \wh M_n(y^{R})= \begin{cases}0&\text{if $\hat x_0\neq0$,}\\
	{R}&\text{if $\hat x_0=0$.}
	\end{cases}
$$  
This shows that the estimator $\wh M_n$ is extremely sensitive with respect to the estimate $\hat x_0$ of the exact initial value $x(0)$, which in typical applications will be unknown. Even in the case that $x(0)$ is  known, the correct value ${R}$ is only obtained asymptotically, whereas $\wh\cR_n(y^{R})={R}$ for all finite $n$. These observations illustrate once again why we deliberately discard the final generation $\wh\theta_{n+1,k}$ of estimated Faber--Schauder coefficients.

\end{example}

\subsection{Pathwise consistency of $\wh\cR_n(y)$}\label{pathwise section}

Let us fix $x \in C[0,1]$ and denote by  $\theta_{m,k}$  its  Faber--Schauder coefficients \eqref{Faber--Schauder eq}. As before, we denote by $y(t)=\int_0^tx(s)\,ds$ the antiderivative of $x$ and by $\vartheta_{n,k}$ the coefficients defined in \eqref{eq vartheta}. To be consistent with   \cite{HanSchiedMatrix}, we introduce the following vector notation,
	\begin{equation}\label{eq vector}
\bar{\bm \theta}_n := \big(\theta_{n,0}, \theta_{n,1} \cdots, \theta_{n,2^{n}}\big)^\top \in \bR^{2^{n}}\quad \text{and} \quad \bar{\bm \vartheta}_n = \big(\vartheta_{n,0}, \vartheta_{n,1} \cdots, \vartheta_{n,2^{n}-1}\big)^\top \in \bR^{2^{n}},
\end{equation}
Then the estimators $\wh R^*_n$ and $\wh\cR_n$ defined in \eqref{cH eq} and \eqref{wh Rn eq} can be written as
\begin{equation}\label{Rn* and cRn eq}
\wh R_n^*(x)=1-\frac1{n}\log_2\|\bar{\bm \theta}_n\|_{\ell_2}\quad\text{and}\quad\wh\cR_n(y)=1-\frac1{n}\log_2\|\bar{\bm \vartheta}_n\|_{\ell_2}.
\end{equation}
Following  \cite{HanSchiedMatrix}, we introduce the column vector $\bm z_n:= (z^{(n)}_i)_{1 \le i \le 2^n}$ with components
\begin{equation}\label{eq def zn}
z^{(n)}_i = 2^{3n/2}\sum_{m = n}^{\infty}2^{-3m/2}\sum_{k = 0}^{2^{m-n}-1}\theta_{m,k+2^{m-n}(i-1)} = \sum_{m = 0}^{\infty}2^{-3m/2}\sum_{k = 0}^{2^{m}-1}\theta_{m+n,k+2^{m}(i-1)} \quad \text{for} \quad 1 \le i \le 2^n.
\end{equation} 
As observed in \cite{HanSchiedMatrix}, the infinite series in \eqref{eq def zn} converges absolutely if $x$ satisfies a H\"older condition, and for simplicity, we are henceforth going to make this assumption.
For $1 \le i,j \le 2^n$, we let furthermore
\begin{equation}\label{etaij def}
\bm \eta_{i,j} = \begin{cases}
\bm r &\quad \text{for} \quad 1 \le i = j \le 2^n,\\
\bm 0_{1 \times 4}  &\quad \text{for} \quad 1 \le i \neq j \le 2^n,
\end{cases}
\end{equation}
where $\bm r := \frac14(-1,+1,+1,-1)$,   and $\bm 0_{m \times n}$ denotes the $m \times n$-dimensional zero matrix. Moreover, we denote 
\begin{equation}\label{eq Cm}
\mathscr{C}_n := \left[
\begin{array}{c c c c c}
\bm \eta_{1,1} & \bm \eta_{1,2} &\cdots  & \bm \eta_{1,2^n-1}  & \bm \eta_{1,2^n}\\
\bm \eta_{2,1} & \bm \eta_{2,2} &\cdots  & \bm \eta_{2,2^n-1}  & \bm \eta_{2,2^n}\\
\vdots & \vdots & \ddots& \vdots & \vdots \\
\bm \eta_{2^n,1} & \bm \eta_{2^n,2} &\cdots  & \bm \eta_{2^n,2^n-1}  & \bm \eta_{2^n,2^n}\\
\end{array}
\right] \in \bR^{2^n \times 2^{n+2}}.
\end{equation}
It was shown in \cite[Lemma 3.1 and  3.2]{HanSchiedMatrix} that the error between the true and estimated Faber--Schauder coefficients can be represented as follows,
\begin{equation}\label{eq vector w}
\bar{\bm\vartheta}_n-\bar{\bm\theta}_n=\bm w_n, \quad\text{where}\quad \bm w_n = \mathscr{C}_n \bm z_{n+2}\in\bR^{2^n}.
\end{equation}
Consider the following condition: 
\begin{equation}\label{(W)}
\text{There exists $\kappa\in\bR\setminus\{1\}$ such that}\quad \frac{\norm{\bm w_n}_{\ell_2}}{\|\bar{\bm \theta}_n\|_{\ell_2}}\lra \kappa\quad\text{as $n\ua\infty$.}
\end{equation}
We will see in \Cref{lemma fbm decaying} that condition \eqref{(W)} is $\bP$-a.s.~satisfied for fractional Brownian motion. 

\begin{lemma}\label{gamma vs sn lemma}
 Under condition \eqref{(W)}, there exist $n_0\in\bN$ and constants $0<\kappa_-\le\kappa_+<\infty$ such that 
\begin{equation}\label{gamma vs sn eq}
 \kappa_-\|\bar{\bm \theta}_n\|_{\ell_2}\le \norm{\bar{\bm \vartheta}_n}_{\ell_2} \le  \kappa_+\|\bar{\bm \theta}_n\|_{\ell_2}\quad\text{for all $n\ge n_0$.}
\end{equation}
\end{lemma}

\begin{proof} Let $\kappa$ be as in  \eqref{(W)}. Then, for any $\varepsilon < |\kappa - 1|/2$, there exists $n_\varepsilon \in \bN$ such that for $n \ge n_\varepsilon$, we have  $\norm{\bar{\bm\theta}_{n}}_{\ell_2}(\kappa - \varepsilon) < \norm{\bm w_n}_{\ell_2} < \norm{\bar{\bm\theta}_{n}}_{\ell_2}(\kappa + \varepsilon)$. Using the representation \eqref{eq vector w} and applying the triangle inequality gives
	\begin{equation*}
	 \norm{\bar{\bm \vartheta}_n}_{\ell_2} = \norm{\bar{\bm \theta}_n - \bm w_n}_{\ell_2}  \le \norm{\bar{\bm\theta}_{n}}_{\ell_2} + \norm{\bm w_{n}}_{\ell_2} \le (\kappa + \varepsilon + 1) \norm{\bar{\bm\theta}_{n}}_{\ell_2}.
	\end{equation*}
	On the other hand, we have 
	\begin{equation*}
	 \norm{\bar{\bm \vartheta}_n}_{\ell_2} = \norm{\bar{\bm \theta}_n - \bm w_n}_{\ell_2}  \ge \Big|\norm{\bar{\bm\theta}_{n}}_{\ell_2} - \norm{\bm w_{n}}_{\ell_2}\Big| \ge \big(|1 - \kappa - \varepsilon| \wedge |1 - \kappa + \varepsilon|\big) \norm{\bar{\bm\theta}_{n}}_{\ell_2}.
	\end{equation*}
	This completes the proof.
\end{proof}

By taking logarithms in \eqref{gamma vs sn eq}, \Cref{gamma vs sn lemma} immediately yields the following result.

\begin{proposition}\label{prop model free} Under condition \eqref{(W)}, the  limit $ \lim_{n}\wh\cR_n(y)$ exists if and only if  $ \lim_{n}\wh R^*_n(x)$ exists. Moreover, in this case, $\lim_{n}\wh\cR_n(y) = \lim_{n}\wh R^*_n(x)$.
	\end{proposition}

\begin{example}\label{Takagi example 2} In the situation of \Cref{Takagi example}, we have seen that $
	\vartheta_{n,k} = \theta_{n,k} = 2^{n(\frac{1}{2}-{R})}
	$. Applying the representation \eqref{eq vector w}  yields that $\bm w_n = \bm 0_{2^{n} \times 1}$. This implies
	$
	\lim_{n}\norm{\bm w_n}_{\ell_2}/\|\bar{\bm \theta}_n\|_{\ell_2} = 0
	$. That is, $x^{R}$ satisfies condition \eqref{(W)}.
	Hence,  \Cref{prop model free} applies, which gives an additional proof of the previously observed fact that $\lim_n\wh\cR_n(y)= {R}$.\end{example}

\begin{example}\label{Random Takagi example}
	Now we consider a variant of  \Cref{Takagi example} in which each Faber--Schauder coefficient is assigned a random sign. That is, we let $\theta_{n,k} = \zeta_{n,k}2^{n(\frac{1}{2}-R)}$, where $R\in(0,1)$ and $(\zeta_{m,k})$ is an i.i.d.~sequence of $\{-1,+1\}$-valued random variables with symmetric Bernoulli distribution. 	Then 	\begin{equation*}
		X^R_t = \sum_{m = 0}^{\infty}\theta_{m,k}e_{m,k}(t)
			\end{equation*}
	is a   stochastic process whose sample paths belong to the class of signed Takagi--Landsberg functions; see, e.g.,  \cite{Allaartflexible, HanSchiedZhang2,MishuraSchied2}.  It was shown in \cite[Theorem 2.1]{MishuraSchied2} that all sample paths of $X^R$ admit the roughness exponent $R$. Let $Y^R_t := \int_{0}^{t}X^R_s\,ds$. Then
\eqref{eq def zn} and \eqref{eq vector w} yield that
	\begin{equation*}
		w_{n,k} = 2^{(n+2)(\frac{1}{2} - R)-2}\sum_{m = 0}^{\infty}2^{-m(1+R)}\sum_{i = 0}^{2^{m+2}-1}\wt\zeta_{m+n+2,i + 2^{m+2}(k-1)}, \qquad 1 \le k \le 2^n,
	\end{equation*}
where each $\wt\zeta_{n,k}$ may differ from $\zeta_{n,k}$ by a deterministic sign. Clearly, $w_{n,k}$ is well-defined. It is also clear that $w_{n,k}$ is independent of $w_{n,j}$. We denote 
	$$\kappa_R := 2^{-4R}\frac{2^{1 + 2R}}{2^{1+2R}-1}\quad\text{and}\quad R^*:=\frac{\log \left(1+\sqrt{17}\right)}{2\log 2}-1.$$
	One checks that $\kappa_{R^*}=1$ and $\kappa_R \neq 1$ for any $R \neq R^\ast $. Then,
	\begin{equation*}
		\begin{split}
			\bE\left[w^2_{n,k}\right] &= 2^{(n+2)(1 - 2R)-4}\sum_{m = 0}^\infty 2^{-2m(1+R)}\sum_{i = 0}^{2^{m+2}-1}\bE\left[\wt\zeta^2_{m+n+2,i + 2^{m+2}(k-1)}\right]\\&= 2^{(n+2)(1 - 2R)-2}\sum_{m = 0}^{\infty}2^{-m(1+2R)} = 2^{(1-2R)n}\kappa_R.
		\end{split}
	\end{equation*}
	Furthermore, an analogous argument shows that for $n \in \bN$, we have $\var[w^2_{n,k}] = 2^{(2-4R)n}\sigma^2_R$ for some $\sigma^2_R > 0$ that does not depend on $n$. Hence, Chebyshev's inequality gives 
	\begin{equation*}
		\bP\left[\left|\frac{\norm{\bm w_n}^2_{\ell_2}}{\|\bar{\bm \theta}_n\|_{\ell_2}^2} - 	\bE\left[\frac{\norm{\bm w_n}^2_{\ell_2}}{\|\bar{\bm \theta}_n\|_{\ell_2}^2}\right] \right| \le \varepsilon\right] = 	\bP\left[\left|2^{(2R-2)n}\sum_{k = 0}^{2^n-1}w_{n,k}^2 - \kappa_R\right| \le \varepsilon\right]  \le 2^{-n}\varepsilon^{-2}.
	\end{equation*}
	Thus, it follows from the Borel--Cantelli lemma that ${\|\bar{\bm w}_n\|_{\ell_2}}/{\|\bar{\bm \theta}_n\|_{\ell_2}} \longrightarrow \sqrt{\kappa_R}$. Hence, condition \eqref{(W)} yields that $\lim_n \wh \cR_n(Y^R) \longrightarrow R$ with probability one for any $R \neq R^\ast$.
\end{example}

In the following, we investigate the roughness exponent of the sum and composition of functions. Furthermore, we also formulate conditions under which the estimator $\wh \cR_n$ is consistent for the sum and composition of functions, respectively. These conditions are crucial in the proof of \Cref{Cor Rough Bergomi}. To wit, let $z$ be a sample trajectory of the (re-scaled) Riemann--Liouville process $Z^{H, \nu}$, then the volatility process $\sigma_t$ can be replicated pathwise by $\exp(\gamma z_t + \frac{\gamma^2}{4H}t^{2H})$ for $t \in [0,1]$. Given that $\wh \cR_n$ is a consistent estimator for $z$, then it remains to verify the following conditions to show that $\wh \cR_n$ is also a consistent estimator for $\sigma$.

We first establish conditions on $x$ and $g$, under which $u:=g\circ x$ admits the same roughness exponent as $x$. To this end, we fix the following notation for Propositions \ref{thm roughness composition} and \ref{Prop tilde H z}.
\begin{equation}\label{u and v eq}
u(t)=g(x(t))\qquad\text{and}\qquad v(t)=\int_0^tu(s)\,ds=\int_0^tg(x(s))\,ds.
\end{equation}
	
\begin{proposition}\label{thm roughness composition}
	If $x$ admits the roughness exponent ${R}$, $g$ belongs to $C^1(\bR)$, and $g'$ is nonzero on the range $x([0,1])$ of $x$, then $u = g\circ x$ also admits the roughness exponent ${R}$.
\end{proposition}

Now we turn to the following question: Under which conditions do we have $\wh\cR_n(v)\to {R}$, where $v$ is as in \eqref{u and v eq}? The conditions we are going to introduce to answer this question are relatively strong. Nevertheless, they hold for the sample paths of fractional Brownian motion. In the following proposition, we need to work with various functions $x$, $y$, $u$ and $v$. For this reason, we will temporarily use a superscript to indicate from which function the Faber--Schauder coefficients will be computed from. That is, for any function $f$, we write
\begin{equation}\label{coefficients depending on f eq}
	\begin{split}
		\theta^f_{n,k}& = 2^{n/2}\left(2f\Big(\frac{2k+1}{2^{n+1}}\Big)-f\Big(\frac{k}{2^{n}}\Big) - f\Big(\frac{k+1}{2^{n}}\Big)\right), \\
		\vartheta^f_{n,k}&= 2^{3n/2+3}\left(f\Big(\frac{4k}{2^{n+2}}\Big)-2f\Big(\frac{4k+1}{2^{n+2}}\Big)+2f\Big(\frac{4k+3}{2^{n+2}}\Big)-f\Big(\frac{4k+4}{2^{n+2}}\Big)\right).
	\end{split}
\end{equation}
 With this notation, the coefficients $\vartheta_{n,k}$ in \eqref{eq vartheta} should be re-written as $\vartheta_{n,k}^y$.

\begin{proposition}\label{Prop tilde H z} Suppose there exists ${R}\in(0,1)$ such that the following conditions hold.
\begin{enumerate}
\item \label{Prop tilde H z part a}We have  
	\begin{equation}\label{Prop tilde H z a assumption eq}
	0< \liminf_{n \ua \infty}2^{n(2{R}-2)}\sum_{k = 0}^{2^n-1}\left(\vartheta^y_{n,k}\right)^2\le \limsup_{n \ua \infty}2^{n(2{R}-2)}\sum_{k = 0}^{2^{n}-1}\left(\vartheta^y_{n,k}\right)^2<\infty.
	\end{equation}
	
\item\label{Prop tilde H z part b}  The function $x$ is H\"older continuous with exponent $\alpha\in(2{R}/5,1]$. 
\end{enumerate}
Then, if $g\in C^2(\bR)$ is strictly monotone,  we have $\lim_n\wh\cR_n(v)={R}$.
\end{proposition}

Subsequently, we consider the following question: What is the roughness exponent of $u + x$, given that $x, u \in C[0,1]$ admit the roughness exponent $R_x$ and $R_u$ respectively? The answer is provided in the following general lemma, which concerns the roughness exponent of the sum of two continuous functions and which is of independent interest. 

\begin{proposition}\label{lemma roughness sum}
	Suppose that two functions $x,z \in C[0,1]$ admit the roughness exponent $R_x$ and $R_{z}$ respectively, such that $R_x < R_{z}$, then $x+z$ admits the roughness exponent $R_x$. 
\end{proposition}

Note that \Cref{lemma roughness sum} does not extend to the case $R_x=R_{z}$. Indeed, taking $z=-x$ yields a trivial counterexample.

Our final result in this section investigates the asymptotic behavior of the estimator $\wh \cR_n$ for the antiderivative of the sum of continuous functions. To this end, let $x,u \in C[0,1]$, and we denote the antiderivatives of $x$ and $u$ by $y$ and $v$, respectively. In the following proposition, we continue to use a superscript to indicate from which function the Faber--Schauder coefficients will be computed from. In addition, for $f \in C[0,1]$,  $s \ge 0$ and $p \ge 1$, we denote the $p^{\rm th}$ variation of the function $t \mapsto f(s+t)$ along the $n^{\rm th}$ dyadic partition for given $s \ge 0$ by
\begin{equation}\label{eq p variation s}
	\<f\>^{(p)}_n(s):= \sum_{k = 0}^{2^n-1}\left|f\left(\frac{k+1}{2^{n}}+s\right)-f\left(\frac{k}{2^{n}}+s\right)\right|^p,
\end{equation}
where we put $f(t):= f(t \wedge 1)$ for all $t \ge 1$ to avoid undefined arguments of functions.
The following lemma establishes sufficient conditions under which $\wh \cR_n$ is a consistent estimator for the sum of continuous functions.

\begin{proposition}\label{lemma add}
	Suppose that there exists $R \in (0,1)$ such that the following conditions hold.
	\begin{enumerate}
		\item \label{item a lemma add} We have 
		\begin{equation}\label{eq y vartheta}
			0 < \liminf_{n \ua \infty} 2^{n(2R-2)} \sum_{k = 0}^{2^n-1} \left(\vartheta^y_{n,k}\right)^2 \le \limsup_{n \ua \infty} 2^{n(2R-2)} \sum_{k = 0}^{2^n-1} \left(\vartheta^y_{n,k}\right)^2 < \infty.
		\end{equation}
		\item \label{item b lemma add}The function $u$ admits the following asymptotic behavior:  
		\begin{equation}\label{eq u QV}
			\lim_{n \ua \infty} 2^{n(2R-1)}\int_0^{1}\<u\>^{(2)}_n\left(\frac{s}{2^n}\right)\,ds= 0.
		\end{equation}
	\end{enumerate}
	Then, we have
	\begin{equation}\label{eq target vartheta}
		0 < \liminf_{n \ua \infty} 2^{n(2R-2)} \sum_{k = 0}^{2^n-1} \left(\vartheta^{y + v}_{n,k}\right)^2 \le \limsup_{n \ua \infty} 2^{n(2R-2)} \sum_{k = 0}^{2^n-1} \left(\vartheta^{y+v}_{n,k}\right)^2 < \infty,
	\end{equation}
	and thus, $\lim_n \wh \cR_n(y+v) = R$. 
\end{proposition}

\section{Proofs}\label{Section Proof}
\subsection{Proof of results in \Cref{pathwise section}}

We begin with proving \Cref{thm roughness composition} and \Cref{Prop tilde H z}, and we fix the notation \eqref{u and v eq} throughout those proofs.

\begin{proof}[Proof of \Cref{thm roughness composition}]	For any $p > 0$, the mean value theorem and the intermediate value theorem yield numbers $\tau_{n,k} \in [k2^{-n},(k+1)2^{-n}]$ such that 
	\begin{equation}\label{eq varitaion u}
		\begin{split}
			\<u\>^{(p)}_n& =    \sum_{k = 0}^{2^n-1}\left|g'(x(\tau_{n,k}))\Big(x\Big(\frac{k+1}{2^n}\Big)-x\Big(\frac{k}{2^n}\Big)\Big)\right|^p	 = \sum_{k = 0}^{2^n-1}\left|g'(x(\tau_{n,k}))\right|^p\left|x\Big(\frac{k+1}{2^n}\Big)-x\Big(\frac{k}{2^n}\Big)\right|^p,
		\end{split}
	\end{equation}
	where  the notation $\<u\>^{(p)}_n$ was introduced in \eqref{eq p variation}.
	Since $g'$ is continuous and nonzero, there are constants $0<c_-<c_+<\infty$ such that $c_-\le|g'(x(t))|\le c_+$ for all $t\in[0,1]$. Hence, $c_-^p	 \<x\>^{(p)}_n\le  \<u\>^{(p)}_n\le c_+^p \<x\>^{(p)}_n$ holds for all $n$. Passing to the limit $n\ua\infty$ for $p>1/{R}$ and $p<1/{R}$ yields the result.
\end{proof}

\begin{proof}[Proof of \Cref{Prop tilde H z}]
	In this proof, we will work with the actual and estimated Faber--Schauder coefficients of the functions $x$,  $u=g\circ x$, and their respective antiderivatives $y$ and $v$. Here, we will use a superscript as in \eqref{coefficients depending on f eq} to indicate from which function the Faber--Schauder coefficients will be computed. Our goal in this proof is to show that \eqref{Prop tilde H z a assumption eq} carries over to the coefficients $\vartheta^v_{n,k}$.
	That is, 
	\begin{equation}\label{Prop tilde H z first goal eq}
		0<\liminf_{n \ua \infty}2^{n(2{R}-2)}\sum_{k = 0}^{2^{n}-1}\big(\vartheta^v_{n,k}\big)^2\le \limsup_{n \ua \infty}2^{n(2{R}-2)}\sum_{k = 0}^{2^{n}-1}\big(\vartheta^v_{n,k}\big)^2<\infty.	\end{equation}
Once this has been achieved, taking logarithms in \eqref{Prop tilde H z first goal eq}, dividing by $2n$, and passing to the limit will then yield
	${R}-\wh\cR_n(v)\to0$, which is the assertion. 
	
	It remains to establish \eqref{Prop tilde H z first goal eq}. 
	Rewriting the second line in \eqref{coefficients depending on f eq}
	gives after a short computation that 
	\begin{equation}\label{eq FS y}
		\vartheta^f_{n,k} =   2^{n+5/2}\left(\theta^f_{n+1,2k+1}-\theta^f_{n+1,2k}\right).	\end{equation}
	Let us  introduce the notation $\theta_{m,k}^f(s):=\theta_{m.k}^{f(s+\cdot)}$. That is, $\theta_{m,k}^f(s)$ are the Faber--Schauder coefficients of the function $t\mapsto f(s+t)$ for given $s\ge0$. Here and in the sequel, we will avoid undefined arguments of functions in case $s+t>1$ by setting $f(t):= f(t \wedge 1)$ for $t \ge 1$. With this notation, 
	we get from \eqref{eq FS y} that for $f\in C^1[0,\infty)$,		\begin{equation}\label{eq zeta v}
		\vartheta^f_{n,k} = 2^{n+5/2}\int_{0}^{2^{-n-1}}\theta^{f'}_{n+1,2k}(s)\,ds .
	\end{equation}
	Applying the mean-value theorem and the intermediate value theorem yields certain intermediate times $\tau_{n+2,k}(s) \in [2^{-n-2}k+s, 2^{-n-2}(k+1)+s]$ such  that for $s \in [0,2^{-n-1}]$, 
	\begin{equation*}
		\begin{split}
			\theta^{u}_{n+1,2k}(s) &= 2^{(n+1)/2}\left(2u\Big(\frac{4k+1}{2^{n+2}}+s\Big)-u\Big(\frac{4k}{2^{n+2}}+s\Big)-u\Big(\frac{4k+2}{2^{n+2}}+s\Big)\right)\\&= 2^{(n+1)/2} g'\big(x(\tau_{n+2,4k}(s))\big)\left(x\Big(\frac{4k+1}{2^{n+2}}+s\Big)-x\Big(\frac{4k}{2^{n+2}}+s\Big)\right) \\
			&+2^{(n+1)/2} g'\big(x(\tau_{n+2,4k+1}(s))\big)\left(x\Big(\frac{4k+1}{2^{n+2}}+s\Big)-x\Big(\frac{4k+2}{2^{n+2}}+s\Big)\right),\\
			&=2^{(n+1)/2}\scalemath{0.9}{\left(\frac{g'\big(x(\tau_{n+2,4k}(s))\big)+g'\big(x(\tau_{n+2,4k+1}(s))\big)}{2}\left(2x\Big(\frac{4k+1}{2^{n+2}}+s\Big)-x\Big(\frac{4k}{2^{n+2}}+s\Big)-x\Big(\frac{4k+2}{2^{n+2}}+s\Big)\right)\right)}\\&\quad+ 2^{(n+1)/2}\left(\frac{g'\big(x(\tau_{n+2,4k}(s))\big)-g'\big(x(\tau_{n+2,4k+1}(s))\big)}{2}\left(x\Big(\frac{4k+2}{2^{n+2}}+s\Big)-x\Big(\frac{4k}{2^{n+2}}+s\Big)\right)\right).
		\end{split}
	\end{equation*}
	The intermediate value theorem  and the mean-value theorem  also imply that there are intermediate times $\tau^\sharp_{n+1,2k}(s),\tau^\flat_{n+1,2k}(s) \in [2^{-n-1}2k+s,2^{-n-1}(2k+1)+s]$ such that 
	\begin{align*}
		\frac{1}{2}\left(g'\big(x(\tau_{n+2,4k}(s))\big)+g'\big(x(\tau_{n+2,4k+1}(s))\big)\right) &= g'\big(\tau^\sharp_{n+1,2k}(s)\big),\\
		\frac{1}{2}\left(g'\big(x(\tau_{n+2,4k}(s))\big)-g'\big(x(\tau_{n+2,4k+1}(s))\big)\right) &= \frac{1}{2}g''\big(\tau^\flat_{n+1,2k}(s)\big)\left(x(\tau_{n+2,4k}(s))-x(\tau_{n+2,4k+1}(s))\right).
	\end{align*} 
	With the shorthand notation
	\begin{equation*}
		\zeta^x_{n+1,2k}(s):= 2^{(n+1)/2}\left(x\Big(\frac{4k+2}{2^{n+2}}+s\Big)-x\Big(\frac{4k}{2^{n+2}}+s\Big)\right)\big(x(\tau_{n+2,4k}(s))-x(\tau_{n+2,4k+1}(s))\big),
	\end{equation*}
	we then have
	\begin{equation*}
		\theta^{u}_{n+1,2k}(s) =g'\big(\tau^\sharp_{n+1,2k}(s)\big)\theta^{x}_{n+1,2k}(s) + g''\big(\tau^\flat_{n+1,2k}(s)\big)\zeta^{x}_{n+1,2k}(s).
	\end{equation*}
	Plugging the preceding equation into \eqref{eq zeta v} and applying the mean value theorem for integrals yields intermediate times $\tau^\sharp_{n+1,k}, \tau^\flat_{n+1,k} \in [2^{-n-1}{k},2^{-n-1}(k+1)]$  that are independent of $s$ such that
	\begin{equation*}
		\begin{split}
			\vartheta^v_{n,k} &=  2^{n+5/2}g'\big(x(\tau^\sharp_{n+1,2k})\big)\int_{0}^{2^{-n-1}}\theta^{x}_{n+1,2k}(s)\,ds + 2^{n+5/2}g''\big(x(\tau^\flat_{n+1,2k})\big)\int_{0}^{2^{-n-1}} \zeta^{x}_{n+1,2k}(s)\,ds \\&= g'\big(x(\tau^\sharp_{n+1,2k})\big)\vartheta^y_{n,k} + 2^{n+5/2}g''\big(x(\tau^\flat_{n+1,2k})\big)\int_{0}^{2^{-n-1}} \zeta^{x}_{n+1,2k}(s)\,ds.
		\end{split}
	\end{equation*}
	Introducing the shorthand notation	\begin{equation*}
		\wt \zeta^x_{n+1,2k}:= 2^{n+5/2}\int_{0}^{2^{-n-1}} \zeta^{x}_{n+1,2k}(s)\,ds,
	\end{equation*}
	we can write 
	\begin{equation}\label{three terms eq}
		\begin{split}
			\big(	\vartheta^v_{n,k} \big)^2
			& = \big(g'\big(x(\tau^\sharp_{n+1,2k})\big)\big)^2\big(\vartheta^y_{n,k}\big)^2 + \left(g''\big(x(\tau^\flat_{n+1,2k})\big)\right)^2\big(\wt\zeta^x_{n+1,2k}\big)^2\\
			&\qquad + 2g'\big(x(\tau^\sharp_{n+1,2k})\big)g''\big(x(\tau^\flat_{n+1,2k})\big)\vartheta^y_{n,k}\wt\zeta^x_{n+1,2k}.	\end{split}
	\end{equation}
	For each of the three terms on the right,  we will now analyze its contribution to the quantities in~\eqref{Prop tilde H z first goal eq}. 
	The main contribution comes from the first term on the right.	Indeed, our assumptions on $g$ imply that there are constants $0<c_-\le c_+<\infty$ such that $c_-<(g'(x(t)))^2<c^+$ for all $t\in[0,1]$, and so 	
	$$c_-2^{n(2{R}-2)}\sum_{k = 0}^{2^{n}-1}\big(\vartheta^y_{n,k}\big)^2\le 2^{n(2{R}-2)}\sum_{k = 0}^{2^{n}-1}\big(\vartheta^v_{n,k}\big)^2\le c_+2^{n(2{R}-2)}\sum_{k = 0}^{2^{n}-1}\big(\vartheta^y_{n,k}\big)^2.
	$$
	This will establish \eqref{Prop tilde H z first goal eq} as soon as
	we have shown that the contributions of the two remaining terms in \eqref{three terms eq} are asymptotically negligible.	For the second term, we use the H\"older continuity of $x$ to get a constant $c_x$ for which
	\begin{equation*}
		|x(\tau_{n+2,4k}(s))-x(\tau_{n+2,4k+1}(s))| \le c_x|\tau_{n+2,4k}(s) - \tau_{n+2,4k+1}(s)|^\alpha \le c_x2^{-\alpha n}
	\end{equation*}
	Furthermore,  there exists $\kappa_x > 0$ such that $32 (g''(x(s)))^2 \le \kappa_x$ for all $s\in[0,1]$. Thus,
	\begin{equation*}
		\begin{split}
			&2^{(2{R}-2)n}\sum_{k = 0}^{2^{n}-1}\left(g''\big(x(\tau^\flat_{n+1,2k})\big)\right)^2\left(\wt\zeta_{n+1,2k}
			\right)^2 \\
			&= 2^{(2{R}-2)n}\sum_{k = 0}^{2^{n}-1}\left(g''\big(x(\tau^\flat_{n+1,2k})\big)\right)^2\left(2^{n+5/2}\int_{0}^{2^{-n-1}}\zeta^{x}_{n+1,2k}(s)\,ds\right)^2\\
			&\le \kappa_x 2^{2{R}n}\sum_{k = 0}^{2^{n}-1}\left(\int_{0}^{2^{-n-1}}\zeta^{x}_{n+1,2k}(s)\,ds\right)^2 \le \kappa_x 2^{2{R}n}\sum_{k = 0}^{2^{n}-1}2^{-n-1}\int_{0}^{2^{-n-1}}\left(\zeta^{x}_{n+1,2k}(s)\right)^2\,ds \\
			&\le \kappa_x 2^{(2{R}-1)n}\int_{0}^{2^{-n-1}}\sum_{k = 0}^{2^{n}-1}\left(x\Big(\frac{4k+2}{2^{n+2}}+s\Big)-x\Big(\frac{4k}{2^{n+2}}+s\Big)\right)^2\left(x(\tau_{n+2,4k}(s))-x(\tau_{n+2,4k+1}(s))\right)^2\,ds \\
			&= \kappa_x 2^{\alpha n}\int_{0}^{2^{-n-1}}2^{(2{R}-1-\alpha)n}\sum_{k = 0}^{2^{n}-1}\left(x\Big(\frac{2k+1}{2^{n+1}}+s\Big)-x\Big(\frac{2k}{2^{n+1}}+s\Big)\right)^2\big(c^2_x2^{-2\alpha n}\big)\,ds \\
			&\le \kappa_x c^2_x \int_{0}^{2^{-n-1}}\sup_{s \in [0,2^{-n-1}]}\left(2^{(2{R}-1-3\alpha)n}\sum_{k = 0}^{2^{n+1}-1}\left(x\Big(\frac{2k+1}{2^{n+1}}+s\Big)-x\Big(\frac{2k}{2^{n}}+s\Big)\right)^2\right)\,ds.	\end{split}
	\end{equation*}
	Moreover, \ref{Prop tilde H z part b} implies the  integrand in the final term converges to zero:
	\begin{equation}\label{Prop tilde H z aux claim}
		\lim_{n\ua\infty}\sup_{0\le s\le 2^{-n-1}}2^{(2{R}-1-3\alpha)n}\sum_{k = 0}^{2^{n+1}-1}\left(x\Big(\frac{2k+1}{2^{n+1}}+s\Big)-x\Big(\frac{2k}{2^{n}}+s\Big)\right)^2=0.
	\end{equation}
	Indeed, by the H\"older continuity of $x$, we can again use the  constant $c_x$ to get 
	\begin{align*}
		2^{(2{R}-1-3\alpha)n}\sum_{k = 0}^{2^{n+1}-1}\left(x\Big(\frac{2k+1}{2^{n+1}}+s\Big)-x\Big(\frac{2k}{2^{n}}+s\Big)\right)^2&\le 2^{(2{R}-1-3\alpha)n}\cdot 2^{n+1}\cdot c_x^22^{-2(n+1)\alpha};
	\end{align*}
	the right-hand side is equal to $c_x^2\cdot2^{1-2\alpha}\cdot2^{(2{R}-5\alpha)n}$, which converges to zero as $n\ua\infty$. Altogether, this shows that the contribution of the second term on the right-hand side of \eqref{three terms eq} is negligible.

	For the cross-product term on the rightmost side of \eqref{three terms eq}, we get from the Cauchy--Schwarz inequality, 
	\begin{equation*}
		\begin{split}
			\lefteqn{\lim_{n \ua \infty}2^{(2{R}-2)n}\sum_{k = 0}^{2^n-1}g'\big(x(\tau^\sharp_{n+1,2k})\big)g''\big(x(\tau^\flat_{n+1,2k})\big)\vartheta_{n,k}\wt\zeta^x_{n+1,2k}}
			\\&\le \sqrt{\lim_{n \ua \infty}2^{(2{R}-2)n}\sum_{k = 0}^{2^n-1}\left(g'\big(x(\tau^\sharp_{n+1,2k})\big)\right)^2\big(\vartheta_{n,k}\big)^2} \sqrt{\lim_{n \ua \infty}2^{(2{R}-2)n}\sum_{k = 0}^{2^n-1}\left(g''\big(x(\tau^\flat_{n+1,2k})\big)\right)^2\big(\wt\zeta^x_{n+1,2k}\big)^2} = 0.
		\end{split}
	\end{equation*}
	Altogether, \eqref{Prop tilde H z first goal eq} follows.\end{proof}

Throughout the proofs of \Cref{lemma roughness sum} and \Cref{lemma add}, we let $x,u \in C[0,1]$, and we denote their antiderivatives by $y$ and $v$ respectively.

\begin{proof}[Proof of \Cref{lemma roughness sum}] 
For simplicity, for $f \in C[0,1]$, let us denote 
	\begin{equation*}
		\Delta \bm f_n = \left(f\left(\frac{1}{2^n}\right)-f(0), f\left(\frac{2}{2^n}\right)-f\left(\frac{1}{2^n}\right), \cdots, f(1)-f\left(\frac{2^n-1}{2^n}\right)\right) \in \bR^{2^n}.
	\end{equation*}
	Then, for $n \in \bN$ and $p \ge 1$, we have $\<f\>^{(p)}_n = \norm{\Delta \bm f}_{\ell_p}^p$. Moreover, the Minkowski inequality yields that 
	\begin{equation*}
		\left(\left|\norm{\Delta \bm x_n}_{\ell_p} - \norm{\Delta \bm z_n}_{\ell_p}\right|\right)^p
		\le \<x +z\>^{(p)}_n = \norm{\Delta \bm x_n + \Delta \bm z_n}^p_{\ell_p} \le \left(\norm{\Delta \bm x_n}_{\ell_p} + \norm{\Delta \bm z_n}_{\ell_p}\right)^p. 
	\end{equation*}
	By definition, for $p > 1/{R_x} > 1/{R_z}$, we have $\lim_{n} \<x\>^{(p)}_n = \lim_{n}\<z\>^{(p)}_n = 0$. Hence,  $\lim_{n} \norm{\Delta \bm x_n}_{\ell_p} = \lim_{n \ua \infty}\norm{\Delta \bm z_n}_{\ell_p} = 0$.
	Thus, for $p > 1/{R_x}$, we have 
	\begin{equation*}
		\lim_{n \ua \infty}\<x + z\>^{(p)}_n \le \lim_{n \ua \infty}\left(\norm{\Delta \bm x_n}_{\ell_p} + \norm{\Delta \bm z_n}_{\ell_p}\right)^p = 0.
	\end{equation*}
	Next, for $1/{R_z} < p < 1/{R_x}$, we have 
	\begin{equation*}
		\lim_{n \ua \infty} \<x\>^{(p)}_n = \lim_{n \ua \infty}\norm{\Delta \bm x_n}_{\ell_p}^p = \infty \qquad \text{and} \qquad \lim_{n \ua \infty} \<z\>^{(p)}_n = \lim_{n \ua \infty}\norm{\Delta \bm z_n}_{\ell_p}^p = 0.
	\end{equation*}
	Hence, we have 
	\begin{equation}\label{eq fg lower bound}
		\lim_{n \ua \infty}\<x + z\>^{(p)}_n \ge \lim_{n \ua \infty} \left(\left|\norm{\Delta \bm x_n}_{\ell_p} - \norm{\Delta \bm z_n}_{\ell_p}\right|\right)^p = \infty.
	\end{equation}
	As the quantity $\<x+z\>^{(p)}_n$ is decreasing with respect to $p$ for sufficiently large $n$, therefore, the above identity \eqref{eq fg lower bound} naturally extends to all $p < {1}/{R_x}$. This completes the proof. 
\end{proof}

\begin{proof}[Proof of \Cref{lemma add}]
	 If \eqref{eq target vartheta} holds, taking logarithms in \eqref{eq target vartheta}, dividing by $2n$, and passing to the limit will yield $\lim_n \wh \cR_n(y+v)  = R$. Moreover, using a Cauchy--Schwarz-type argument as in the proof of \Cref{Prop tilde H z} and applying the fact $\vartheta^{y + v}_{n,k} = \vartheta^v_{n,k} + \vartheta^y_{n,k}$, it suffices to establish the following identity
	\begin{equation}\label{eq target id}
		\limsup_{n \ua \infty} 2^{n(2R-2)}\sum_{k = 0}^{2^n-1}\big(\vartheta^v_{n,k}\big)^2 = 0.
	\end{equation} To this end, we continue to use the notation $\theta_{m,k}^f(s):=\theta_{m.k}^{f(s+\cdot)}$. Applying Jensen's inequality to \eqref{eq FS y} yields that 
	\begin{equation*}
		\big(\vartheta^{v}_{n,k}\big)^2 = \left(2^{n + 5/2}\int_{0}^{2^{-n-1}}\theta^{u}_{n+1,2k}(s)\,ds\right)^2 = \left(2\sqrt{2}\int_0^{1}\theta^u_{n+1,2k}\left(\frac{s}{2^{n+1}}\right)\,ds\right)^2 \le 8 \int_0^1 \left(\theta^u_{n+1,2k}\left(\frac{s}{2^{n+1}}\right)\right)^2\,ds.
	\end{equation*}
	Thus, 
	\begin{equation*}
		\begin{split}
			2^{n(2R-2)}\sum_{k = 0}^{2^n-1}\big(\vartheta^{v}_{n,k}\big)^2 &\le 2^{n(2R-2)+3}\int_0^1 \sum_{k = 0}^{2^n-1}\left(\theta^u_{n+1,2k}\left(\frac{s}{2^{n+1}}\right)\right)^2\,ds\\ &\le 2^{n(2R-2)+3}\int_0^1 \sum_{k = 0}^{2^{n+1}-1}\left(\theta^u_{n+1,k}\left(\frac{s}{2^{n+1}}\right)\right)^2\,ds \\&= 2^{n(2R-1)+2}\int_0^1 \<u\>^{2}_{n+1}\left(\frac{s}{2^{n+1}}\right)\,ds,
		\end{split}
	\end{equation*}
	where the final equality follows from \cite[Proposition 2.1]{MishuraSchied}. Moreover, condition \eqref{eq u QV} yields that the rightmost term  in the above inequality converges to zero as $n \ua \infty$, which in turn leads to \eqref{eq target id}. 
\end{proof}

\subsection{Proof of \Cref{thm main rv}}\label{Section Proof fbm}

 Let us start by briefly outlining  the proof. We will first of prove \Cref{thm main rv} for the case $\xi = 0$ and $x_0 = 0$, so that $X = W^H$. To show that  $ g(W^H)$ admits the roughness exponent $H$ $\bP$-a.s.~is relatively straightforward. It was shown in \cite[Theorem 5.1]{HanSchiedHurst} that $W^H$ admits $\bP$-a.s.~the roughness exponent $H$. Hence, it follows from  \Cref{thm roughness composition} that the sample paths of $ g(W^H)$ also admit the roughness exponent $H$.  Next, two steps are needed to prove that $\lim_n \wh \cR_n(Y) = H$ for $Y_t = \int_0^t g(W^H_s)\,ds$. First, we need to establish condition \eqref{(W)} for the fractional Brownian motion $W^H$. Second,  we need to make sure that we may apply \Cref{Prop tilde H z}. This summarizes the steps to the proof of \Cref{thm main rv} for the case $X = W^H$. To extend the results to a general $X$ with $x_0 \in \bR$ and a non-trivial process $\xi$, we verify the conditions in Propositions \ref{lemma roughness sum} and \ref{lemma add} for the components $W^H$ and $x_0 + \int_0^\cdot \xi_s ds$.

Let us begin with stating and proving the following lemma, which will be needed for the proof of \Cref{lemma fbm decaying}. For $n,k \in \bN$, let us consider the vector $\bm z_{(n,k)} = (z^{(n,k)}_i) \in \bR^{2^n}$, where 
\begin{equation}\label{eq znk}
	z^{(n,k)}_i = 2^{3n/2}\sum_{m = n}^{n+k}2^{-3m/2}\sum_{j = 0}^{2^{m-n}-1}\theta_{m,j+2^{m-n}(i-1)} \quad \text{for} \quad 1 \le i \le 2^n.
\end{equation}
It is clear that the vector $\bm z_{(n,k)}$ is a truncated version of the vector $\bm z_n$ defined in \eqref{eq def zn}. Since each Faber--Schauder coefficient $\theta_{m,k}$ is  a linear combination  of the values $x(j2^{-n-k-1})$,  each $z^{(n,k)}_i$ must admit the following representation,
\begin{equation}\label{eq xi n}
	z^{(n,k)}_i = \sum_{j = 0}^{2^{n+k+1}}\xi^{(n,k,i)}_j x\Big(\frac{j}{2^{n+k+1}}\Big),
\end{equation}
for certain coefficients $\xi^{(n,k,i)}_j$. The following lemma computes the values of these coefficients. 

\begin{lemma}\label{Lemma xi n}
	We have 
	\begin{equation}\label{eq xi represent}
		\xi_{j}^{(n,k,i)}   = \begin{cases}
			0 &\quad \text{if} \quad j \le 2^{k+1}(i-1)-1 \quad \text{or} \quad j \ge 2^{k+1}i+1,\\
			2^{n/2}(2^{-k}-2) &\quad \text{if} \quad j =  2^{k+1}(i-1) \quad \text{or} \quad j = 2^{k+1}i,\\
			2^{1-k+n/2} &\quad \text{if} \quad  2^{k+1}(i-1)+1 \le j \le 2^{k}i-1.
		\end{cases}
	\end{equation}
\end{lemma}

\begin{proof}
	We fix $n \in \bN$ and $1 \le i \le 2^n$ and proceed by induction on $k \in \bN$. First, let us establish the base case $k = 0$. Then 
	\begin{equation}\label{eq initial step}
		z^{(n,0)}_{i} = \theta_{n,i-1} = 2^{n/2+1}x\Big(\frac{2i-1}{2^{n+1}}\Big) - 2^{n/2} x\Big(\frac{2i}{2^{n+1}}\Big)-2^{n/2}x\Big(\frac{2i-2}{2^{n+1}}\Big).
	\end{equation}
	Moreover, plugging $k = 0$ into \eqref{eq xi represent} yields that $\xi_{2i-2}^{(n,0,i)} = \xi_{2i}^{(n,0,i)} = -2^{n/2}$, $\xi_{2i-1}^{(n,0,i)} = 2^{n/2+1}$ and $\xi_{j}^{(n,0,i)} = 0$ otherwise. This proves our induction for the initial step $k = 0$. 
	
	Next, let us assume that \eqref{eq xi represent} holds for $k = m$ and subsequently prove that this identity also holds for $k = m + 1$. It follows from \eqref{eq znk} that
	\begin{equation}\label{eq induction step}
		z^{(n,m+1)}_{i} = z^{(n,m)}_{i} + 2^{n/2-m - 1}\sum_{j = 2^{m+1}(i-1)}^{2^{m + 1}i-1}\bigg(2x\Big(\frac{2j+1}{2^{n+m +2}}\Big)-x\Big(\frac{j}{2^{n+m +1}}\Big)-x\Big(\frac{j+1}{2^{n+m +1}}\Big)\bigg).
	\end{equation}
	For $2^{m + 1}(i-1) < j < 2^{m + 1}i-1$, the point $2^{-n-m-2}(2j+1)$ cannot be written in the form $\ell2^{-n-m-1}$ for some $\ell\in \bN_0$. Hence 
	$$
	\xi_{2j+1}^{(n,m+1,i)} = 2\cdot 2^{n/2-m-1} = 2^{n/2-m},
	$$
	as the term $x(2^{-n-m-2}(2j+1))$ does not appear in the linear combination \eqref{eq xi n} for $k = m$. Next, for $2^{k-n}(i-1) < j < 2^{k-n}i$, the point $2^{-n-m-2}2j = 2^{-n-m-1}j $ can be written in the form $\ell2^{-n-m-1}$ for some $\ell\in \bN_0$. It thus follows from \eqref{eq xi n} and \eqref{eq xi represent} that
	$$
	\xi^{(n,m+1,i)}_{2j} = \xi^{(n,m,i)}_{j} - 2 \cdot 2^{n/2-m-1} =  2^{n/2-m},$$
	as the term $x(2^{-n-m-1}j)$ contributes to the representation of $\bm z^{(n,m)}_i$ with $\xi^{(n,m,i)}_{j} = 2^{n/2-m+1}$. Moreover, for $j = 2^{k-n}(i-1)$ or $j = 2^{k-n}i$, we have 
	\begin{equation*}
		\xi_{j}^{(n,m+1,i)} = \xi_{j}^{(n,m,i)} - 2^{n/2-m-1} = 2^{n/2}(2^{-m}-2) - 2^{n/2-m-1} = 2^{n/2}(2^{-m-1}-2).
	\end{equation*}
	Last, for $j \le 2^{m+2}(i-1)-1$ or $j \ge 2^{m+2}i + 1$, the term $x(2^{-n-m-2})$ does not appear on right-hand side of \eqref{eq induction step}. Thus, we have $\xi_{j}^{(n,m+1,i)} = 0$. Comparing the above identities with \eqref{eq xi represent} proves the case for $k = m+1$.
\end{proof}

Using \Cref{Lemma xi n}, we can establish condition \eqref{(W)} for almost all sample paths of fractional Brownian motion, which is achieved in the following proposition.

\begin{proposition}\label{lemma fbm decaying}
	With probability one, the sample paths of fractional Brownian motion $W^H$ satisfy condition \eqref{(W)}.
\end{proposition}

To prove \Cref{lemma fbm decaying}, we need to obtain the asymptotic behavior of the $\norm{\bm w_n}_{\ell_2}$ associated with a fractional Brownian motion $W^H$. Let $\bar{\bm \vartheta}_{n}$, $\bm z_{(n,k)}$, and $\bm z_n$ be defined  as in \eqref{eq vartheta}, \eqref{eq znk}, and  \eqref{eq def zn} for the sample paths of $W^H$. It is clear that $\bm z_n$   is well defined, since the sample paths of $W^H$ satisfy a H\"older condition. Moreover, all three are  Gaussian random vectors. Our next lemma characterizes the covariance structure of the Gaussian vector $\bm z_n$.  To this end, consider the function $$g_H := h_1 + h_2 + h_3,$$
 where the $h_{i}: \bN_0 \rightarrow \bR$ are defined as follows,
\begin{equation}\label{eq delta}
\begin{split}
h_1({\varsigma}) &= -2\left(2{\varsigma}^{2H} + |{\varsigma}-1|^{2H}+|{\varsigma}+1|^{2H}\right) \\
h_2({\varsigma}) &= \begin{cases}
\frac{8}{2H+1}\left(({\varsigma}+1)^{2H+1} -({\varsigma}-1)^{2H+1}\right)
&\quad \text{for} \quad {\varsigma} \ge 1,\\
\frac{16}{2H+1}
&\quad \text{for} \quad {\varsigma} = 0,\\
\end{cases}\\ h_3({\varsigma})&=-\frac{8}{(2H+2)(2H+1)}\left(|{\varsigma}+1|^{2H+2}-2{\varsigma}^{2H+2}+|{\varsigma}-1|^{2H+2}\right).
\end{split}
\end{equation}
Furthermore, we introduce the  Toeplitz matrix $G_n := (g_H(|i-j|))_{1 \le i,j \le 2^n}$.

\begin{lemma}\label{lemma gamma ij}
	For each $n \in \bN$, the random vector $\bm z_{n}$ is a well-defined zero-mean Gaussian vector with covariance matrix 
	\begin{equation*}
	\Gamma_n = (\gamma^{(n)}_{i,j})_{1 \le i,j \le 2^n} =  2^{(1-2H)n}G_n.
	\end{equation*}
\end{lemma}

\begin{proof}
For $n,k \in \bN$, let us denote 
\begin{equation*}
	\Gamma_{(n,k)} = (\gamma^{(n,k)}_{i,j})_{1\le i,j \le 2^n} := \bE\left[ \bm{z}_{(n,k)}\bm{z}_{(n,k)}^\top\right].
\end{equation*}
It suffices to show that the components $\gamma^{(n,k)}_{i,j}$  converge to $\gamma^{(n)}_{i,j}$  as $k \ua \infty$. Moreover, by symmetry, it suffices to consider the case $j\ge i$. \Cref{Lemma xi n} yields
$$\gamma^{(n,k)}_{i,j} = \bE\big[z^{(n,k)}_{i}z^{(n,k)}_{j}\big]=\sum_{\tau_1 = 0}^{2^{n+k+1}}\sum_{\tau_2 = 0}^{2^{n+k+1}}\xi^{(n,k,i)}_{\tau_1}\xi^{(n,k,j)}_{\tau_2}\bE\Big[W^H_{\frac{\tau_1}{2^{n+k+1}}}\cdot W^H_{\frac{\tau_2}{2^{n+k+1}}}\Big].
$$
We also get from \Cref{Lemma xi n} that 
  $\sum_{\tau = 0}^{2^{n+k+1}}\xi_{\tau}^{(n,k,i)} = 0$ and $\xi_{\tau}^{(n,k,i)}   = 	0 $ for $ \tau \le 2^{k+1}(i-1)-1$ or $\tau \ge 2^{k+1}i+1$. Hence, for  $\varsigma := j-i \ge0$,
\begin{equation*}
\begin{split}
\gamma^{(n,k)}_{i,j}&=-\sum_{\tau_1 = 0}^{2^{n+k+1}}\sum_{\tau_2 = 0}^{2^{n+k+1}}\frac{\xi^{(n,k,i)}_{\tau_1}\xi^{(n,k,j)}_{\tau_2}}{2}\left|\frac{\tau_1-\tau_2}{2^{n+k+1}}\right|^{2H} =  -2^{-2Hn}\sum_{\tau_1 = 0}^{2^{k+1}}\sum_{\tau_2 =0}^{2^{k+1}}\frac{\xi^{(n,k,i)}_{\tau_1}\xi^{(n,k,j)}_{\tau_2}}{2}\left|\frac{\tau_1-\tau_2}{2^{k+1}} + \varsigma\right|^{2H}.
\end{split}
\end{equation*}
Using once again \eqref{eq xi represent} yields that 
\begin{equation}\label{eq expansion}
\gamma^{(n,k)}_{i,j} = 2^{(1-2H)n}\big(h_{1,k}(\varsigma)+h_{2,k}(\varsigma)+h_{3,k}(\varsigma)\big),
\end{equation}
where functions $h_{i,k}$ are defined as follows,
\begin{equation*}
\begin{split}
h_{1,k}(\varsigma) &= -\frac{(2^{-k}-2)^2}{2}\left(2|\varsigma|^{2H} + |\varsigma-1|^{2H}+|\varsigma+1|^{2H}\right),\\
h_{2,k}(\varsigma) &= 2^{-k}(2-2^{-k})\sum_{\tau = 1}^{2^{k+1}-1}\Big(\left|\frac{\tau}{2^{k+1}} + \varsigma\right|^{2H}+ \left|\frac{\tau}{2^{k+1}} + \varsigma-1\right|^{2H} + \left|\frac{-\tau}{2^{k+1}} + \varsigma\right|^{2H}+ \left|\frac{-\tau}{2^{k+1}} + \varsigma+1\right|^{2H}\Big),\\
h_{3,k}(\varsigma) &= -2^{1-2k}\sum_{\tau_1 = 1}^{2^{k+1}-1}\sum_{\tau_2 = 1}^{2^{k+1}-1}\left|\frac{\tau_1-\tau_2}{2^{k+1}} + \varsigma\right|^{2H}.
\end{split}
\end{equation*}
Let us first consider the case $\varsigma \ge 1$. Then,
\begin{equation}\label{eq delta 1}
\lim_{k \ua \infty}h_{1,k}(\varsigma) = -2\left(2\varsigma^{2H} + (\varsigma-1)^{2H}+(\varsigma+1)^{2H}\right).
\end{equation} 
Furthermore, 
\begin{equation*}
\begin{split}
\lim_{k \ua \infty}h_{2,k}(\varsigma)&= \lim_{k \ua \infty}2^{1-k}(2-2^{-k})\sum_{\tau = 1}^{2^{k+1}-1}\Big(\left|\frac{\tau}{2^{k+1}} + \varsigma\right|^{2H}+ \left|\frac{\tau}{2^{k+1}}+\varsigma-1\right|^{2H}\Big)\\&= 8\lim_{k \ua \infty}2^{-k-1}\sum_{\tau = 1}^{2^{k+1}-1}\Big(\left|\frac{\tau}{2^{k+1}} + \varsigma\right|^{2H}+ \left|\frac{\tau}{2^{k+1}}+\varsigma-1\right|^{2H}\Big)\\&= 8 \left(\int_{\varsigma}^{\varsigma+1}t^{2H}\,dt + \int_{\varsigma-1}^{\varsigma}t^{2H}\,dt\right)= \frac{8\left((\varsigma+1)^{2H+1} - (\varsigma-1)^{2H+1}\right)}{2H+1}
\end{split}
\end{equation*}
We also get in a similar way that
\begin{equation*}
\begin{split}
\lim_{k \ua \infty}h_{3,k}(\varsigma) &= -8\lim_{k \ua \infty}2^{-2-2k}\sum_{\tau_1 = 1}^{2^{k+1}-1}\sum_{\tau_2 = 1}^{2^{k+1}-1}\left(\frac{\tau_1-\tau_2}{2^{k+1}} + \varsigma\right)^{2H}=-8 \int_{0}^{1}\int_{0}^{1}\left(t-s+\varsigma\right)^{2H}\,ds\,dt\\& = -\frac{8\left((\varsigma+1)^{2H+2}-2\varsigma^{2H+2}+(\varsigma-1)^{2H+2}\right)}{(2H+2)(2H+1)}.
\end{split}
\end{equation*}   
For the case $\varsigma = 0$, $\lim_{k \ua \infty}h_{1,k}(0) = h_1(0)$ as in \eqref{eq delta 1}. Next, we have 
\begin{equation*}
\lim_{k \ua \infty}h_{2,k}(0) = \lim_{k \ua \infty}2^{2-k}(2-2^{-k})\sum_{\tau = 1}^{2^{k+1}-1}\left|\frac{\tau}{2^{k+1}}\right|^{2H} = 16\int_{0}^{1}t^{2H}\,dt = \frac{16}{2H+1}.
\end{equation*}
Finally, 
\begin{equation*}
\begin{split}
\lim_{k \ua \infty}h_{3,k}(0) &= \lim_{k \ua \infty}-2^{1-2k}\sum_{\tau_1 = 1}^{2^{k+1}-1}\sum_{\tau_2 = 1}^{2^{k+1}-1}\left|\frac{\tau_1-\tau_2}{2^{k+1}}\right|^{2H} = -8\int_{0}^{1}\int_{0}^{1}|t-s|^{2H}\,dsdt \\&= 16\int_{0}^{1}\int_{0}^{t}(t-s)^{2H}\,dsdt = \frac{16}{(2H+1)(2H+2)}.
\end{split}
\end{equation*}
Comparing the above equations with \eqref{eq delta} completes the proof.
\end{proof}

Our next lemma investigates the limit of $2^{n(H-1)}\norm{\bm z_n}_{\ell_2}$ as $n \ua \infty$ by applying a concentration inequality from \cite[Lemma 3.1]{BaudoinHairer}. In the form needed here, it states that if $\bm Z$ is a centered Gaussian random vector with covariance matrix $C$, $T:=\sqrt{\trace{C}}$, and $\gamma:=\sqrt{\norm{C}_2}$, then there exists a universal constant $\kappa$ independent of $C$ such that
\begin{equation}\label{concentration inequality}
\bP\Big[\big|\|\bm Z\|_{\ell_2}-T\big|\ge t\Big]\le \kappa\exp\Big(-\frac{t^2}{4\gamma^2}\Big)\qquad\text{for all $t>0$.}
\end{equation}

\begin{lemma}\label{lemma fbm strong}
	With probability one,
	\begin{equation*}
	\lim_{n \ua \infty}2^{n(H-1)}\norm{\bm z_n}_{\ell_2}=2\sqrt{\frac{1-H}{H+1}}.
	\end{equation*}
\end{lemma}

\begin{proof}
It follows from \eqref{eq delta} that	
\begin{equation}\label{eq Tn fbm}
	\begin{split}
	\sqrt{\trace{\Gamma_n}} &= \sqrt{\sum_{i = 1}^{2^n} \gamma^{(n)}_{i,i}} = 2^{(1-H)n}\sqrt{g_H(0)}=2^{(1-H)n}\sqrt{-4 + \frac{16}{2H+1} - \frac{16}{(2H+1)(2H+2)}} \\&=  2^{(1-H)n+1}\sqrt{\frac{1-H}{H+1}}.
	\end{split}
	\end{equation}
	Let $\|\cdot\|_p$ denote the $\ell_p$-induced operator norm. As shown in \Cref{lemma gamma ij}, the covariance matrix $\Gamma_n$ is a symmetric Toeplitz matrix and so $\gamma^{(n)}_{i,j} = \gamma^{(n)}_{j,i} = \gamma^{(n)}_{1,|j-i+1|}$. Hence, we have $\norm{\Gamma_n}_1 = \norm{\Gamma_n}_{\infty}$, and this gives
	\begin{equation}\label{eq Gamma 2 norm}
	\begin{split}
	\norm{\Gamma_n}_2 &\le \sqrt{\norm{\Gamma_n}_1\norm{\Gamma_n}_{\infty}} = \norm{\Gamma_n}_1 = \max_{1 \le j \le 2^n}\sum_{i = 1}^{2^n}|\gamma^{(n)}_{i,j}| = \max_{1 \le j \le 2^n}\sum_{i = 1}^{2^n}|\gamma^{(n)}_{1,|j-i+1|}|\\&\le 2\sum_{i=1}^{2^{n}}|\gamma^{(n)}_{1,i}| \le 2^{(1-2H)n+1}\sum_{\varsigma = 0}^{2^n}|g_H(\varsigma)|,
	\end{split}
	\end{equation}
	where the first inequality is a well-known bound for the spectral norm of a matrix; see, e.g., \cite[proof of Theorem 2.3]{Viitasaari}.
		In the next step, we will show that $g_H(\varsigma) = \mathcal{O}(\varsigma^{2H-4})$ as $\varsigma \ua \infty$.   For $\varsigma \ge 3$, Taylor expansion yields $u_1 \in (\varsigma-1, \varsigma)$ and $u_2 \in (\varsigma, \varsigma+1)$ such that
	\begin{equation*}
	\begin{split}
	(\varsigma - 1)^{2H}  &= \varsigma^{2H} + \sum_{i = 1}^{3}\frac{(-1)^i\prod_{j = 1}^{i}(2H-j+1)}{i!} \varsigma^{2H-i}+ \frac{\prod_{j = 1}^{4}(2H-j+1)}{4!}u_1^{2H-4},\\
	(\varsigma + 1)^{2H}  &= \varsigma^{2H} + \sum_{i = 1}^{3}\frac{\prod_{j = 1}^{i}(2H-j+1)}{i!}\varsigma^{2H-i} + \frac{\prod_{j = 1}^{4}(2H-j+1)}{4!}u_2^{2H-4}.
	\end{split}
	\end{equation*}
 Note that 
	\begin{equation*}
	\sum_{i = 1}^{3}\frac{\left((-1)^i + 1\right)\prod_{j = 1}^{i}(2H-j+1)}{i!} \varsigma^{2H-i}= 2H(2H-1)\varsigma^{2H-2},
	\end{equation*}
	and therefore, we have 
	\begin{equation}\label{eq taylor h1}
	h_1(\varsigma) = -2\left(4 \varsigma^{2H} + 2H(2H-1)\varsigma^{2H-2} + \frac{\prod_{j = 1}^{4}(2H-j+1)}{4!}(u_1^{2H-4}+u_2^{2H-4})\right).
	\end{equation}
	In the same way, we obtain
	\begin{align}
	h_2(\varsigma)  &= 8\left(2\varsigma^{2H} + \frac{2(2H)(2H-1)}{3!}\varsigma^{2H-2} +\frac{\prod_{j = 1}^{4}(2H-j+1)}{5!} (u_3^{2H-4}+u_4^{2H-4})\right),\label{eq taylor h2}\\
	h_3(\varsigma)  &= -8\left(\frac{2}{2!}\varsigma^{2H} + \frac{2(2H)(2H-1)}{4!}\varsigma^{2H-2} +\frac{\prod_{j = 1}^{4}(2H-j+1)}{6!}(u_5^{2H-4}+u_6^{2H-4}) \right),\label{eq taylor h3}
	\end{align}
	for some $u_3,u_5 \in (\varsigma-1, \varsigma)$ and $u_4, u_6 \in (\varsigma, \varsigma+1)$. Since $u_i \ge \varsigma - 1$,  we get $u_i^{2H-4} \le (\varsigma-1)^{2H-4}$. Summing up \eqref{eq taylor h1}, \eqref{eq taylor h2} and \eqref{eq taylor h3} yields that $g_H(\varsigma) = \mathcal{O}(\varsigma^{2H-4})$ as $\varsigma \ua \infty$, which with \eqref{eq Gamma 2 norm} implies that
	\begin{equation}\label{eq fbm variance}
	\norm{\Gamma_n}_2  \le 2^{(1-2H)n}\left(g_H(0)+g_H(1)+\sum_{\varsigma = 2}^{2^n}g_H(\varsigma)\right) = \mathcal{O}(2^{(1-2H)n}) \quad \text{for} \quad H \in (0,1).
	\end{equation}
	Therefore, for each $H \in (0,1)$, there exist $c_H\ge 0$ and $n_{c,H} \in \bN$ such that for $n \ge n_{c,H}$, we have $\norm{\Gamma_n}_2 \le c_H$. Thus, for $n \ge n_{c,H}$ and any given $\varepsilon > 0$, the concentration inequality \eqref{concentration inequality} gives	\begin{equation*}
	\begin{split}
	&\bP\left(\left|2^{n(H-1)}\norm{\bm z_n}_{\ell_2}- 2\sqrt{\frac{1-H}{H+1}}\right| \ge \varepsilon\right)= \bP\left(\left|\norm{\bm z_n}_{\ell_2}- 2^{(1-H)n+1}\sqrt{\frac{1-H}{H+1}}\right| \ge 2^{(1-H)n}\varepsilon\right)\\ &\le \kappa\exp\left(-\frac{2^{n(2-2H)}\varepsilon^2}{\norm{\Gamma_n}_2}\right)= \kappa\exp(-c^{-1}_H2^{n}\varepsilon^2)
	\end{split}
	\end{equation*}
	The latter expression is summable in $n$ for every $\eps>0$, and so a Borel--Cantelli argument yields that $2^{n(H-1)}\norm{\bm z_n}_{\ell_2} \rightarrow 2\sqrt{\frac{1-H}{H+1}}$ with probability one as $n \ua \infty$.\end{proof}

In the following lemma, we will derive the asymptotic behavior of the norms of $\bm w_n$ defined in~\eqref{eq vector w}.
 
\begin{lemma}\label{lemma fbm wn}
	With probability one, we have 
	\begin{equation*}
	\lim_{n \ua \infty}2^{n(H-1)}\norm{\bm w_n}_{\ell_2}=2^{-2H}\sqrt{\alpha(H)}, 
	\end{equation*}
	where $\alpha(H) = g_H(0) - \frac{1}{2}g_H(1) - g_H(2) +  \frac{1}{2}g_H(3)$.
\end{lemma}

\begin{proof}
	Let us denote the covariance matrix of $\bm w_n$ by $\Phi_n := (\phi^{(n)}_{i,j})_{i,j = 1}^{2^{n}} = \mathscr{C}_n \Gamma_{n+2}\\mathscr{C}_{n}^\top$. We first show that 
	\begin{equation*}
	\trace{\Phi_n}= 2^{(2-2H)n-4H}\alpha(H).
	\end{equation*}
	For the fixed $n \in \bN$, consider the following partition of the covariance matrix $\Gamma_{n+2}$, 
	\begin{equation}\label{eq Gamma n partition 2}
	\Gamma_{n+2} = \left[\begin{array}{c c c c}
	\Gamma^{\ast}_{1,1} & \Gamma^{\ast}_{1,2} &\cdots &\Gamma^{\ast}_{1,2^{n}}\\
	\Gamma^{\ast}_{2,1} & \Gamma^{\ast}_{2,2} &\cdots &\Gamma^{\ast}_{2,2^{n}}\\
	\vdots &\vdots & \ddots & \vdots\\
	\Gamma^{\ast}_{2^{n},1} & \Gamma^{\ast}_{2^{n},2} &\cdots &\Gamma^{\ast}_{2^{n},2^{n}}
	\end{array}\right],
	\end{equation}
	where $\Gamma^{\ast}_{i,j}$ are $4 \times 4$-dimensional matrices. In particular, for $1 \le i \le 2^{n}$, the diagonal partitioned matrices $\Gamma^{\ast}_{i,i}$ are of the form: 
	\begin{equation*}
	\Gamma^{\ast}_{i,i} = 2^{(1-2H)(n+2)}G_2 = 2^{(1-2H)(n+2)}\left[\begin{array}{c c c c}
	g_H(0) & g_H(1) & g_H(2) & g_H(3)\\
	g_H(1) & g_H(0) & g_H(1) & g_H(2)\\
	g_H(2) & g_H(1) & g_H(0) & g_H(1)\\
	g_H(3) & g_H(2) & g_H(1) & g_H(0)\\
	\end{array}
	\right].
	\end{equation*}
	Recalling the definition of $\bm\eta_{i,j}$ from \eqref{etaij def}, we get 
	\begin{equation*}
	\begin{split}
	\phi^{(n)}_{i,i}  &= \left(
	\bm\eta_{i,1}, \bm\eta_{i,2}, \dots,    \bm\eta_{i,2^{n}}	\right)\Gamma_{n+1}\left(
	\bm\eta_{i,1},  \bm\eta_{i,2},\dots , \bm\eta_{i,2^{n}}
	\right)^\top 
	\\&= \left(
		\bm 0_{1\times 4},  \dots, \bm\eta_{i,i},   \dots,\bm 0_{1\times 4}
	\right)\Gamma_{n+1}\left(
	\bm 0_{1\times 4},  \dots, \bm\eta_{i,i},  \dots,  \bm 0_{1\times 4}
	\right)^\top \\&= \bm r\Gamma^{\ast}_{i,i}\bm r^\top = 2^{(1-2H)(n+2)}\bm rG_2\bm r^\top.
	\end{split}
	\end{equation*}
		To evaluate the final term in the above equation, we have 
	\begin{equation*}
	\bm rG_2\bm r^\top  =\frac{1}{16}\bm 1_{1 \times 4}\left[\begin{array}{c c c c}
	g_H(0) & -g_H(1) & -g_H(2) & g_H(3)\\
	-g_H(1) & g_H(0) & g_H(1) & -g_H(2)\\
	-g_H(2) & g_H(1) & g_H(0) & -g_H(1)\\
	g_H(3) &- g_H(2) & -g_H(1) & g_H(0)\\
	\end{array}
	\right]\bm 1_{4 \times 1} = \frac{\alpha(H)}{4}.
	\end{equation*}
	Therefore, we have $\phi^{(n)}_{i,i} = 2^{(1-2H)(n+2)-2}\alpha(H)$ for every $1 \le i \le 2^n$, and 
	\begin{equation*}
	\trace{\Phi_n} = \sum_{i = 1}^{2^n}\phi^{(n)}_{i,i} = 2^{(2-2H)n-4H}\alpha(H).
	\end{equation*}
	In our next step, we shall show that $2^{2n(H-1)}\norm{\bm w_n}_{\ell_2}$ converges to $2^{-2H}\sqrt{\alpha(H)}$. First of all, it follows from \cite[Lemma 3.4]{HanSchiedMatrix} that $\norm{\mathscr{C}_n }_2 = 1/2$, and due to \eqref{eq fbm variance}, there exists a constant $c_H > 0$ such that  
	\begin{equation*}
	\norm{\Phi_n}_2 \le \norm{\mathscr{C}_n }^2_2 \norm{\Gamma_{n+2}}_2 \le c_H2^{(1-2H)n}.
	\end{equation*}
	For any given $\varepsilon > 0$, the concentration inequality \eqref{concentration inequality}
 yields that 
	\begin{equation*}
	\begin{split}
	&\bP\left(\left|2^{n(H-1)}\norm{\bm w_n}_{\ell_2}- 2^{n(H-1)}\sqrt{	\trace{\Phi_n}}\right| \ge \varepsilon\right) = \bP\left(\left|2^{n(H-1)}\norm{\bm w_n}_{\ell_2}- \sqrt{2^{2-4H}\alpha(H)}\right| \ge \varepsilon\right)  \\& \le \kappa\exp\left(-\frac{2^{n(2-2H)}\varepsilon^2}{\norm{\Phi_n}_2}\right)= \kappa\exp(-c^{-1}_H2^{n}\varepsilon^2).
	\end{split}
	\end{equation*}
From here, a Borel--Cantelli argument yields the assertion.
\end{proof}

\begin{proof}[Proof of \Cref{lemma fbm decaying}]
By  \eqref{gamman(WH) convergence eq} and \Cref{lemma fbm wn},
	\begin{equation*}
	\lim_{n \ua \infty}\frac{\norm{\bm w_n}_{\ell_2}}{\norm{\bar{\bm\theta}_{n}}_{\ell_2}} = \lim_{n \ua \infty}\sqrt{\frac{2^{n(2H-2)}\norm{\bm w_n}^2_{\ell_2}}{2^{n(2H-2)}\norm{\bar{\bm\theta}_{n}}_{\ell_2}^2}} = \sqrt{\frac{2^{-4H}\alpha(H)}{2^{2-2H}-1}} = \sqrt{\frac{\alpha(H)}{2^{2+2H}-2^{4H}}} < 1.
	\end{equation*}
See \Cref{alpha beta plot} for an illustration of the latter inequality.
\end{proof}

\begin{figure}[H]
	\centering
	\includegraphics[width=8cm]{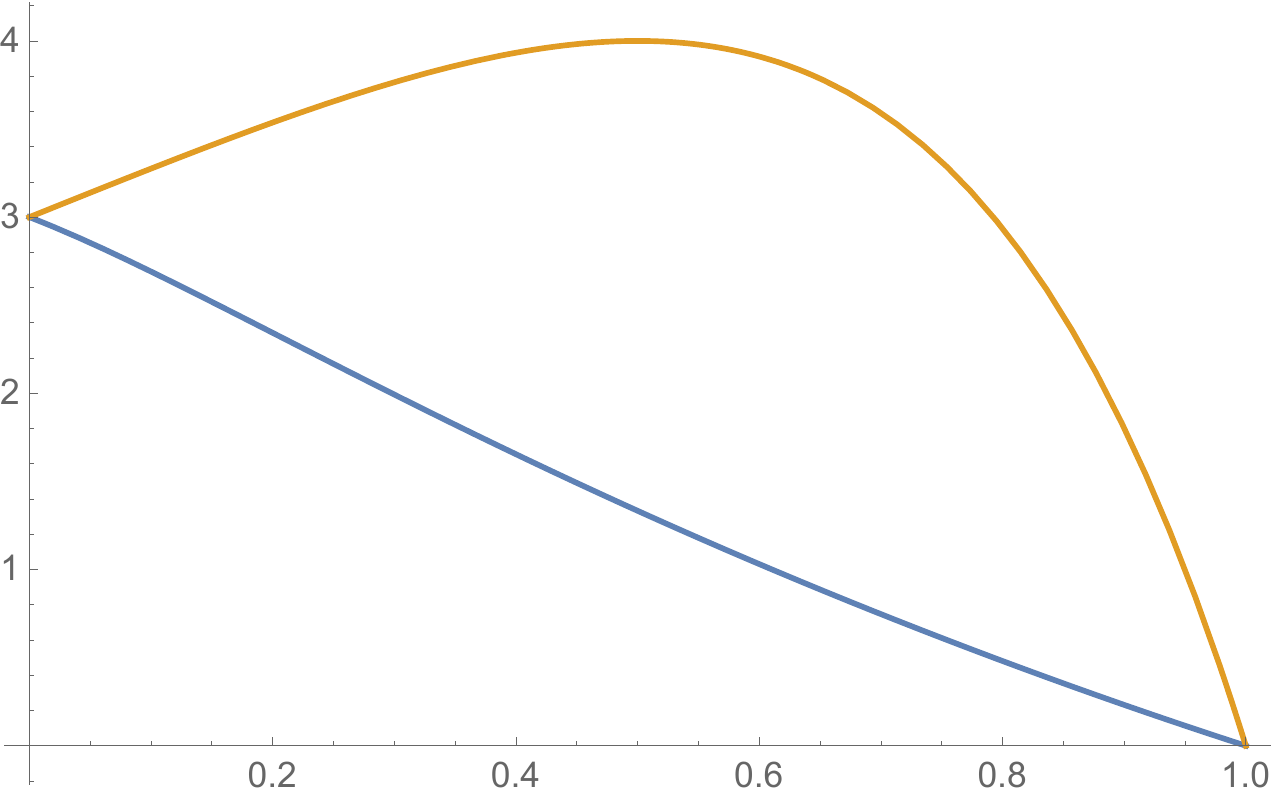}
	\caption{Plot of functions $\alpha(H)$ (blue) and $\beta(H):=2^{2+2H}-2^{4H}$ (orange) as functions of $H\in(0,1)$.}\label{alpha beta plot}
\end{figure}

\begin{proof}[Proof of \Cref{thm main rv} for $X=W^H$] It was shown in \cite[Theorem 5.1]{HanSchiedHurst} that $W^H$ admits $\bP$-a.s.~the roughness exponent $H$. It now follows from  \Cref{thm roughness composition} that the sample paths of $g(W^H)$
	also admit the roughness exponent $H$. 
	
	Now we prove that $\wh\cR_n(Y)\to H$ with probability one. To this end, we use the following result by Gladyshev~\cite[Theorem 1]{Gladyshev} on the convergence of the weighted quadratic variation of $W^H$,
		$$\lim_{n\ua\infty}2^{(2H-1)n}\sum_{k = 0}^{2^n-1}\big(W^H_{\frac{k+1}{2^n}} - W^H_{\frac{k}{2^n}} \big)^2= 1\quad\text{$\bP$-a.s.}
		$$
	Hence, if $\bar{\bm\theta}_{n}=(\theta_{n,k})$ are the Faber--Schauder coefficients of the sample paths of $W^H$, then  \cite[Proposition 4.1]{HanSchiedHurst} yields that 
	\begin{equation}\label{gamman(WH) convergence eq}
		2^{(2H-2)n} \norm{\bar{\bm\theta}_{n}}_{\ell_2}^2=2^{(2H-2)n}\sum_{k = 0}^{2^n-1}\theta^2_{n,k} \lra 2^{2-2H}-1\quad\text{$\bP$-a.s.}
	\end{equation}
	\Cref{gamma vs sn lemma} now implies that condition (a) of \Cref{Prop tilde H z} is satisfied. Condition (b) of that proposition is also satisfied, because the sample paths of $W^H$ are $\bP$-a.s.~H\"older continuous for every exponent $\alpha<H$; see, e.g., \cite[Section 1.16]{Mishura}. Hence, we may apply \Cref{Prop tilde H z}  and so $\wh\cR_n(X)\to H$ follows. 
\end{proof}

\begin{proof}[ Proof of \Cref{thm main rv} for $X = x_0 + \int_0^\cdot \xi_s \, ds + W^H$]
	 We first show that the process $g(X)$ admits roughness exponent $H$. Since the drift $\calU(t) := x_0 + \int_{0}^{\cdot} \xi_s \, ds$ is absolutely continuous, and hence, of bounded variation, it admits the roughness exponent $R = 1$ $\bP$-a.s.  It then follows from \Cref{lemma roughness sum} that $X$ admits the roughness exponent $H$ and from \Cref{thm roughness composition} that $g(X)$ admits the roughness exponent~$H$. 
	
	To show that $\wh \cR_n(\int_0^\cdot g(X_s) \,ds) \rightarrow H$, rewriting \eqref{eq new cond} yields that 
	\begin{equation*}
		\begin{split}
			2^{n(2H-1)}\int_0^1 \<\calU\>^{(2)}_n\left(\frac{s}{2^n}\right)\,ds &= 2^{n(2H-1)}\int_0^1 \sum_{k = 0}^{2^n-1} \left|\calU\left(\frac{k+1 +s }{2^n}\right) - \calU\left(\frac{k+s}{2^n}\right) \right|^2 \,ds \\&= 2^{n(2H-1)}\int_0^1 \sum_{k = 0}^{2^n-1}\left|\int_{\frac{k+s}{2^n}}^{\frac{k+s+1}{2^n}} \xi_u \, du\right|^2 \, ds
		\end{split}
	\end{equation*}
	converges to zero as $n \ua \infty$ $\bP$-a.s. Moreover, it follows from \Cref{gamma vs sn lemma} and \Cref{lemma fbm decaying} that  $W^H$ satisfies condition \ref{Prop tilde H z part a} of \Cref{Prop tilde H z} with probability one. Hence, it follows from \Cref{lemma add} that for a typical sample path $x$ of $X$ and its antiderivative $y$, 
	\begin{equation*}
		0 < \liminf_{n \ua \infty} 2^{n(2H-2)} \sum_{k = 0}^{2^n-1} \left(\vartheta^y_{n,k}\right)^2 \le \limsup_{n \ua \infty} 2^{n(2H-2)} \sum_{k = 0}^{2^n-1} \left(\vartheta^y_{n,k}\right)^2 < \infty.
	\end{equation*}
Next, recall our assumption that the sample path of $\xi$ belong to  $ L^p([0,1])$ for some $p>5/(5-2H)$. Therefore,  Jensen's inequality yields that 
$$	\left|\calU(t) - \calU(s)\right| \le \left(\int_s^t |\xi_s|^p\, ds \right)^{\frac{1}{p}} (t - s)^{1 - \frac{1}{p}} \le \norm{\xi}_{L^p[0,1]}(t - s)^{1 - \frac{1}{p}}. 
$$
Thus, the drift $\calU$ is $\bP$-a.s.~H\"{o}lder continuous with exponent $\alpha = 1- 1/p > 2H/5$. In addition, $W^H$ is $\bP$-a.s.~H\"{o}lder continuous for any exponent less than $H$. Hence, the sample paths of $X$ are $\bP$-a.s.~H\"{o}lder continuous for any given exponent $\beta \in( 2H/5,H)$. Hence, it follows from \Cref{Prop tilde H z} that $\wh \cR_n(\int_0^\cdot g(X_s) \,ds) \rightarrow H$ as $n \ua \infty$. This completes the proof.
\end{proof}

\subsection{Proof of \Cref{Thm Rough Bergomi} and \Cref{Cor Rough Bergomi}}\label{Section Proof Bergomi}
The proof of the case $\nu > 0$ is straightforward: Since the laws of $Z^{H,\nu}/\Gamma(H + 1/2)$ and $W^H$ are equivalent according to \cite[Theorem 20]{PicardRepresentationFBM},  we can simply replace  $W^H$  with $Z^{H,\nu}$ in the $\bP$-almost sure statements of \Cref{thm main rv}.

 To prove the case $\nu = 0$, recall the the Mandelbrot--van Ness representation of fractional Brownian motion,
\begin{equation*}
	W^H_t = \frac{1}{\Gamma(H + \frac{1}{2})}\left(\int_{-\infty}^{0}\left[(t-u)^{H-\frac{1}{2}} - (-u)^{H-\frac{1}{2}}\right]\,dB_u + \int_0^t (t-u)^{H-\frac{1}{2}}\,dB_u\right), \qquad t \ge 0,
\end{equation*}
where $B$ is a two-sided standard Brownian motion. Using this representation, we can decompose \eqref{eq RL with drift} as follows, 
\begin{equation*}
		X_t = \left(x_0 + \int_0^t \xi_s \, ds\right) + \Gamma\left(H + \frac{1}{2}\right)W^H_t - \wt X^H_t,
\end{equation*}
where the process $(\wt X^H_t)_{t \in [0,1]}$ defined by
\begin{equation}\label{eq XH 2}
	\wt X^H_t := \Gamma\left(H+\frac{1}{2}\right) W^H_t - Z^H_t = \int_{-\infty}^{0}\left[(t-u)^{H-\frac{1}{2}} - (-u)^{H-\frac{1}{2}}\right]\,dB_u, \qquad t \in [0,1].
\end{equation}
It was shown in \cite[Theorem 17]{PicardRepresentationFBM} that the process \(\wt X^H\) is well defined. In Lemmas~\ref{lemma XH QV} and~\ref{lemma XH roughness}, we verify that \(\wt X^H\) satisfies the assumptions of Propositions~\ref{lemma roughness sum} and~\ref{lemma add}. As a consequence, the process X admits the roughness exponent $H$, and the estimator \(\wh \cR_n\!\left(\int_0^\cdot X_s\,ds\right)\) converges to $H$ as \(n \ua \infty\). Finally, the desired result will follow by applying Propositions~\ref{thm roughness composition} and~\ref{Prop tilde H z}.

We begin with the following lemma, which verifies condition \ref{item b lemma add} of \Cref{lemma add} for the process $\wt X^H$.

\begin{lemma}\label{lemma XH QV}
	Let $\wt X^H$ be as in \eqref{eq XH 2}. Then, with probability one, we have
	\begin{equation*}
		\lim_{n \ua \infty} 2^{n(2H-1)}\int_0^{1}\<\wt X^H\>^{(2)}_n\left(\frac{s}{2^n}\right)\,ds = 0.
	\end{equation*}
\end{lemma}
\begin{proof}
	We have
	\begin{equation}\label{eq XH decomposition}
		\scalemath{0.9}{	\begin{split}
				2^{n(2H-1)}\int_0^{1}\<\wt X^H\>^{(2)}_n\left(\frac{s}{2^n}\right)\,ds &= 2^{n(2H-1)}\int_0^{1}\sum_{k = 0}^{2^n-1}\left(\wt X^H_{\frac{k+1+s}{2^n}} - \wt X^H_{\frac{k+s}{2^n}}\right)^2\,ds\\&= 2^{n(2H-1)}\int_0^{1}\left(\wt X^H_{\frac{s+1}{2^n}} - \wt X^H_{\frac{s}{2^n}}\right)^2\,ds + 2^{n(2H-1)}\int_0^{1}\sum_{k = 1}^{2^n-1}\left(\wt X^H_{\frac{k+1+s}{2^n}} - \wt X^H_{\frac{k+s}{2^n}}\right)^2\,ds.
		\end{split}}
	\end{equation}
	In the following, we will prove this lemma by showing that the two rightmost terms in \eqref{eq XH decomposition} converge to zero with probability one. To this end, note that for $0 \le s \le t < \infty$, we have 
	\begin{equation*}
		\wt X^H_t - \wt X^H_s = \int_{-\infty}^{0}\left[(t-u)^{H - \frac{1}{2}} - (s-u)^{H-\frac{1}{2}}\right]dB_u = \left(H-\frac{1}{2}\right)^2\int_{-\infty}^{0}\left(\int_s^t(\tau - u)^{H-\frac{3}{2}}\,d\tau \right)\,dB_u. 
	\end{equation*}
	Applying the Cauchy--Schwarz inequality then yields 
	\begin{equation}\label{eq XH second moment}
		\begin{split}
			\bE\left[\left(\wt X^H_t - \wt X^H_s\right)^2\right] &= \left(H-\frac{1}{2}\right)^2\int_{-\infty}^0 \left(\int_s^t(\tau - u)^{H-\frac{3}{2}}\,d\tau \right)^2\,du\\&\le \left(H-\frac{1}{2}\right)^2\cdot (t-s)\int_{-\infty}^0 \int_s^t (\tau - u)^{2H-3}\, d\tau du\\&= \left(H-\frac{1}{2}\right)^2\cdot (t-s)\int_s^t\int_{-\infty}^{0} (\tau - u)^{2H-3}\,dud\tau \\&= \frac{\left(H-\frac{1}{2}\right)^2}{2-2H}\cdot (t-s)\int_s^t \tau^{2H-2}\, d\tau.
		\end{split}
	\end{equation}
	In particular, the first identity in \eqref{eq XH second moment} yields that the function $s \mapsto \bE\big[\big(\wt X^H_{\frac{1+s}{2^n}} - \wt X^H_\frac{s}{2^n}\big)^2\big]$ is decreasing. Let us now denote 
	\begin{equation*}
		c_H := \int_0^{1}\left((1-u)^{H-\frac{1}{2}} - (-u)^{H-\frac{1}{2}}\right)^2\,du.
	\end{equation*}
	It was shown in \cite[Proposition 2.3]{Nourdin2012} that $c_H < \infty$ for all $H \in (0,1)$. Furthermore, the process $\wt X^H$ is $H$-self-similar, i.e., $(\wt X^H_{\alpha t}) \overset{d}{=} (\alpha^{H} \wt X^H_t)$. This gives 
	\begin{equation*}
		\begin{split}
			\bE\left[\left(\wt X^H_{\frac{1+s}{2^n}} - \wt  X^H_\frac{s}{2^n}\right)^2\right]  &\le \bE\left[\left(\wt X^H_{\frac{1}{2^n}} - \wt X^H_0\right)^2\right] = \bE\left[\left(\wt X^H_{\frac{1}{2^n}}\right)^2\right] = 2^{-2Hn}\bE\left[\left(\wt X^H_1\right)^2\right]\\&= 2^{-2Hn}\bE\left[\left(\int_{-\infty}^{0}\left[(1-u)^{H-\frac{1}{2}} - (-u)^{H-\frac{1}{2}}\right]\,dB_u\right)^2\right]\\&= 2^{-2Hn} \int_0^{1}\left((1-u)^{H-\frac{1}{2}} - (-u)^{H-\frac{1}{2}}\right)^2\,du = 2^{-2Hn}c_H,
		\end{split}
	\end{equation*}
	Hence, applying the Fubini--Tonelli theorem, we have 
	\begin{equation*}
		2^{n(2H-1)}\bE\left[\int_0^1\left(\wt X^H_{\frac{1+s}{2^n}} - \wt X^H_\frac{s}{2^n}\right)^2\,ds\right] = 2^{n(2H-1)}\int_{0}^{1}\bE\left[\left(\wt X^H_{\frac{1+s}{2^n}} - \wt X^H_\frac{s}{2^n}\right)^2\right]\,ds \le 2^{-n}c_H,
	\end{equation*}
	and thus, $2^{n(2H-1)}\bE\big[\int_0^1\big(\wt X^H_{\frac{1+s}{2^n}} - \wt X^H_\frac{s}{2^n}\big)^2\,ds\big] \rightarrow 0$ as $n \ua \infty$. 
	
	Next, since $\wt X^H_t$ is a centered Gaussian process, we have  
	\begin{equation*}
		\bE\left[\left(\wt X^H_{\frac{1+s}{2^n}} - \wt X^H_\frac{s}{2^n}\right)^4\right] = 3\cdot \bE\left[\left(\wt X^H_{\frac{1+s}{2^n}} - \wt X^H_\frac{s}{2^n}\right)^2\right]^2 \le 3\cdot\bE\left[\left(\wt X^H_{\frac{1}{2^n}} - \wt X^H_0\right)^2\right]^2 \le  2^{-4Hn+2}c^2_H.
	\end{equation*}
	Hence, by a direct application of the Cauchy--Schwarz inequality and the Fubini--Tonelli theorem, 
	\begin{equation*}
		\begin{split}
			\var\left[2^{n(2H-1)}\int_0^{1}\left(\wt X^H_{\frac{s+1}{2^n}} - \wt X^H_{\frac{s}{2^n}}\right)^2\,ds\right] &\le2^{n(4H-2)} \bE\left[\left(\int_0^{1}\left(\wt X^H_{\frac{s+1}{2^n}} - \wt X^H_{\frac{s}{2^n}}\right)^2\,ds\right)^2\right]\\&\le 2^{n(4H-2)}\bE\left[\int_0^{1}\left(\wt X^H_{\frac{s+1}{2^n}} - \wt X^H_{\frac{s}{2^n}}\right)^4\,ds\right]  \le 2^{-2n+2}c_H^2.
		\end{split}
	\end{equation*}
	Clearly, the rightmost term converges to zero exponentially fast. Thus, by fast $L_2$-convergence, $2^{n(2H-1)}\int_0^{1}\big(\wt X^H_{\frac{s+1}{2^n}} - \wt X^H_{\frac{s}{2^n}}\big)^2\,ds$ converges to zero with probability one as $n \ua \infty$. 
	
	Next, we aim to show that the second term in \eqref{eq XH decomposition} also converges to zero with probability one. To this end, we first  show that for $\gamma > 0$ and almost every sample path $x$ of $\wt X^H$, there exists $\kappa_{\gamma} > 0$ and $n_0 \in \bN$ such that for $n \ge n_0$, we have 
	\begin{equation}\label{eq XH Lip}
		\left|x\left(\frac{k+1+s}{2^n}\right)-x\left(\frac{k+s}{2^n}\right)\right| \le \kappa_\gamma 2^{-(1-\gamma)n}\left(\frac{k}{2^{n}}\right)^{H-1}, \qquad s \in [0,1] \quad \text{and} \quad k \ge 1.
	\end{equation}
	To this end, applying \eqref{eq XH second moment} yields that 
	\begin{equation*}
		\bE\left[\left(\wt X^H_{\frac{k+1}{2^n}}-\wt X^H_{\frac{k}{2^n}}\right)^2\right] \le \frac{2^{-n}(H-\frac{1}{2})^2}{2-2H}\int_{\frac{k}{2^n}}^{\frac{k+1}{2^n}}\tau^{2H-2}\,d\tau \le \frac{2^{-2n}(H-\frac{1}{2})^2}{2-2H}\left(\frac{k}{2^n}\right)^{2H-2}.
	\end{equation*}
	For any given $\gamma > 0$, Chebyshev's inequality gives 
	\begin{equation*}
		\scalemath{0.9}{\begin{split}
				\bP\left(\max_{1 \le k \le 2^n-1}\left(\frac{k}{2^n}\right)^{1-H}|\wt X^H_{\frac{k+1}{2^n}}-\wt X^H_{\frac{k}{2^n}}| \ge 2^{-(1-\gamma)n}\right) &\le \sum_{k = 1}^{2^n-1}\bP\left(\left(\frac{k}{2^n}\right)^{1-H}|\wt X^H_{\frac{k+1}{2^n}}-\wt X^H_{\frac{k}{2^n}}| \ge 2^{-(1-\gamma)n}\right) \\&\le \sum_{k = 1}^{2^n-1}\frac{\bE\left[\left(\wt X^H_{\frac{k+1}{2^n}}-\wt X^H_{\frac{k}{2^n}}\right)^2\right]\left(\dfrac{k}{2^n}\right)^{2-2H}}{2^{-(2-2\gamma)n}}\le \frac{(H-\frac{1}{2})^2}{2-2H}\cdot 2^{-2\gamma n}.
		\end{split}}
	\end{equation*}
	Clearly, the rightmost term in the above inequality is summable. Thus, the Borel--Cantelli lemma yields that for a typical sample path $x$, there exists $n_0 \in \bN$ such that 
	\begin{equation}\label{eq XH dyadic}
		\left|x\left(\frac{k+1}{2^n}\right)-x\left(\frac{k}{2^n}\right)\right| \le  2^{-(1-\gamma)n}\left(\frac{k}{2^{n}}\right)^{H-1}, \qquad  k \ge 1.
	\end{equation}
	To extend the above inequality to any given $s \in [0,1]$ as in \eqref{eq XH Lip},  consider an increasing sequence $(s_j)_{j \ge n}$ such that for $j \ge n$
	\begin{equation*}
		s_{n} = \frac{k}{2^{n}}, \quad |s_{j+1} - s_{j}| \in \left\{0,\frac{1}{2^{j+1}}\right\} \quad \text{and} \quad s_j \ua \frac{k+s}{2^n}
	\end{equation*}
	In a similar manner, we can construct a sequence $(s'_k)_{k \ge n}$ analogously that converges to $ \frac{k+1+s}{2^n}$. Then, we have 
	\begin{equation*}
		\scalemath{1}{x\left(\frac{k+1+s}{2^n}\right)-x\left(\frac{k+s}{2^n}\right) = \sum_{j = n+1}^{\infty}x(s_{j+1}) - x(s_{j})  + x\left(\frac{k+1}{2^n}\right)-x\left(\frac{k}{2^n}\right) + \sum_{j = n+1}^{\infty}x(s'_{j+1}) - x(s'_{j}).}
	\end{equation*}
	Therefore, applying the triangle inequality gives 
	\begin{align*}
		\lefteqn{\left|x\left(\frac{k+1+s}{2^n}\right)-x\left(\frac{k+s}{2^n}\right)\right|}\\ &= \sum_{j = n+1}^{\infty}|x(s_{j+1}) - x(s_{j})|  + \left|x\left(\frac{k+1}{2^n}\right)-x\left(\frac{k}{2^n}\right)\right| + \sum_{j = n+1}^{\infty}|x(s'_{j+1}) - x(s'_{j})|\\&\le \sum_{j = n+1}^\infty\left(s_j\right)^{H-1}|s_{j+1} - s_j|^{H-1} + 2^{-(1-\gamma)n}\left(\frac{k}{2^{n}}\right)^{H-1} + \sum_{j = n+1}^\infty\left(s'_j\right)^{H-1}|s'_{j+1} - s'_j|^{H-1} \\&\le\left( \frac{k}{2^n}\right)^{H-1}\left(\sum_{j = n+1}^{\infty} 2^{-(1-\gamma)j} + 2^{-(1-\gamma)n}+\sum_{j = n+1}^{\infty} 2^{-(1-\gamma)j} \right) = \left(\frac{k}{2^n}\right)^{H-1} 2^{-(1-\gamma)n}\left(1 + \frac{2}{1-2^{-(1-\gamma)}}\right).
	\end{align*}
	Taking $\kappa_\gamma = 1 + \frac{2}{1-2^{-(1-\gamma)}}$ yields back \eqref{eq XH Lip}. For given $H \in (0,1)$, we take $\gamma < 1 - H \vee \frac{1}{2}$. Then 
	\begin{equation*}
		\begin{split}
			2^{n(1-2H)}\sum_{k = 1}^{2^n}\left|x\left(\frac{k+1+s}{2^n}\right)-x\left(\frac{k+s}{2^n}\right)\right|  &\le \kappa_\gamma\sum_{k = 1}^{2^n-1} 2^{n(1-2H)-2(1-\gamma)n}\left(\frac{k}{2^{n}}\right)^{2H-2}= \kappa_\gamma2^{n(2\gamma-1)}\sum_{k = 1}^{2^n-1}k^{2H-2}.
		\end{split}
	\end{equation*}
	Clearly, we have 
	\begin{equation*}
		\sum_{k = 1}^{2^n-1}k^{2H-2} = \begin{cases}
			\mathcal{O}(1) &\qquad \text{for} \quad H \in (0,1/2),\\
			\mathcal{O}(2^{(2H-1)n}) &\qquad \text{for} \quad H \in (1/2,1).
		\end{cases}
	\end{equation*}
	Thus, we have $2^{n(2\gamma-1)}\sum_{k = 1}^{2^n-1}k^{2H-2} = \mathcal{O}(2^{n(2\gamma - 1)} \vee 2^{n(2H + 2\gamma -2)})$. Hence, for any $s \in [0,1]$, we have  
	\begin{equation*}
		\lim_{n \ua \infty} 2^{n(1-2H)}\sum_{k = 1}^{2^n}\left|x\left(\frac{k+1+s}{2^n}\right)-x\left(\frac{k+s}{2^n}\right)\right|   = 0.
	\end{equation*}
	Finally, the assertion follows from Lebesgue's dominated convergence theorem.
\end{proof}

The next lemma establishes that for any $H \in (0,1)$, the processes $\wt X^H$ and $Z^H$ admit $\bP$-a.s.~the respective roughness exponents $R = 1$ and $R = H$.

\begin{lemma}\label{lemma XH roughness}
	For any $H \in (0,1)$, the process $\wt X^H$ admits the roughness exponent $R = 1$, and the Riemann--Liouville process $Z^H$ admits the roughness exponent $R = H$ with probability one.
\end{lemma}
\begin{proof}
	We begin with showing that $\wt X^H$ admits the roughness exponent $R = 1$ almost surely. Let $x$ be a typical sample path of $\wt  X^H$. It suffices to show that $\<x\>^{(p)}_n \da 0$ as $n \ua \infty$ for any $p > 1$. To this end, note that 
	\begin{equation*}
		\<x\>^{(p)}_n = \sum_{k = 0}^{2^n-1}\left|x\left(\frac{k+1}{2^n}\right)-x\left(\frac{k}{2^n}\right)\right|^p = |x\left(\frac{1}{2^n}\right) - x\left(0\right)|^p + \sum_{k = 1}^{2^n-1}\left|x\left(\frac{k+1}{2^n}\right)-x\left(\frac{k}{2^n}\right)\right|^p.
	\end{equation*}
	Clearly, $|x(2^{-n}) - x(0)| \rightarrow 0$ as $n \ua \infty$, as $x$ is continuous. Hence, it suffices to prove that the second term converges to zero. For given $p > 1$, we take $\gamma < (1 - 1/p) \wedge H$. Then \eqref{eq XH dyadic} implies that there exists $n_0 \in \bN$ such that for $n > n_0$, we have 
	\begin{equation*}
		\begin{split}
			\sum_{k = 1}^{2^n-1}\left|x\left(\frac{k+1}{2^n}\right)-x\left(\frac{k}{2^n}\right)\right|^p \le \sum_{k = 1}^{2^n-1}2^{-p(1-\gamma)n}\left(\frac{k}{2^n}\right)^{p(H-1)} = 2^{(\gamma - H)pn}\sum_{k = 1}^{2^n-1}k^{pH-p}.
		\end{split}
	\end{equation*}
	It is clear that $\sum_{k = 1}^{2^n-1}k^{pH-p} = \mathcal{O}(2^{(pH-p+1)n} \vee 1)$, and hence, 
	\begin{equation*}
		\sum_{k = 1}^{2^n-1}\left|x\left(\frac{k+1}{2^n}\right)-x\left(\frac{k}{2^n}\right)\right|^p = \mathcal{O}(2^{(1+\gamma-p)n} \vee 2^{-(\gamma -H)p}).
	\end{equation*}
	As a result, the above term converges to zero as $n \ua \infty$, and this demonstrates the $\wt X^H$ admits the roughness exponent $R = 1$ with probability one. Recall that it was shown in \cite[Theorem 5.1]{HanSchiedHurst} that $W^H$ admits $\bP$-a.s.~the roughness exponent $H$. Together with the definition $Z^H = \Gamma(H + 1/2)W^H - \wt X^H$, \Cref{lemma add} leads to the assertion.
\end{proof}

\begin{proof}[Proof of \Cref{Thm Rough Bergomi}] The proof for the case $\nu>0$ was already explained at the beginning of this section, based on the absolute continuity result from \cite[Theorem 20]{PicardRepresentationFBM}. So let us now consider the case $\nu=0$. 
	Using the results from Sections \ref{pathwise section} and \ref{Section Proof fbm} and arguing as in the proof of \Cref{thm main rv} allows us to assume without loss of generality that $x_0 = 0$ and $\xi_t  = 0$ for all $t \in [0,1]$. That is $X = Z^H$.

	That $g\circ X$ admits the $\bP$-a.s.~roughness exponent $R = H$ for and strictly monotone  $g \in C^2(\bR)$ follows immediately from \Cref{lemma XH roughness} and \Cref{thm roughness composition}.
	
	Finally, we prove $\wh \cR_n\left(\int_0^\cdot g(X_s)\,ds\right) \rightarrow H$ as $n \ua \infty$. To this end, we denote $Y_t:=\int_0^tX_s\,ds$. It follows from \Cref{lemma XH QV} that $\bP$-a.s.,
	\begin{equation*}
		\lim_{n \ua \infty} 2^{n(2H-1)}\int_0^{1}\<\wt X^H\>^{(2)}_n\left(\frac{s}{2^n}\right)\,ds = 0,
	\end{equation*}
	and it follows from  \Cref{gamma vs sn lemma} and \Cref{lemma fbm decaying} that  $W^H$ satisfies condition \ref{item a lemma add} of \Cref{lemma add} with probability one. Thus, we have $\bP$-a.s.,
	\begin{equation*}
		0< \liminf_{n \ua \infty}2^{n(2{R}-2)}\sum_{k = 0}^{2^n-1}\left(\vartheta^Y_{n,k}\right)^2\le \limsup_{n \ua \infty}2^{n(2{R}-2)}\sum_{k = 0}^{2^{n}-1}\left(\vartheta^Y_{n,k}\right)^2<\infty,
	\end{equation*}
	which validates condition \ref{Prop tilde H z part a} of \Cref{Prop tilde H z} for the process $X$. In addition, the process $X = Z^H$ is $\bP$-a.s.~H\"older continuous with exponent $\alpha < H$. Now, applying \Cref{Prop tilde H z} yields the assertion. This completes the proof. 
\end{proof}

\begin{proof}[Proof of \Cref{Cor Rough Bergomi}]
	
	Note that under the dynamics \eqref{eq VHu}, we have
	\begin{equation}\label{eq g1}
		\sigma_t = \exp(V^{H,\nu}_t) = \exp\left(\gamma\left(Z^{H,\nu}_t - \frac{\gamma}{2}\int_0^t s^{2H-1}\, ds\right)\right) = g_1\left(Z^{H,\nu}_t - \frac{\gamma}{2}\int_0^t s^{2H-1}\, ds\right),
	\end{equation}
	and 
	\begin{equation}\label{eq g2}
		\begin{split}
			\int_0^t \sigma^2_s\,ds &= \int_0^t \exp(2V^{H,\nu}_s)\, ds = \int_0^t \exp\left(2\gamma\left(Z^{H,\nu}_s - \int_{0}^{t}\frac{\gamma}{2}s^{2H-1}\,ds\right)\right)\\ &= \int_0^1 g_2\left(Z^{H,\nu}_t - \frac{\gamma}{2}\int_0^t s^{2H-1}\, ds\right) dt,
		\end{split}
	\end{equation}
	where $g_1(t) = e^{\gamma t}$ and $g_2(t) = e^{2\gamma t}$. From the above, it suffices to show that the process $\xi(t) = t^{2H-1}$ satisfies the conditions in \Cref{Thm Rough Bergomi}. If so, together with \eqref{eq g1}, \Cref{Thm Rough Bergomi} yields that $\sigma_t$ admits the roughness exponent $H$, and, together with \eqref{eq g2}, \Cref{Thm Rough Bergomi} yields that $\wh \cR_n(\int_{0}^{\cdot}\sigma^2_s\, ds) \rightarrow H$ as $n \ua \infty$.
	
	To this end, we first verify condition \ref{thm main rv part a} in \Cref{thm main rv}, which is required by \Cref{Thm Rough Bergomi}. Note that if $H \ge \frac{1}{2}$, then $\xi \in L^\infty([0,1])$, which implies condition \ref{thm main rv part a}. On the other hand, if $H < \frac{1}{2}$, then $\xi \in L^p([0,1])$ for any $p < \frac{1}{1 - 2H}$. As $\frac{1}{1 - 2H} > \frac{5}{5 - 2H}$, then the condition \ref{thm main rv part a} also holds if $H < \frac{1}{2}$. To verify condition \ref{thm main rv part b} in \Cref{thm main rv}, note that for $H \ge 1/2$, the function $t \mapsto t^{2H}$ is Lipschitz continuous, and there exists a positive constant $L$ such that $|t^{2H} - s^{2H}| \le L|t -s|$ for $s,t \in [0,1]$. Therefore, we have
	\begin{equation*}
		\lim_{n \ua \infty} 2^{n(2H-1)}\int_0^{1}\sum_{k = 0}^{2^n-1}\left|\left(\frac{k+1+s}{2^{n}}\right)^{2H}-\left(\frac{k+s}{2^{n}}\right)^{2H}\right|^2\,ds \le \lim_{n \ua \infty} 2^{n(2H-1)} \sum_{k = 0}^{2^n-1}L\cdot 2^{-2n} = \lim_{n \ua \infty}L \cdot 2^{n(2H-2)} = 0.
	\end{equation*}
	For the case $H < 1/2$, we have $|t^{2H} - s^{2H}| \le |t-s|^{2H}$ for $s,t \in [0,1]$. Moreover, the mean-value theorem yields that $|t^{2H} - s^{2H}| \le 2H s^{2H-1} |t-s|$ for $0 < s < t$. Hence, we have 
	\begin{equation*}
		\begin{split}
			\sum_{k = 0}^{2^n-1}\int_{0}^{1}\left|\left(\frac{k+1+s}{2^{n}}\right)^{2H}-\left(\frac{k+s}{2^{n}}\right)^{2H}\right|^2\,ds &\le 2^{-2Hn} + 2^{-2n}\sum_{k = 1}^{2^n-1}\int_{0}^{1}\left(\frac{k+s}{2^n}\right)^{4H-2}\,ds\\& = 2^{-2Hn} + 2^{-4Hn}\sum_{k = 1}^{2^n-1}\int_{k}^{k+1} s^{4H-2}\, ds  \\& = 2^{-2Hn} + 2^{-4Hn}\int_{1}^{2^n} s^{4H-2}\, ds\\&\le 2^{-2Hn} +2^{-4Hn}\frac{2^{n(4H-1)}\vee1}{|4H-1|}= 2^{-2Hn} + \frac{2^{-n}\vee 2^{-4Hn}}{|4H-1|}.
		\end{split}
	\end{equation*}
	Therefore, we have 
	\begin{equation*}
		\lim_{n \ua \infty}2^{n(2H-1)}\int_0^{1}\sum_{k = 0}^{2^n-1}\left|\left(\frac{k+1+s}{2^{n}}\right)^{2H}-\left(\frac{k+s}{2^{n}}\right)^{2H}\right|^2\,ds \le \lim_{n \ua \infty} 2^{-n} +\lim_{n \ua \infty} \frac{2^{(2H-2)n} \vee 2^{(-1-2H)n}}{|4H-1|}  = 0.
	\end{equation*}
	This completes the proof.
\end{proof}

\subsection{Proof of \Cref{Theorem scale rate}}\label{Section Proof Rate}

Theorem 1.4 in \cite{HanSchiedGirsanov} states  that the law of the process $X=x_0+W^H+\int_0^\cdot\xi_s\,ds$ is absolutely continuous with respect to the law of $x_0+W^H$ if $\xi$ satisfies the conditions of \Cref{cor main rv}.  Hence, it is sufficient to prove the theorem for the case $\xi=0$ and $x_0=0$, so that $X=W^H$. 

The proof of \Cref{Theorem scale rate} begins with the following lemma, which provides an upper bound for the covariance between Faber--Schauder coefficients of the fractional Brownian motion at different generations. Such a bound is obtained \`{a} la \cite[Theorem 1]{Gladyshev}. 
\begin{lemma}\label{Lemma Faber Bound}
	Let $m,n \in \bN$, $0 \le k \le 2^n-1$ and $0 \le j \le 2^m-1$ be given such that 
	\begin{equation*}
		\big|2^n j - 2^m k\big| > 2^{m \vee n}.
	\end{equation*}
	Then there exists a positive constant $\rho_H$ such that
	\begin{equation}\label{eq Faber bound}
		\big|\cov\big(\theta_{n,k},\theta_{m,j}\big)\big| \le \rho_H \cdot 2^{-3(m+n)/2}\left(\left|\frac{k+1}{2^n}-\frac{j}{2^m}\right| \wedge \left|\frac{k}{2^n}-\frac{j+1}{2^m}\right|\right)^{2H-4}.
	\end{equation}
\end{lemma}

\begin{proof}
	For fixed $n,m \in \bN$ and $s,t \in \bR$, we denote
	\begin{equation}\label{eq fnm}
		\begin{split}
			f_{n,m}(s,t) &:= \cov\left(W^H_{s + \frac{1}{2^{n+1}}} - W^H_s, W^H_{t + \frac{1}{2^{m+1}}} - W^H_t \right)\\&= \frac{1}{2}\left(\left|s-t+\frac{1}{2^{n+1}}\right|^{2H}+\left|s-t-\frac{1}{2^{m+1}}\right|^{2H}-\left|s-t\right|^{2H}-\left|s-t+\frac{1}{2^{n+1}}-\frac{1}{2^{m+1}}\right|^{2H}\right).
		\end{split}
	\end{equation}
	Recall that $\theta_{n,k} = 2^{n/2}\big((W^H_{\frac{2k+1}{2^{n+1}}} - W^H_{\frac{2k}{2^{n+1}}})-(W^H_{\frac{2k+2}{2^{n+1}}} - W^H_{\frac{2k+1}{2^{n+1}}})\big)$. Hence,	\begin{equation}\label{eq partial bound}
				\cov\big(\theta_{n,k},\theta_{m,j}\big)= 2^{\frac{n+m}{2}}\int_{\frac{2k}{2^{n+1}}}^{\frac{2k+1}{2^{n+1}}}\int_{\frac{2j}{2^{m+1}}}^{\frac{2j+1}{2^{m+1}}}\frac{\partial^2 f_{n,m}(s,t)}{\partial s \partial t}\,dsdt \le 2^{-\frac{m+n}{2}-2}\sup_{(s,t) \in \calS^{(n+1,m+1)}_{\frac{k}{2^n},\frac{j}{2^m}}}\left|\frac{\partial^2 f_{n,m}(s,t)}{\partial s \partial t}\right|,
	\end{equation}
	where for fixed $m,n \in \bN$, $\calS^{(m,n)}_{a,b}$ denotes the rectangular area of the following form $$\calS^{(m,n)}_{a,b}:= \left[a, a + \frac{1}{2^{n}}\right] \times \left[b, b+\frac{1}{2^{m}}\right] \qquad \text{for} \quad a,b \in \bR.$$

	To obtain an upper bound for the rightmost term in \eqref{eq partial bound}, let $\rho_1 = 2H|2H-1|$ and $\rho_2 = \prod_{i = 0}^{3}|2H-i|$. Then,	\begin{equation*}
		\scalemath{0.95}{
			\begin{split}
				\left|\frac{\partial^2 f_{n,m}(s,t)}{\partial s \partial t}\right| &= \frac{\rho_1}{2}\left||t-s|^{2H-2} + \left|s-t+\frac{1}{2^{n+1}}-\frac{1}{2^{m+1}}\right|^{2H-2} - \left|s-t+\frac{1}{2^{n+1}}\right|^{2H-2}-\left|s-t-\frac{1}{2^{m+1}}\right|^{2H-2}\right|\\
				&\le\frac{\rho_2}{2}\int_{s}^{s+\frac{1}{2^{n+1}}}\int_{t}^{t+\frac{1}{2^{m+1}}}|u-v|^{2H-4}\,dudv  \le \frac{\rho_2}{2^{n+m+3}}\sup_{(u,v) \in \calS^{(n+1,m+1)}_{s,t}}|u - v|^{2H - 4}.
		\end{split}}
	\end{equation*}
	Applying the above inequality to \eqref{eq partial bound} yields 
	\begin{equation}\label{eq partial bound 2}
		\begin{split}
			\big|\cov\big(\theta_{n,k},\theta_{m,j}\big)\big| &\le \rho_2\cdot 2^{-\frac{3(m+n)}{2}-4}\sup_{(s,t) \in \calS^{(n+1,m+1)}_{\frac{k}{2^n},\frac{j}{2^m}}} \sup_{(u,v) \in \calS^{(n+1,m+1)}_{s,t}}|u - v|^{2H - 4}.
		\end{split}
	\end{equation}
	Since $\big|2^n j - 2^m k\big| > 2^{m \vee n}$, we have $\left[\frac{k}{2^n},\frac{k+1}{2^n}\right] \cap \left[\frac{j}{2^m},\frac{j+1}{2^m}\right] = \emptyset$. As a result, we have 
	\begin{equation*}
		\begin{split}
			\sup_{(s,t)\in \calS^{(n+1,m+1)}_{\frac{k}{2^n},\frac{j}{2^m}}} \sup_{(u,v) \in \calS^{(n+1,m+1)}_{s,t}}|u - v|^{2H - 4} &\le \sup_{(s,t)\in \calS^{(n,m)}_{\frac{k}{2^n},\frac{j}{2^m}}}|s - t|^{2H - 4} \\&\le \left(\left|\frac{k}{2^n} - \frac{j+1}{2^m}\right| \wedge \left|\frac{k+1}{2^n} - \frac{j}{2^m}\right|\right)^{2H-4}.
		\end{split}
	\end{equation*}
	Inserting the above inequality to \eqref{eq partial bound 2} and taking $\rho_H= 2^{-4}\rho_2$ completes the proof. 
\end{proof}
Recall that $\bm w_n = \mathscr{C}_n \bm z_{n+2} = \left(w_{n,0},w_{n,1},\cdots,w_{n,2^n-1}\right)^\top$. The representation of $\mathscr{C}_n $ in \eqref{eq Cm} yields that 
\begin{equation}\label{eq w ni}
	w_{n,i} = \bm r^\top \left(z^{(n+2)}_{4i+1},z^{(n+2)}_{4i+2},z^{(n+2)}_{4i+3},z^{(n+2)}_{4i}\right)^\top = \frac{1}{4}\left(-z^{(n+2)}_{4i+1}+z^{(n+2)}_{4i+2}+z^{(n+2)}_{4i+3}-z^{(n+2)}_{4i+4}\right).
\end{equation}
Applying the bound in \Cref{Lemma Faber Bound}, we can obtain the following result, which concerns the upper bound between the Faber--Schauder coefficient $\theta_{n,k}$ and the Gaussian random variable $w_{n,i}$.

\begin{lemma}
	For $n \in \bN$, let $0 \le i,k\le 2^n-1$ be given such that $|i-k| > 1$. Then we have 
	\begin{equation}\label{eq Faber w bound}
		\big|\cov\big(\theta_{n,k},w_{n,i}\big)\big| \le \rho_H 2^{(1-2H)n}\left(\left|{k-i-1}\right| \wedge \left|{k-i+1}\right|\right)^{2H-4}.
	\end{equation}
\end{lemma}
\begin{proof}
	We first consider the case $i > k + 1$. We have for $1 \le \delta \le 4$, 
	\begin{equation*}
		\begin{split}
			\left|\cov\big(\theta_{n,k}, z^{(n+2)}_{4i+\delta}\big)\right|&= 2^{3(n+2)/2}\left|\cov\left(\theta_{n,k},\sum_{m = n+2}^\infty 2^{-3m/2}\sum_{j = 0}^{2^{m - (n+2)}-1} \theta_{m,j+2^{m-(n+2)}(4i+\delta-1)}\right)\right| \\
			&\le 2^{3(n+2)/2}\sum_{m = n+2}^\infty 2^{-3m/2}\sum_{j = 0}^{2^{m - (n+2)}-1} \left|\cov\big(\theta_{n,k},\theta_{m,j+2^{m-(n+2)}(4i+\delta-1)}\big)\right| \\
			& \le \rho_H \cdot 2^{3(n+2)/2}\sum_{m = n+2}^\infty 2^{-3m-3n/2}\sum_{j = 0}^{2^{m - (n+2)}-1}\left(\frac{i-k-1}{2^n} + \frac{\delta -1}{2^{n+2}} + \frac{j}{2^m}\right)^{2H-4}.
	\end{split}
	\end{equation*}
	As $\delta \ge 1$, $j \ge 0$ and $i > k+1$,  we have 
	\begin{equation*}
		\frac{i-k-1}{2^n} + \frac{\delta -1}{2^{n+2}} + \frac{j}{2^m} \ge \frac{i-k-1}{2^n}.
	\end{equation*}
	Hence, 
	\begin{equation*}
			\left|\cov\big(\theta_{n,k}, z^{(n+2)}_{4i+\delta}\big)\right|  \le \rho_H\cdot 2^{-n+1}\sum_{m = n+2}^\infty 2^{-2m}\left(\frac{i-k-1}{2^{n}}\right)^{2H-4}  \le \rho_H \cdot 2^{(1-2H)n}\left(i - k -1\right)^{2H-4}.
	\end{equation*}
	For the case $i < k -1$, an analogous argument yields that 
	\begin{equation*}
		\left|\cov\big(\theta_{n,k}, z^{(n+2)}_{4i+\delta}\big)\right| \le \rho_H \cdot 2^{3(n+2)/2}\sum_{m = n+2}^\infty 2^{-3m-3n/2}\sum_{j = 0}^{2^{m - (n+2)}-1}\left(\frac{k-i}{2^n} - \frac{\delta -1}{2^{n+2}} - \frac{j+1}{2^m}\right)^{2H-4}
	\end{equation*}
	As $j \le 2^{m-(n+2)} - 1$, $\delta \le 4$ and $k > i + 1$, we have 
	\begin{equation*}
		\frac{k-i}{2^n} - \frac{\delta -1}{2^{n+2}} - \frac{j+1}{2^m} \ge \frac{k - i -1}{2^n}
	\end{equation*}
	Hence, we have
	\begin{equation*}
		\left|\cov\big(\theta_{n,k}, z^{(n+2)}_{4i+\delta}\big)\right|  \le \rho_H\cdot 2^{-n+1}\sum_{m = n+2}^\infty 2^{-2m}\left(\frac{k-i-1}{2^{n}}\right)^{2H-4}  \le \rho_H \cdot 2^{(1-2H)n}\left(k-i -1\right)^{2H-4}.
	\end{equation*}
	Combining the above bounds for different cases, we have that for $|i -k| > 1$, 
	\begin{equation*}
		\left|\cov\big(\theta_{n,k}, z^{(n+2)}_{4i+\delta}\big)\right| \le \rho_H\cdot 2^{(1-2H)n}\left(\left|{k-i-1}\right| \wedge \left|{k-i+1}\right|\right)^{2H-4}. 
	\end{equation*}
	Finally, applying \eqref{eq w ni} gives 
	\begin{equation*}
		\big|\cov\big(\theta_{n,k},w_{n,i}\big)\big| \le \frac{1}{4}\sum_{\delta = 1}^{4}\left|\cov\big(\theta_{n,k}, z^{(n+2)}_{4i+\delta}\big)\right|  \le \rho_H \cdot 2^{(1-2H)n}\left(\left|{k-i-1}\right| \wedge \left|{k-i+1}\right|\right)^{2H-4}.
	\end{equation*}
	This completes the proof.
\end{proof}

Under the conditions of $X$ be as in \Cref{cor main rv}, \cite[Theorem 1.4]{HanSchiedGirsanov} states that the law of $X$ is absolutely continuous with respect to that of $W^H$. Thus, it suffices to establish \Cref{Theorem scale rate} for $X=W^H$. To this end, let $\bar{\bm \theta}_n$, $\bar{\bm \vartheta}_n$, and $\bm w_n$ be as defined in \eqref{eq vector} and \eqref{eq vector w}, respectively, associated with the fractional Brownian motion. Additionally, denote the covariance matrix of $\bar{\bm \vartheta}_n$ by $\Psi_n := \big(\psi^{(n)}_{i,j}\big)_{i,j = 0}^{2^n-1}$. The subsequent lemma explores the behavior of the traces of the covariance matrices $(\Psi_n)_{n \in \bN}$.

\begin{lemma}\label{lemma trace vartheta}
	For $n \in \bN$, we have 
	\begin{equation*}
		\trace{\Psi_{n+1}} = 2^{2-2H}\trace{\Psi_n}.
	\end{equation*}
\end{lemma}
\begin{proof}
	We begin with showing that the Faber--Schauder coefficients satisfy the following properties: 
	\begin{equation}\label{eq Faber WH}
		\cov\big(\theta_{n,k},\theta_{m,j}\big) = 2^{2H-1}\cov\left(\theta_{n+1,k},\theta_{m+1,j}\right) \quad  \text{and} \quad \cov\big(\theta_{n,k},\theta_{m,j}\big) = \cov\big(\theta_{n,k+i},\theta_{m,j+2^{m-n}i}\big).
	\end{equation}
	To show that the first identity holds, recall that $W^H$ is $H$-self-similar. That is, $(W^H_{\alpha t}) \overset{d}{=} (\alpha^HW^H_{t})$. This then implies
	\begin{equation*}
		\begin{split}
				\cov\big(\theta_{n,k},\theta_{m,j}\big) &= 2^{\frac{n+m}{2}}\cov\left(2W^H_{\frac{2k+1}{2^{n+1}}} - W^H_{\frac{2k}{2^{n+1}}}-W^H_{\frac{2k+2}{2^{n+1}}}, 2W^H_{\frac{2j+1}{2^{m+1}}} - W^H_{\frac{2j}{2^{m+1}}}-W^H_{\frac{2j+2}{2^{m+1}}}\right)\\&= 2^{2H+\frac{n+m}{2}}\cov\left(2W^H_{\frac{2k+2}{2^{n+1}}} - W^H_{\frac{2k}{2^{n+2}}}-W^H_{\frac{2k+2}{2^{n+2}}}, 2W^H_{\frac{2j+1}{2^{m+2}}} - W^H_{\frac{2j}{2^{m+2}}}-W^H_{\frac{2j+2}{2^{m+2}}}\right)\\&= 2^{2H-1}\cov\left(\theta_{n+1,k},\theta_{m+1,j}\right).
		\end{split}
	\end{equation*}
	To prove the second identity, take $k' = k + i$ and $j' = j + 2^{m-n}i$. Then $2^{-n}k - 2^{-m}j = 2^{-n}k' - 2^{-m}j'$. Since $W^H$ admits stationary increments, we have $\cov\big(\theta_{n,k},\theta_{m,j}\big) = \cov\big(\theta_{n,k'},\theta_{m,j'}\big)$. Moreover, the above properties remain to hold for the covariance  between $\theta_{n,k}$ and $w_{n,k}$, that is,
	\begin{equation}\label{eq cov W theta}
		\cov(\theta_{n,k}, w_{n,j}) = 2^{2H-1}\cov(\theta_{n+1,k}, w_{n+1,j}) \quad \text{and} \quad \cov(\theta_{n,k}, w_{n,j}) = \cov(\theta_{n,k+i}, w_{n,j+i}).
	\end{equation}
	To prove the first identity in \eqref{eq cov W theta}, note that for $1 \le \delta \le 4$
	\begin{equation*}
		\begin{split}
			\cov\big(\theta_{n+1,k}, z^{(n+3)}_{4j+\delta}\big)
			&= \cov\left(\theta_{n+1,k},\sum_{m = 0}^\infty 2^{-3m/2}\sum_{i= 0}^{2^{m}-1} \theta_{m+(n+3),i+2^{m}(4j+\delta-1)}\right)\\
			&= 2^{1-2H}\cov\left(\theta_{n,k},\sum_{m = 0}^\infty 2^{-3m/2}\sum_{i = 0}^{2^{m}-1} \theta_{m+(n+2),i+2^{m}(4j+\delta-1)}\right) = 2^{1-2H}\big(\theta_{n,k}, z^{(n+2)}_{4j+\delta}\big).
		\end{split}
	\end{equation*}
	In addition, let $k' = k + i$ and $j' = j+i$, then we have 
	\begin{equation*}
		\begin{split}
			\cov(\theta_{n,k'}, z^{(n+2)}_{4j'+\delta}) 	&= \sum_{m = 0}^\infty 2^{-3m/2}\sum_{i = 0}^{2^{m}-1}\cov\left(\theta_{n,k'}, \theta_{m+(n+2),i+2^{m}(4j'+\delta-1)}\right)\\
			&=  \sum_{m = 0}^\infty 2^{-3m/2}\sum_{i = 0}^{2^{m}-1}\cov\left(\theta_{n,k}, \theta_{m+(n+2),i+2^{m}(4j+\delta-1)}\right) = \cov(\theta_{n,k}, z^{(n+2)}_{4j+\delta}),\\
		\end{split}
	\end{equation*} 
	where the second equality follows from \eqref{eq Faber WH}. Since $w_{n,i} = (-z^{(n+2)}_{4i+1}+z^{(n+2)}_{4i+2}+z^{(n+2)}_{4i+3}-z^{(n+2)}_{4i+4})/4$ and covariance is a bilinear map, the identity \eqref{eq cov W theta} holds. Furthermore, an analogous argument leads to 
	\begin{equation}\label{eq cov w}
		\cov(w_{n,k}, w_{n,j}) = 2^{2H-1}\cov(w_{n+1,k}, w_{n+1,j}) \quad \text{and} \quad \cov(w_{n,k}, w_{n,j}) = \cov(w_{n,k+i}, w_{n,j+i}).
	\end{equation}
	Finally, \eqref{eq Faber WH}, \eqref{eq cov W theta} and \eqref{eq cov w} yield that 
	\begin{equation*}
		\begin{split}
			\trace \Psi_{n+1} &= \sum_{k = 0}^{2^{n+1}-1} \psi^{(n+1)}_{k,k} = \sum_{k = 0}^{2^{n+1}-1}\left(\bE\left[\theta^2_{n+1,k}\right] + \bE\left[w^2_{n+1,k}\right] + 2\bE\left[\theta_{n+1,k}w_{n+1,k}\right]\right)\\&= 2^{n+1+(1-2H)}\left(\bE\left[\theta^2_{n,0}\right]+\bE\left[w^2_{n,0}\right]+2\bE\left[\theta_{n,0}w_{n,0}\right]\right) = 2^{2-2H}\trace \Psi_{n}.
		\end{split}
	\end{equation*}
	This completes the proof. 
\end{proof}

To prove \Cref{Theorem scale rate}, it suffices to prove the following proposition. \Cref{Theorem scale rate} then simply follows from assertion \ref{scale-inv item d} in \Cref{scale-inv prop} and \cite[Theorem 1.4]{HanSchiedGirsanov}.

\begin{proposition}\label{previous prop 4.9}
	For $H \in (0,1)$, we take $\varepsilon_n = 2^{-n/2}n^{-1}\sqrt{\log n}$ and $Y^H = \int_{0}^{\cdot}W^H_s\,ds$. Then there exist $\tau_H > 0$ and a positive constant $c$, such that 
	\begin{equation*}
		\limsup_{n \ua \infty}\varepsilon_n^{-1}\left|H - \wh \cR_n\left(\frac{Y^H}{\sqrt{\tau_H}}\right)\right| \le c \qquad \bP\text{-}a.s.
	\end{equation*}
\end{proposition}
\begin{proof}
	Let us denote $\tau_H = \trace\Psi_0$. It is clear that $\trace\Psi_n = 2^{(2-2H)n}\tau_H$. Furthermore, as shown in \Cref{lemma fbm decaying}, as 
	\begin{equation*}
		\frac{\bE\big[w^2_{n,0}\big]}{\bE\big[\theta^2_{n,0}\big]} = \frac{\alpha(H)}{2^{2+2H}-2^{4H}} < 1. 
	\end{equation*}
	This then implies that $\tau_H > 0$. Furthermore, following from \Cref{lemma gamma ij}, \Cref{lemma fbm wn} and \Cref{lemma trace vartheta}, the covariance matrix $\Psi_n$ is a symmetric Toeplitz matrix and so $\psi^{(n)}_{i,j} = \psi^{(n)}_{j,i} = \psi^{(n)}_{0,|j-i|}$. Hence, we have $\norm{\Psi_n}_1 = \norm{\Psi_n}_\infty$, and applying inequalities \eqref{eq fbm variance}, \eqref{eq Faber bound} and \eqref{eq Faber w bound} yields that there exists a positive constant $\kappa_H > 0$ such that  
	\begin{equation*}
		\begin{split}
			\norm{\Psi_n}_2 &\le \sqrt{\norm{\Psi_n}_1\norm{\Psi}_\infty} = \norm{\Psi_n}_1 = \max_{0 \le j \le 2^n-1}\sum_{i = 0}^{2^n-1}|\psi^{(n)}_{i,j}| \le 2\sum_{i = 0}^{2^n-1}|\psi^{(n)}_{0,i}|\\& \le 2\left(|\psi^{(n)}_{0,0}|+ |\psi^{(n)}_{0,1}| + \sum_{i = 2}^{2^n-1}|\psi^{(n)}_{0,i}|\right) \le \kappa_H 2^{n(1-2H)},
		\end{split}
	\end{equation*}
	where the final inequality holds as it follows from \eqref{eq fbm variance}, \eqref{eq Faber bound} and \eqref{eq Faber w bound} that both $2^{n(2H-1)}|\bE[\theta_{n,0}\theta_{n,i}]|$, $2^{n(2H-1)}|\bE[w_{n,0}w_{n,i}]|$ and $2^{n(2H-1)}|\bE[\theta_{n,0}w_{n,i}]|$ are of order $\mathcal{O}(i^{2H-4})$ as $i \ua \infty$.
	
	Now, the concentration inequality \eqref{concentration inequality} yields that there exists a constant $\lambda > 0$ such that for all $n \in \bN$ and $\delta > 0$, 
	\begin{equation*}
		\begin{split}
			\bP\left[\left|2^{n(H-1)}\norm{\frac{\bar{\bm \vartheta}_n}{\sqrt{\tau_H}}}_{\ell_2} - 1\right|\ge \delta \right] &= \bP\left[\left|2^{n(H-1)}\norm{\bar{\bm \vartheta}_n}_{\ell_2} - \sqrt{\tau_H}\right|\ge \delta\sqrt{\tau_H} \right] \\&= \bP\left[\left|\norm{\bar{\bm \vartheta}_n}_{\ell_2} - \sqrt{\trace \Psi_n}\right|\ge 2^{n(1-H)}\delta\tau_H \right] \le \lambda \exp\left(-\frac{2^{n(2-2H)}\delta^2\tau^2_H}{4\norm{\Psi_n}_2}\right).
		\end{split}
	\end{equation*}
	Taking $\delta_n := 2^{-n/2}\tau^{-1}_H\sqrt{8\kappa_H\log n}$ yields that 
	\begin{equation*}
		\bP\left[\left|2^{n(H-1)}\norm{\frac{\bar{\bm \vartheta}_n}{\sqrt{\tau_H}}}_{\ell_2} - 1\right|\ge \delta \right]\le \frac{\lambda}{n^2}.
	\end{equation*}
	As the expression on the right-hand side of the above inequality is absolutely summable, the Borel--Cantelli lemma yields that 
	\begin{equation*}
		\limsup_{n \ua \infty}\delta_n^{-1}\left|2^{n(H-1)}\norm{\frac{\bar{\bm \vartheta}_n}{\sqrt{\tau_H}}}_{\ell_2} - 1\right| \le 1 \qquad \bP\text{-}a.s.
	\end{equation*}
	Note that\label{previous page 36}
	\begin{equation*}
		H - \wh \cR_n\left(\frac{Y^H}{\sqrt{\tau_H}}\right) = (H-1) + \frac{1}{n}\log_2\norm{\frac{\bar{\bm \vartheta}_n}{\sqrt{\tau_H}}}_{\ell_2} = \frac{1}{n}\log_2 \left(2^{n(H-1)}\norm{\frac{\bar{\bm \vartheta}_n}{\sqrt{\tau_H}}}_{\ell_2}\right),
	\end{equation*}
	and now using the fact $\sqrt{1 + x} - 1 \sim x/2$ as $x \rightarrow 0$ yields the result. This completes the proof. 
\end{proof}

\subsection{Proof of results in \Cref{RV section}}
Before presenting our proofs, we will outline our approach to establishing the results in \Cref{RV section}. It follows from \eqref{eq process RV} that
\begin{equation*}
	\wh Y^{(n)}_{\frac{k}{2^n}} - \wh Y^{(n)}_{\frac{k-1}{2^n}} = \sum_{j = m_n(k-1)+1}^{m_n k} \left(\log S_{\frac{j}{2^n m_n}} - \log S_{\frac{j-1}{2^n m_n}} \right)^2 =  \sum_{j = m_n(k-1)+1}^{m_n k} \left(\int^{\frac{j}{2^n m_n}}_{\frac{j-1}{2^n m_n}}\sigma_s dB_s\right)^2.
\end{equation*}
For short, let us denote 
\begin{equation*}
	\xi^{(n)}_{j} := \left(\int^{\frac{j}{2^n m_n}}_{\frac{j-1}{2^n m_n}}\sigma_s \,dB_s\right)^2 - \int^{\frac{j}{2^n m_n}}_{\frac{j-1}{2^n m_n}}\sigma^2_s\,ds \quad \text{and} \quad Z^{(n)}_k := \sum_{j = m_n(k-1)+1}^{m_n k}\xi^{(n)}_{j}.
\end{equation*}
We can then obtain the following decomposition
\begin{equation*}
	\begin{split}
		\wh Y^{(n)}_{\frac{k}{2^n}} - \wh Y^{(n)}_{\frac{k-1}{2^n}} &= \int_{\frac{k-1}{2^n}}^{\frac{k}{2^n}} \sigma^2_s \,ds + \sum_{j = m_n(k-1)+1}^{m_n k} \left[\left(\int^{\frac{j}{2^n m_n}}_{\frac{j-1}{2^n m_n}}\sigma_s dB_s\right)^2 - \int^{\frac{j}{2^n m_n}}_{\frac{j-1}{2^n m_n}}\sigma^2_s \,ds\right]\\&= \left(Y_{\frac{k}{2^n}} - Y_{\frac{k-1}{2^n}}\right) + \sum_{j = m_n(k-1)+1}^{m_n k} \xi^{(n)}_j = \left(Y_{\frac{k}{2^n}} - Y_{\frac{k-1}{2^n}}\right) + Z^{(n)}_k.
	\end{split}
\end{equation*}
Clearly, if the quantity $Z^{(n)}_k$ is relatively small compared with $Y_{\frac{k}{2^n}} - Y_{\frac{k-1}{2^n}}$ as $n \ua \infty$. Hence,  the difference $\wh Y^{(n)}_{\frac{k}{2^n}} - \wh Y^{(n)}_{\frac{k-1}{2^n}}$ will be sufficiently close to the difference of $Y_{\frac{k}{2^n}} - Y_{\frac{k-1}{2^n}}$, and so will be  the values of $\vartheta_{n,k}$ and $\wt \vartheta_{n,k}$. Hence,  $\wh \cR_n(Y)$ and $\wt \cR_n(S)$ will eventually converge to the same value.

To prove \Cref{thm main realized}, the following lemma derives upper bounds for $\bE\big[\big(Z^{(n)}_{k}\big)^2\big]$ and $\var[\big(Z^{(n)}_{k}\big)^2]$. 

\begin{lemma}\label{lemma bound RV}
	Suppose that $\bE\left[\sup_{t \in [0,1]} \sigma^{16}_t\right] < \infty$. Then, there exists a positive constant $\kappa$ such that for $n \in \bN$ and $1 \le k \le 2^n$ such that   
	\begin{equation}\label{eq bound RV}
		\bE\left[\left(Z^{(n)}_{k}\right)^2\right] \le \frac{\kappa}{2^{2n}m_n} \quad \text{and} \quad  \var\left[\left(Z^{(n)}_k\right)^2\right] \le \bE\left[\left(Z^{(n)}_{k}\right)^4\right] \le \frac{\kappa}{2^{4n}m_n^2}.
	\end{equation}
\end{lemma}
\begin{proof}
	It is clear that 
	\begin{equation*}
		\left(Z^{(n)}_{k}\right)^2 =\left(\sum_{j = m_n(k-1)+1}^{m_n k} \xi^{(n)}_j\right)^2 = \sum_{j = m_n(k-1)+1}^{m_n k}\left(\xi^{(n)}_j\right)^2 + 2\sum_{m_n(k-1)+1 \le i < j \le m_n k}\xi^{(n)}_{i}\xi^{(n)}_{j}.
	\end{equation*}
	Let us now consider the natural filtration $\cF_t := \sigma(\{B_s, 0 \le s \le t\})$. Note that for $i < j$, we have 
	\begin{equation*}
		\bE\left[\xi^{(n)}_{i}\xi^{(n)}_{j}\right] = \bE\left[\bE\left[\xi^{(n)}_{i}\xi^{(n)}_{j}\Big|\cF_{\frac{i}{2^n m_n}} \right]\right] = \bE\left[\xi^{(n)}_{i}\bE\left[\xi^{(n)}_{j}\Big|\cF_{\frac{i}{2^n m_n}}\right]\right].
	\end{equation*}
	Given that  $\sigma_t$ is progressively measurable, and therefore adapted to $\cF_t$, and since $\bE\left[\int_0^1 \sigma_s^{2}\,ds\right] < \infty$, we can apply It\^{o}'s isometry to conclude that 
	\begin{equation}\label{eq martingale diff}
		\begin{split}
			\bE\left[\xi^{(n)}_{j}\Big|\cF_{\frac{i}{2^n m_n}}\right] &= \bE\left[\left(\int^{\frac{j}{2^n m_n}}_{\frac{j-1}{2^n m_n}}\sigma_s \,dB_s\right)^2 - \int^{\frac{j}{2^n m_n}}_{\frac{j-1}{2^n m_n}}\sigma^2_s\,ds\bigg|\cF_{\frac{i}{2^n m_n}}\right]\\&= \bE\left[\int^{\frac{j}{2^n m_n}}_{\frac{j-1}{2^n m_n}}\sigma^2_s\,ds\bigg|\cF_{\frac{i}{2^n m_n}}\right]
			-\bE\left[\int^{\frac{j}{2^n m_n}}_{\frac{j-1}{2^n m_n}}\sigma^2_s\,ds\bigg|\cF_{\frac{i}{2^n m_n}}\right] = 0.
		\end{split}
	\end{equation}
	Hence, we have $\bE\left[\xi^{(n)}_{i}\xi^{(n)}_{j}\right] = 0$. It remains to compute $\bE\left[(\xi^{(n)}_j)^2\right]$. Applying the binomial expansion and the Cauchy inequality to $\bE\left[(\xi^{(n)}_j)^2\right]$ yields 
	\begin{equation*}
		\scalemath{0.85}{\begin{split}
				\bE\left[\left(\xi^{(n)}_{j}\right)^2\right] &= \bE\left[\left(\int^{\frac{j}{2^n m_n}}_{\frac{j-1}{2^n m_n}}\sigma_s\,dB_s\right)^4\right] - 2\bE\left[\left(\int^{\frac{j}{2^n m_n}}_{\frac{j-1}{2^n m_n}}\sigma_s\,dB_s\right)^2\int^{\frac{j}{2^n m_n}}_{\frac{j-1}{2^n m_n}}\sigma^2_s\,ds\right] + \bE\left[\left(\int^{\frac{j}{2^n m_n}}_{\frac{j-1}{2^n m_n}}\sigma^2_s\,ds\right)^2\right] \\&\le \bE\left[\left(\int^{\frac{j}{2^n m_n}}_{\frac{j-1}{2^n m_n}}\sigma_s\,dB_s\right)^4\right] + 2\sqrt{\bE\left[\left(\int^{\frac{j}{2^n m_n}}_{\frac{j-1}{2^n m_n}}\sigma_s\,dB_s\right)^4\right]}\sqrt{\bE\left[\left(\int^{\frac{j}{2^n m_n}}_{\frac{j-1}{2^n m_n}}\sigma^2_s\,ds\right)^2\right]}+ \bE\left[\left(\int^{\frac{j}{2^n m_n}}_{\frac{j-1}{2^n m_n}}\sigma^2_s\,ds\right)^2\right].
		\end{split}}
	\end{equation*}
	The Burkholder--Davis--Gundy inequality yields  a constant $\kappa_1 > 0$ such that 
	\begin{equation*}
		\bE\left[\left(\int^{\frac{j}{2^n m_n}}_{\frac{j-1}{2^n m_n}}\sigma_s\,dB_s\right)^4\right] \le \kappa_1 \bE\left[\left(\int^{\frac{j}{2^n m_n}}_{\frac{j-1}{2^n m_n}}\sigma^2_s\,ds\right)^2\right].
	\end{equation*}
	Setting $M_1 := \bE\left[\left(\sup_{s \in [0,1]} \sigma^2_s \right)^2 \right] < \infty$, it follows that 
	\begin{equation*}
		\begin{split}
			\bE\left[\left(\xi^{(n)}_{j}\right)^2\right] \le \big(\sqrt{\kappa_1} + 1\big)^2\bE\left[\left(\int^{\frac{j}{2^n m_n}}_{\frac{j-1}{2^n m_n}}\sigma^2_s\,ds\right)^2\right]\le \big(\sqrt{\kappa_1} + 1\big)^2 \frac{\bE\left[\left(\sup_{s \in [0,1]}\sigma_s^2\right)^2\right]}{2^{2n}m_n^2} \le \frac{\big(\sqrt{\kappa_1} + 1\big)^2 M_1}{2^{2n}m_n^2}.
		\end{split}
	\end{equation*}
	Hence,
	\begin{equation*}
		\bE\left[\left(Z^{(n)}_{k}\right)^2\right] = \sum_{j = m_n(k-1)+1}^{m_n k}\bE\left[\left(\xi^{(n)}_{j}\right)^2\right] \le  \frac{\big(\sqrt{\kappa_1} + 1\big)^2 M_1}{2^{2n}m_n}.
	\end{equation*}
	Next, we are going to provide an upper bound for $\bE\left[\big(Z^{(n)}_{k}\big)^4\right]$. It follows from the multinomial expansion and a similar argument as in \eqref{eq martingale diff} that
	\begin{equation}\label{eq Zn fourth}
		\scalemath{0.9}{\begin{split}
				\bE\left[\left(Z^{(n)}_k\right)^4\right] &= \sum_{j = m_n(k-1) + 1}^{m_n k} \bE\left[\left(\xi^{(n)}_j\right)^4\right] + 6 \sum_{m_n(k-1)+1 \le i < j \le m_n k} \bE\left[\left(\xi^{(n)}_i\right)^2\left(\xi^{(n)}_j\right)^2\right] \\&\le \sum_{j = m_n(k-1) + 1}^{m_n k} \bE\left[\left(\xi^{(n)}_j\right)^4\right] + 6 \sum_{m_n(k-1)+1 \le i < j \le m_n k} \sqrt{\bE\left[\left(\xi^{(n)}_i\right)^4\right]\bE\left[\left(\xi^{(n)}_j\right)^4\right]}.
		\end{split}}
	\end{equation}
	Furthermore, it follows from the Cauchy--Schwarz inequality that 
	\begin{equation*}
		\begin{split}
			\bE\left[\left(\xi^{(n)}_{j}\right)^4\right] &= \bE\left[\left(\left(\int^{\frac{j}{2^n m_n}}_{\frac{j-1}{2^n m_n}}\sigma_s\,dB_s\right)^2 - \int^{\frac{j}{2^n m_n}}_{\frac{j-1}{2^n m_n}}\sigma^2_s\,ds\right)^4\right]\\&\le \sum_{i = 0}^{4}{4 \choose i}\bE\left[\left(\int^{\frac{j}{2^n m_n}}_{\frac{j-1}{2^n m_n}}\sigma_s\,dB_s\right)^{8-2i}\left(\int^{\frac{j}{2^n m_n}}_{\frac{j-1}{2^n m_n}}\sigma^2_s\,ds\right)^{i}\right]\\&\le \sum_{i = 0}^{4}{4 \choose i}\sqrt{\bE\left[\left(\int^{\frac{j}{2^n m_n}}_{\frac{j-1}{2^n m_n}}\sigma_s\,dB_s\right)^{16-4i}\right]}\sqrt{\bE\left[\left(\int^{\frac{j}{2^n m_n}}_{\frac{j-1}{2^n m_n}}\sigma^2_s\,ds\right)^{2i}\right]}
		\end{split}
	\end{equation*}
	The Burkholder--Davis--Gundy inequality yields that there exists $\kappa_2 > 0$ such that for $0 \le i \le 4$, 
	\begin{equation*}
		\bE\left[\left(\int^{\frac{j}{2^n m_n}}_{\frac{j-1}{2^n m_n}}\sigma_s\,dB_s\right)^{16-4i}\right] \le \kappa_2 \bE\left[\left(\int^{\frac{j}{2^n m_n}}_{\frac{j-1}{2^n m_n}}\sigma^2_s\,ds\right)^{8-2i}\right].
	\end{equation*}
	Letting $M_2 := \sup\limits_{i = 0,\cdots,4} \bE\left[\left(\sup\limits_{t \in [0,1]}\sigma_t^2\right)^{2i}\right] < \infty$,  we  have  
	\begin{equation*}
		\begin{split}
			\bE\left[\left(\xi^{(n)}_{j}\right)^4\right] &\le \sqrt{\kappa_2}\sum_{i = 0}^{4}{4 \choose i}\sqrt{\bE\left[\left(\int^{\frac{j}{2^n m_n}}_{\frac{j-1}{2^n m_n}}\sigma^2_s \,ds\right)^{8-2i}\right]}\sqrt{\bE\left[\left(\int^{\frac{j}{2^n m_n}}_{\frac{j-1}{2^n m_n}}\sigma^2_s\,ds\right)^{2i}\right]}\\
			&\le \sqrt{\kappa_2} \sum_{i = 0}^{4} {4 \choose i} \frac{\bE\left[\left(\sup_{t \in [0,1]}\sigma^2_t\right)^{8 - 2i}\right]}{(2^n m_n)^{4-i}}\frac{\bE\left[\left(\sup_{t \in [0,1]}\sigma^2_t\right)^{2i}\right]}{(2^n m_n)^{i}} \\
			&\le \sqrt{\kappa_2} \sum_{i = 0}^{4} {4 \choose i} \frac{M_2}{\left(2^n m_n\right)^{4-i}}\frac{M_2}{\left(2^n m_n\right)^{i}}  = \frac{16\sqrt{\kappa_2}M_2^{2}}{2^{4n}m_n^4}.
		\end{split}
	\end{equation*}
	Inserting the above inequality into \eqref{eq Zn fourth} yields that 
	\begin{equation*}
		\var\left(\left(Z^{(n)}_k\right)^2\right) \le \bE\left[\left(Z^{(n)}_k\right)^4\right] \le 16\sqrt{\kappa_2}M_2^2\left(\frac{1}{2^{4n}m_n^3} + \frac{6}{2^{4n}m_n^2}\right).
	\end{equation*}
	As $m_n \ge 1$, taking $\kappa:= (\sqrt{\kappa_1} + 1)^2 M_1 \vee 96 \sqrt{\kappa_2}M_2^2$ yields back \eqref{eq bound RV}. This completes the proof.
\end{proof}
\begin{proof}[Proof of \Cref{thm main realized}]
	Letting $\bm a := (-1,+1,+1,-1)^\top$, we can write	\begin{equation*}
		\begin{split}
			\wt \vartheta_{n,k} &= 2^{3n/2+3}\left(\wh Y^{(n+2)}_{\frac{4k}{2^{n+2}}} - 2\wh Y^{(n+2)}_{\frac{4k+1}{2^{n+2}}} + 2\wh Y^{(n+2)}_{\frac{4k+3}{2^{n+2}}} - \wh Y^{(n+2)}_{\frac{4k+4}{2^{n+2}}}\right)= 2^{3n/2+3} \sum_{i = 1}^{4} a_i \left(\wh Y^{(n+2)}_{\frac{4k+i}{2^{n+2}}}-\wh Y^{(n+2)}_{\frac{4k+i-1}{2^{n+2}}}\right)\\& = 2^{3n/2+3}\left( \sum_{i = 1}^{4} a_i \left( Y_{\frac{4k+i}{2^{n+2}}}- Y_{\frac{4k+i-1}{2^{n+2}}}\right) + \sum_{i = 1}^{4} a_i Z^{(n+2)}_{4k+i}\right) = \vartheta_{n,k} + 2^{3n/2+3}\sum_{i = 1}^{4} a_i Z^{(n+2)}_{4k+i}.
		\end{split}
	\end{equation*}
	Denoting $\zeta_{n,k} := 2^{3n/2+3}\sum_{i = 1}^{4} a_i Z^{(n+2)}_{4k+i}$,  it follows from Jensen's inequality that 
	\begin{equation*}
		\zeta_{n,k}^2 = 2^{3n + 6}\left(\sum_{i = 1}^{4} a_i Z^{(n+2)}_{4k+i}\right)^2 \le 2^{3n+8}\sum_{i = 1}^{4}\left(Z^{(n+2)}_{4k+i}\right)^2.
	\end{equation*}
	Hence, we have 
	\begin{equation*}
		\bE\left[2^{(2R-2)n}\sum_{k = 0}^{2^n-1}\zeta_{n,k}^2\right] = 2^{(2R-2)n} 2^{3n+8}\sum_{k = 1}^{2^{n+2}}\bE\left[\left(Z^{(n+2)}_{k}\right)^2\right] \le 2^{(2R-2)n + 3n+8 +(n+2)}\frac{\kappa}{2^{2(n+2)}m_{n+2}} = \frac{2^{2Rn+4}\kappa}{m_{n+2}}.
	\end{equation*}
	Furthermore, the Cauchy--Schwarz inequality yields that 
	\begin{equation*}
		\begin{split}
			\var\left[\zeta_{n,k}^2\right] &\le \bE\left[\zeta^4_{n,k}\right] = 2^{6n+12}\bE\left[\left(\sum_{i = 1}^{4}a_iZ^{(n+2)}_{4k+i}\right)^4\right] \le 2^{6n+20} \sum_{i = 1}^{4}\bE\left[\left(Z^{(n+2)}_{4k+i}\right)^4\right]\\& \le 2^{6n + 20}\cdot \frac{\kappa}{2^{4(n+2)}m_{n+2}^2} = \frac{\kappa \cdot 2^{2n+12}}{m_{n+2}^2}.
		\end{split}
	\end{equation*}
	Applying the above inequality yields that 
	\begin{equation*}
		\begin{split}
			\var\left[2^{(2R-2)n}\sum_{k = 0}^{2^n-1}\zeta_{n,k}^2\right] &= 2^{(4R -4)n}\sum_{k = 0}^{2^{n}-1}\left(\var\left[\zeta_{n,k}^2\right] + \sum_{0 \le i \neq j \le 2^{n}-1} \cov\left(\zeta_{n,k}^2,\zeta_{n,j}^2\right)\right)\\&\le 2^{(4R - 4)n}\sum_{k = 0}^{2^{n}-1}\left(\var\left[\zeta_{n,i}^2\right] + \sum_{0 \le i \neq j \le 2^{n}-1} \sqrt{\var\left[\zeta_{n,i}^2\right]\var\left[\zeta_{n,j}^2\right]}\right) \le \frac{2^{4Rn+12}\kappa}{m^2_{n+2}}.
		\end{split}
	\end{equation*}
	Since $\lim_n n^{-1}\log_2 m_n > 2R$, we must have $2^{2Rn}/m_n \rightarrow \infty$ as $n \ua \infty$. As a result, we have 
	\begin{equation*}
		\lim_{n \ua \infty} \bE\left[2^{(2R-2)n}\sum_{k = 0}^{2^n-1}\zeta^2_{n,k}\right] = \lim_{n \ua \infty} \var\left[2^{(2R-2)n}\sum_{k = 0}^{2^n-1}\zeta^2_{n,k}\right] = 0.
	\end{equation*}
	Hence, it follows from the fast $L_2$-convergence that
	\begin{equation*}
		2^{(2R-2)n}\sum_{k = 0}^{2^n-1}\zeta^2_{n,k} \lra 0 \qquad \bP\text{-a.s.},
	\end{equation*}
	Finally, note that 
	\begin{equation*}
		2^{(2R-2)n}\sum_{k = 0}^{2^n-1}\wt \vartheta^2_{n,k} = 2^{(2R-2)n}\left(\sum_{k = 0}^{2^n-1}\vartheta^2_{n,k} + \sum_{k = 0}^{2^n-1}\zeta^2_{n,k} + 2\sum_{k = 0}^{2^n-1}\vartheta_{n,k}\zeta_{n,k}\right),
	\end{equation*}
	and it remains to show that the third term on the right-hand side  converges to zero as $n \ua \infty$. The Cauchy--Schwarz inequality yields
	\begin{equation*}
		\limsup_{n \ua \infty}2^{(2R-2)n}\left|\sum_{k = 0}^{2^n-1}\vartheta_{n,k}\zeta_{n,k}\right| \le \sqrt{\limsup_{n \ua \infty}2^{(2R-2)n}\sum_{k = 0}^{2^n-1}\vartheta^2_{n,k}}\sqrt{\lim_{n \ua \infty}2^{(2R-2)n}\sum_{k = 0}^{2^n-1}\zeta^2_{n,k}} = 0.
	\end{equation*} 
	Hence, we get 
	\begin{equation*}
		\limsup_{n \ua \infty}2^{(2R-2)n}\sum_{k = 0}^{2^n-1}\wt \vartheta^2_{n,k} = \limsup_{n \ua \infty}2^{(2R-2)n}\sum_{k = 0}^{2^n-1} \vartheta^2_{n,k}.
	\end{equation*}
	Moreover, in an analogous manner, we can obtain that 
	\begin{equation*}
		\liminf_{n \ua \infty}2^{(2R-2)n}\sum_{k = 0}^{2^n-1}\wt \vartheta^2_{n,k} = \liminf_{n \ua \infty}2^{(2R-2)n}\sum_{k = 0}^{2^n-1} \vartheta^2_{n,k},
	\end{equation*}
	Taking logarithms on both sides of the above identities, dividing by $n$ and letting $n \ua \infty$ complete the proof.
\end{proof}

\begin{proof}[Proof of \Cref{thm main realized necessary}]
	It suffices to consider the case $\alpha < R$; the  case $\alpha > R$ follows directly from \Cref{thm main realized}. Due to the independence between $(\sigma_t)_{t \in [0,1]}$ and $(B_t)_{t \in [0,1]}$, it suffices to show that for a typical sample path $t \mapsto \sigma_t$, we have $\wt \cR_n(S) \rightarrow R$ $\bQ$-a.s. Recall that 
	\begin{equation*}
		\xi^{(n)}_{j} := \left(\int^{\frac{j}{2^n m_n}}_{\frac{j-1}{2^n m_n}}\sigma_s \,dB_s\right)^2 - \int^{\frac{j}{2^n m_n}}_{\frac{j-1}{2^n m_n}}\sigma^2_s\,ds \quad \text{and} \quad Z^{(n)}_k = \sum_{j = m_n(k-1)+1}^{m_n k}\xi^{(n)}_{j}.
	\end{equation*}
	For a given sample path $t \rightarrow \sigma_t$, we have $\xi^{(n)}_i \indep \xi^{(n)}_j$ for $i \neq j$. Let $\{U^{(n)}_{j}, 1 \le j \le 2^nm_n, n \in \bN\}$ be an i.i.d.~array of $\calN(0,1)$-random variables defined on the probability space $(\wt \Omega, \cG, \bQ)$. Then, we have 
	\begin{equation*}
		\xi^{(n)}_j \overset{d}{=} \left[\left(U^{(n)}_j\right)^2 - 1\right] \cdot \int^{\frac{j}{2^n m_n}}_{\frac{j-1}{2^n m_n}}\sigma^2_s\,ds.
	\end{equation*}
	Therefore, we have 
	\begin{equation}\label{eq Znk RV}
		\begin{split}
			\bE\left[\left(Z^{(n)}_{k}\right)^2\right] &= \bE\left[\left(\sum_{j = m_n(k-1)+1}^{m_n k} \xi^{(n)}_j\right)^2\right] = \sum_{j = m_n(k-1)+1}^{m_n k}\bE\left[\left(\xi^{(n)}_{j} \right)^2\right] \\&= \sum_{j = m_n(k-1)+1}^{m_n k}\bE\left[\left(\left(U^{(n)}_j\right)^2 - 1\right)^2\right]\cdot\left(\int^{\frac{j}{2^n m_n}}_{\frac{j-1}{2^n m_n}}\sigma^2_s\,ds\right)^2 \\&= 2 \sum_{j = m_n(k-1)+1}^{m_n k}\left(\int^{\frac{j}{2^n m_n}}_{\frac{j-1}{2^n m_n}}\sigma^2_s\,ds\right)^2.
		\end{split}
	\end{equation} 
	Let $\bm a = (-1,+1,+1,-1)^\top$, and recall that 
	\begin{equation*}
		\zeta_{n,k} = 2^{3n/2+3}\sum_{i = 1}^{4} a_i Z^{(n+2)}_{4k+i}.
	\end{equation*}
	Since $\xi^{(n)}_{i} \indep \xi^{(n)}_j$ for $i \neq j$, then we have $Z^{(n)}_i \indep Z^{(n)}_j$ for $1 \le i \neq j \le 2^n$. Thus, it follows that 
	\begin{equation*}
		\bE\left[\zeta_{n,k}^2\right] = 2^{3n+6}\bE\left[\left(\sum_{i = 1}^{4} a_i Z^{(n+2)}_{4k+i}\right)^2\right] = 2^{3n+6}\sum_{i = 1}^{4}\bE\left[(Z^{(n+2)}_{4k+i})^2\right].
	\end{equation*}
	Thus, applying \eqref{eq Znk RV} to the above equation yields that 
	\begin{equation}\label{eq sum xi square}
		\begin{split}
			\bE\left[\sum_{k = 0}^{2^n-1}\zeta_{n,k}^2\right] &= 2^{3n+6}\sum_{k = 1}^{2^{n+2}}\bE\left[\left(Z^{(n+2)}_k\right)^2\right] = 2^{3n+7}\sum_{k = 1}^{2^{n+2}}\sum_{j = m_{n+2}(k-1)+1}^{m_{n+2} k}\left(\int^{\frac{j}{2^{n+2} m_{n+2}}}_{\frac{j-1}{2^{n+2} m_{n+2}}}\sigma^2_s\,ds\right)^2 \\&= 2^{3n+7}\sum_{j = 1}^{2^{n+2}m_{n+2}}\left(\int^{\frac{j}{2^{n+2} m_{n+2}}}_{\frac{j-1}{2^{n+2} m_{n+2}}}\sigma^2_s\,ds\right)^2.
		\end{split}
	\end{equation}
	For $n \in \bN$, the intermediate value theorem and the mean-value theorem also imply that there exist intermediate times $s_j \in [\frac{j-1}{2^{n+2}m_{n+2}}, \frac{j}{2^{n+2}m_{n+2}}]$ such that 
	\begin{equation*}
		2^{n+2}m_{n+2}\int^{\frac{j}{2^{n+2} m_{n+2}}}_{\frac{j-1}{2^{n+2} m_{n+2}}}\sigma^2_s\,ds = \sigma^2_{s_j}.
	\end{equation*}
	Applying the above identity to \eqref{eq sum xi square} yields 
	\begin{equation*}
		\begin{split}
			\left(\frac{2^{2n+5}}{m_{n+2}}\right)^{-1}\bE\left[\sum_{k = 0}^{2^n-1}\zeta_{n,k}^2\right] &= 2^{n+2}m_{n+2}\sum_{j = 1}^{2^{n+2}m_{n+2}} \left(\frac{\sigma^2_{s_j}}{2^{n+2}m_{n+2}}\right)^2 = \frac{1}{2^{n+2}m_{n+2}}\sum_{j = 1}^{2^{n+2}m_{n+2}}\sigma^4_{s_j} \\
			& \longrightarrow \int_0^1 \sigma^4_s \, ds, \qquad \quad \text{as} \quad n \ua \infty.
		\end{split}
	\end{equation*}
	Now, let us evaluate $\var\big(\sum_{k = 0}^{2^n-1}\zeta_{n,k}^2\big)$. Since $Z_i^{(n+2)} \indep Z_j^{(n+2)}$ for $1 \le i \neq j \le 2^{n+2}$, we have $\xi_{n,k} \indep \xi_{n,j}$ for $1 \le j \neq k \le 2^n$. Thus, 
	\begin{equation*}
		\var\left[\sum_{k = 0}^{2^n-1}\zeta_{n,k}^2\right] = \sum_{k = 0}^{2^n-1}\var\left[\zeta_{n,k}^2\right] \le \sum_{k = 0}^{2^n-1}\bE\left[\zeta^4_{n,k}\right].
	\end{equation*}  
	To bound the fourth moment $\bE\left[\zeta^4_{n,k}\right]$, note that by Jensen's inequality, we have 
	\begin{equation*}
		\bE\left[\zeta_{n,k}^4\right] = 2^{6n+12}\bE\left[\left(\sum_{i = 1}^{4}a_iZ^{(n+2)}_{4k+i}\right)^4\right] \le 2^{6n+18} \sum_{i = 1}^{4}\bE\left[\left(Z^{(n+2)}_{4k+i}\right)^4\right] \le \kappa \cdot \frac{2^{2n+12}}{m^2_{n+2}}.
	\end{equation*}
		Therefore, we have 
	\begin{equation*}
		\var\left[\sum_{k = 0}^{2^n-1}\zeta_{n,k}^2\right] \le \sum_{k = 0}^{2^n-1}\kappa \cdot \frac{2^{2n+12}}{m^2_{n+2}} \le \kappa \cdot \frac{2^{3n+12}}{m^2_{n+2}}.
	\end{equation*}
	Now, we first prove that as $n \ua \infty$,
	\begin{equation*}
		\frac{m_{n+2}}{2^{2n+5}}\sum_{k = 0}^{2^n-1}\zeta_{n,k}^2 \longrightarrow \int_0^1 \sigma_s^4\,ds, \qquad \bQ\text{-a.s.}
	\end{equation*}
	This is because that 
	\begin{equation*}
		\var\left(\frac{m_{n+2}}{2^{2n+5}}\sum_{k = 0}^{2^n-1}\zeta_{n,k}^2\right) \le \kappa \cdot \frac{m^2_{n+2}}{2^{4n+10}} \cdot \frac{2^{3n+16}}{m^2_{n+2}} = \frac{\kappa}{2^{n-6}} \longrightarrow 0.
	\end{equation*}
	Hence, it follows from  fast $L_2$-convergence that $\frac{m_{n+2}}{2^{2n+5}}\sum_{k = 0}^{2^n-1}\zeta_{n,k}^2$ converges to $\int_0^1\sigma_s^4\, ds$ with probability one. Now, recall that 
	\begin{equation*}
		\sum_{k = 0}^{2^n-1}\wt \vartheta^2_{n,k} = \sum_{k = 0}^{2^n-1}\vartheta^2_{n,k} + \sum_{k = 0}^{2^n-1}\zeta^2_{n,k} + 2\sum_{k = 0}^{2^n-1}\vartheta_{n,k}\zeta_{n,k}.
	\end{equation*} 
	Since $\lim_n n^{-1}\log_2 m_n = 2\alpha < 2R$, there exists $n_0 \in \bN$ such that $m_n \le 2^{2(R-\varepsilon)n}$ for $0 < \varepsilon < (R - \alpha)/2$. 
	\begin{equation*}
		\lim_{n \ua \infty}\frac{m_{n+2}}{2^{2n+5}}\sum_{k = 0}^{2^n-1}\vartheta^2_{n,k} \le \lim_{n \ua \infty} 2^{-2\varepsilon n} \cdot \lim_{n \ua \infty} 2^{(2R-2)n} \sum_{k = 0}^{2^n-1}\vartheta^2_{n,k} = 0.
	\end{equation*}
	Furthermore, the Cauchy--Schwarz inequality yields that 
	\begin{equation*}
		\lim_{n \ua \infty}\frac{m_{n+2}}{2^{2n+5}}\left|\sum_{k = 0}^{2^n-1}\vartheta_{n,k}\zeta_{n,k}\right| \le \sqrt{\lim_{n \ua \infty}\frac{m_{n+2}}{2^{2n+5}}\sum_{k = 0}^{2^n-1}\zeta^2_{n,k}}\cdot\sqrt{\lim_{n \ua \infty}\frac{m_{n+2}}{2^{2n+5}}\sum_{k = 0}^{2^n-1}\vartheta^2_{n,k}}  = 0.
	\end{equation*}
	Hence, we have 
	\begin{equation*}
		\lim_{n \ua \infty}\frac{m_{n+2}}{2^{2n+5}}\sum_{k = 0}^{2^n-1}\wt \vartheta^2_{n,k} = \int_0^1 \sigma^4_s \, ds. 
	\end{equation*}
	Thus, we have 
	\begin{equation*}
		\begin{split}
			\lim_{n \ua \infty} \wt \cR_n(S) &= 1 - \lim_{n \ua \infty} \frac{1}{n}\log_2 \sqrt{\sum_{k = 0}^{2^n-1}\wt \vartheta^2_{n,k}} = 1 - \lim_{n \ua \infty} \frac{1}{n}\log_2 \sqrt{\frac{m_{n+2}}{2^{2n+5}}\sum_{k = 0}^{2^n-1}\wt \vartheta^2_{n,k}} + \lim_{n \ua \infty}\frac{1}{2n}\log_2\frac{m_{n+2}}{2^{2n+5}}\\&= 1 - \lim_{n \ua \infty}\frac{1}{n}\log_2\sqrt{\int_0^1 \sigma_s^4\,ds} + (\alpha - 1) = \alpha. 
		\end{split}
	\end{equation*}
	This completes the proof.
\end{proof}

\noindent \textbf{Acknowledgement.} The authors gratefully acknowledge support from the Natural Sciences and Engineering Research Council of Canada through grants RGPIN-2017-04054  and RGPIN-2024-03761. The authors would like to thank anonymous reviewers and the associate editor for helpful remarks that significantly improved this paper. We also wish to express our gratitude to Yuqi Jing for the help in setting up the server that supported our simulation studies.

\bibliography{CTbook}{}
\bibliographystyle{plain}

\end{document}